\documentclass[twocolumn,reprint,longbibliography]{revtex4-2}
\pdfoutput=1
\usepackage[english]{babel}
\usepackage[T1]{fontenc}
\usepackage{amsmath}
\usepackage{amsfonts}
\usepackage{amsthm}
\usepackage{comment}
\usepackage{array}
\usepackage{graphics}
\usepackage{mathtools}
\usepackage{enumitem}
\usepackage{ragged2e}
\usepackage[noabbrev]{cleveref}
\usepackage[group-separator={,},group-minimum-digits={3}]{siunitx}

\theoremstyle{plain}
\newtheorem{theorem}{Theorem}[section]

\theoremstyle{definition}

\newcolumntype{C}[1]{>{\centering\arraybackslash}m{#1}}
\newcolumntype{L}[1]{>{\raggedright\arraybackslash}m{#1}}
\newcolumntype{R}[1]{>{\raggedleft\arraybackslash}m{#1}}
\newcolumntype{J}[1]{>{\arraybackslash}m{#1}}
\newcommand{\saveparinfo}{\xdef\savedindent{\the\parindent}\xdef\savedparskip{\the\parskip}}
\newcommand{\useparinfo}{\setlength{\parindent}{\savedindent}\setlength{\parskip}{\savedparskip}\justifying}

\DeclareMathOperator{\tr}{tr}
\DeclareMathOperator{\ad}{ad}
\DeclareMathOperator*{\esssup}{ess\,sup}
\DeclareMathOperator*{\essinf}{ess\,inf}

\newcommand{\Prb}{\mathbb{P}}
\newcommand{\Exp}{\mathbb{E}}
\newcommand{\Var}{\mathbb{V}}
\newcommand{\Cum}{\mathbb{M}}

\newcommand{\bfrac}[2]{\frac{\displaystyle#1}{\displaystyle#2}}

\begin{document}
\saveparinfo

\title{The Relative Entropy of Expectation and Price}
\author{Paul McCloud}
%\affiliation{Nomura}
\date{\today}
%\orcid{0000-0001-9531-5045}
%\email{paul.mccloud@yahoo.co.uk}
\begin{abstract}
As operators acting on the undetermined final settlement of a derivative security, expectation is linear but price is non-linear. When the market of underlying securities is incomplete, non-linearity emerges from the bid-offer around the mid price that accounts for the residual risks of the optimal funding and hedging strategy. At the extremes, non-linearity also arises from the embedded options on capital that are exercised upon default. In this essay, these convexities are quantified in an entropic risk metric that evaluates the strategic risks, which is realised as a cost with the introduction of bilateral margin. Price is then adjusted for market incompleteness and the risk of default caused by the exhaustion of capital.

In the complete market theory, price is derived from a martingale condition. In the incomplete market theory presented here, price is instead derived from a log-martingale condition:
\begin{equation}
p=-\frac{1}{\alpha}\log\Exp\exp[-\alpha P] \notag
\end{equation}
for the price $p$ and payoff $P$ of a funded and hedged derivative security, where the price measure $\Exp$ has minimum entropy relative to economic expectations, and the parameter $\alpha$ matches the risk aversion of the investor. This price principle is easily applied to standard models for market evolution, with applications considered here in model risk analysis, data-driven hedging and quantum information.
\end{abstract}
\maketitle

\section{Introduction}

A derivative transaction exchanges the execution price $p$ at initial time $t$ for a commitment to return the contractual price $p+dp$ at final time $t+dt$. In determining its economic viability, the counterparties assess whether the trade enhances or degrades the overall performance of their investment portfolios, balancing expected return against net risk. Asymmetry between the expectations and investment objectives of participants then drives market evolution, following the dynamic that trades execute when parties agree on price but disagree on value.

Capital is deployed to finance the transaction, entailing economic consequences for the investor beyond the contractual terms of the derivative. In addition to the expected final settlements, valuation depends on the funding of initial and variation margin and the level of default protection they confer. The impact of margin is typically developed as a valuation adjustment to a price model derived from the hedging of contractual settlements. In this essay, an entropic margin model is instead placed at the centre of the pricing methodology, providing the risk counterbalance to expected return in the investment strategy. There are numerous benefits to this approach.
\begin{description}[leftmargin=0\parindent]
\item[Investment strategy]Traditional strategies balance expected return against risk. Here, the risk metric is encoded in the entropic margin model as an additional settlement. Maximising the return on capital net of all settlements naturally embeds risk management.
\item[Incomplete markets]Market equilibrium identifies a unique price model with minimum entropy relative to economic expectations, calibrated to available funding and hedging. In the process, this resolves the ambiguity of pricing in incomplete markets.
\item[Information model]Entropic risk optimisation is applicable wherever the concept of entropy is defined. This includes theoretical or data-driven models with classical or quantum information. Entropic pricing unlocks the novel algorithmic potential of these frameworks.
\item[Market and model risk]Uniquely determining price from expectation, market and model risks are distinguished respectively as sensitivities to observed prices and subjective economic assumptions. This enables a hierarchical approach that recognises market tiering.
\item[Decentralised finance]The level of margin is controlled by a single parameter, and the link with funding settlements is transparent and programmatic. Implemented in the scripting of smart contracts for derivatives, this supports the transition to decentralisation.
\end{description}

Advocates of deep hedging replace risk neutralisation with risk optimisation as the guiding principle for pricing, as data-driven models see the complex correlations among economic variables as opportunity rather than hindrance. Risk optimisation is similarly used here to derive the price measure from the expectation measure, though the objectives differ. Model risk is quantified as the residual risk after funding and hedging, which requires for its definition a precise understanding of the relationship between expectation and price. Taking this idea to its logical conclusion, the convex risk metric is implemented as margin held against model risk. Optimisation of the strategic return net of funding and hedging then adjusts price to account for the risk management benefits of margin.

\begin{figure*}[!p]
\setlength{\abovecaptionskip}{5pt}
\setlength{\belowcaptionskip}{20pt}
\centering
\includegraphics[width=0.8\textwidth]{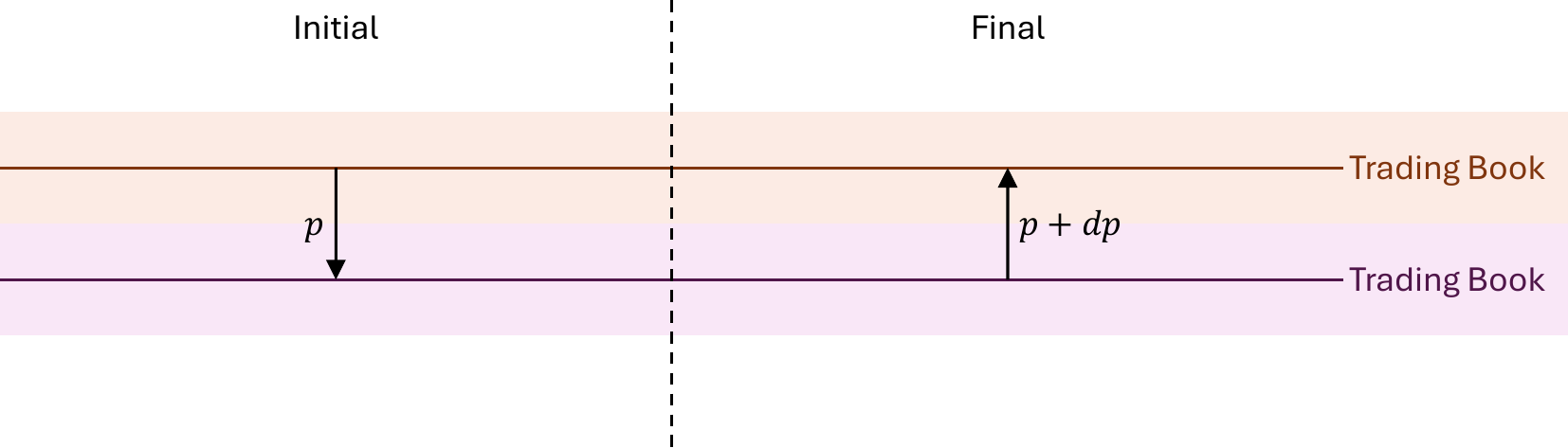}
\caption{The contractual settlements of the trade.}
\label{fig:cashflows1}
\includegraphics[width=0.8\textwidth]{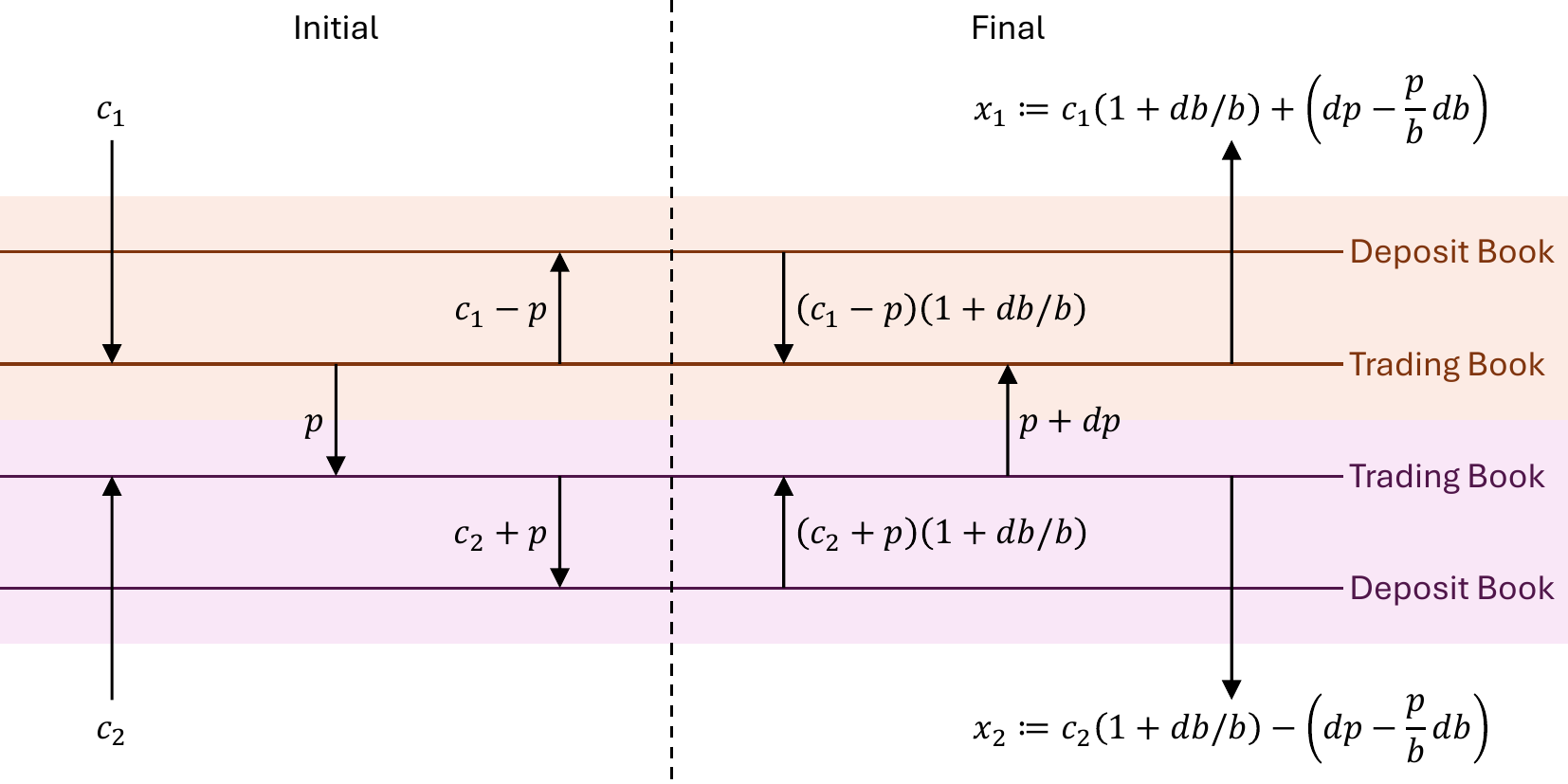}
\caption{The contractual and funding settlements of the trade.}
\label{fig:cashflows2}
\includegraphics[width=0.8\textwidth]{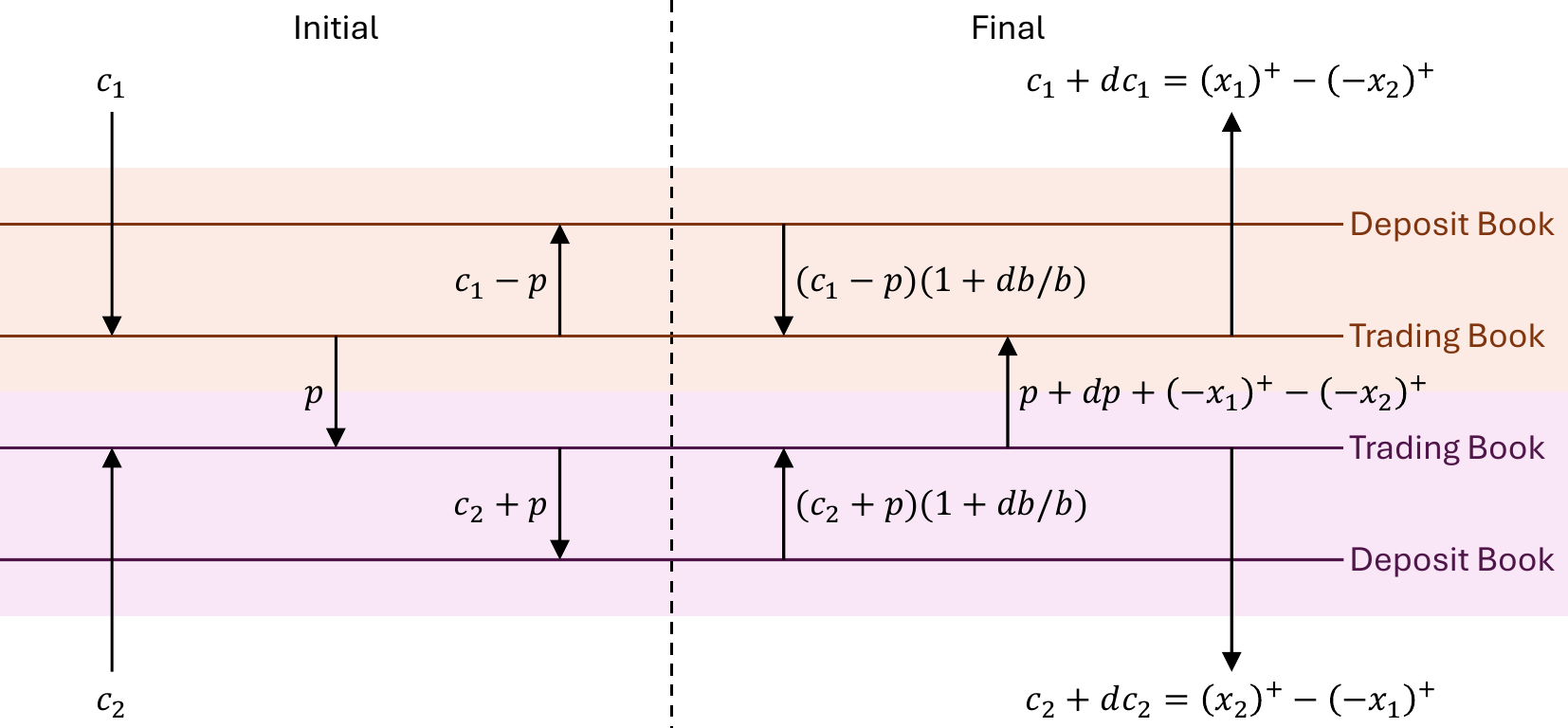}
\caption{The contractual, funding and default settlements of the trade.}
\label{fig:cashflows3}
\end{figure*}

Suppose that two counterparties deploy initial capital $c_1>0$ and $c_2>0$ to the transaction and deposit the proceeds in a liquid margin account with positive unit price $b>0$. The first counterparty buys the derivative from the second counterparty, and the capital must be sufficient to enter the transaction:
\begin{equation}
-c_2\le p\le c_1
\end{equation}
The capital $c_1$ deployed by the first counterparty is partitioned into the variation margin $p$ and the initial margin $m_1:=c_1-p$, and the capital $c_2$ deployed by the second counterparty is partitioned into the variation margin $-p$ and the initial margin $m_2:=c_2+p$. With the initial margin deposited in the margin account, the final margin available to settle the derivative is $m_1(1+db/b)$ for the first counterparty and $m_2(1+db/b)$ for the second counterparty. This is sufficient to settle the contractual commitments of the derivative when:
\begin{equation}
-m_1(1+db/b)\le p+dp\le m_2(1+db/b)
\end{equation}
Define the default trigger levels:
\begin{align}
x_1:={}&m_1(1+db/b)+(p+dp) \\
={}&c_1(1+db/b)+\left(dp-\frac{p}{b}db\right) \notag \\
x_2:={}&m_2(1+db/b)-(p+dp) \notag \\
={}&c_2(1+db/b)-\left(dp-\frac{p}{b}db\right) \notag
\end{align}
The first counterparty defaults when $x_1<0$ and the second counterparty defaults when $x_2<0$, and these two conditions are mutually exclusive. The actual final settlement of the derivative, net of default, is then:
\begin{equation}
p+dp+(-x_1)^+-(-x_2)^+
\end{equation}
and the final capital of the counterparties is:
\begin{align}
c_1+dc_1&=(x_1)^+-(-x_2)^+ \\
c_2+dc_2&=(x_2)^+-(-x_1)^+ \notag
\end{align}
respectively. These expressions for the change in capital are rearranged to:
\begin{align}
d\frac{c_1}{b}&=\left(-\frac{c_1}{b}-d\frac{p}{b}\right)^{\!+}-\left(d\frac{p}{b}-\frac{c_2}{b}\right)^{\!+}+d\frac{p}{b} \\
d\frac{c_2}{b}&=\left(d\frac{p}{b}-\frac{c_2}{b}\right)^{\!+}-\left(-\frac{c_1}{b}-d\frac{p}{b}\right)^{\!+}-d\frac{p}{b} \notag
\end{align}
demonstrating that the discounted return on capital matches the discounted return on the derivative adjusted by a put option on $d(p/b)$ with strike $-c_1/b$, exercised when the first counterparty defaults, and a call option on $d(p/b)$ with strike $c_2/b$, exercised when the second counterparty defaults. When assessing the viability of the trade, it is these net returns that must be considered, accounting for funding costs and the possibility that the settling counterparty defaults on its commitments.

Fixing a target for the investment strategy, the expected performance is quantified by the entropy-adjusted mean return net of funding and default:
\begin{equation}
\Exp[\alpha][dc]:=-\frac{1}{\alpha}\log\Exp\exp[-\alpha\,dc]
\end{equation}
where the measure $\Exp$ encapsulates the expectations of the investor and the parameter $\alpha>0$ controls their risk aversion. In this strategy, the trade is considered to be viable when the entropy-adjusted mean is positive. Accounting for the expectations and risk appetites of both counterparties, the trade is viable when:
\begin{equation}
\Exp_2[-\alpha_2][-dc_2]\le0\le\Exp_1[\alpha_1][dc_1]
\end{equation}
Setting both sides of this window of viability to zero imposes two fair pricing conditions on the initial and variation margins. The remaining degree of freedom controls the level of default protection embedded in the trade.

This simple framework for bilateral margin is used in the following to derive equilibrium pricing expressions that account for funding and default. It is an idealisation of real trading activity, and does not account for the default protections provided by additional capital covering the net risks of the institution. Extensions to multilateral margin, which can be incorporated into the framework, are not considered here.

In this essay, entropic foundations for pricing are investigated that do not rely on unrealistic assumptions regarding market completeness or the effectiveness of the trading strategy, properly account for residual model risk, and establish consistency across diverse applications including algorithmic trading, derivative pricing and hedging, and the modelling of margin and capital. In developing these foundations, the entropic risk metric has dual application: as the measure of performance used to optimise the investment strategy; and as the determinant of margin used to protect against default. These results leverage the remarkable versatility and tractability of relative entropy as a universal metric of uncertainty, applicable across a wide range of models and algorithms.

\section{Entropic investment strategies}

\begin{figure*}[!t]
\centering
\includegraphics[width=0.45\textwidth]{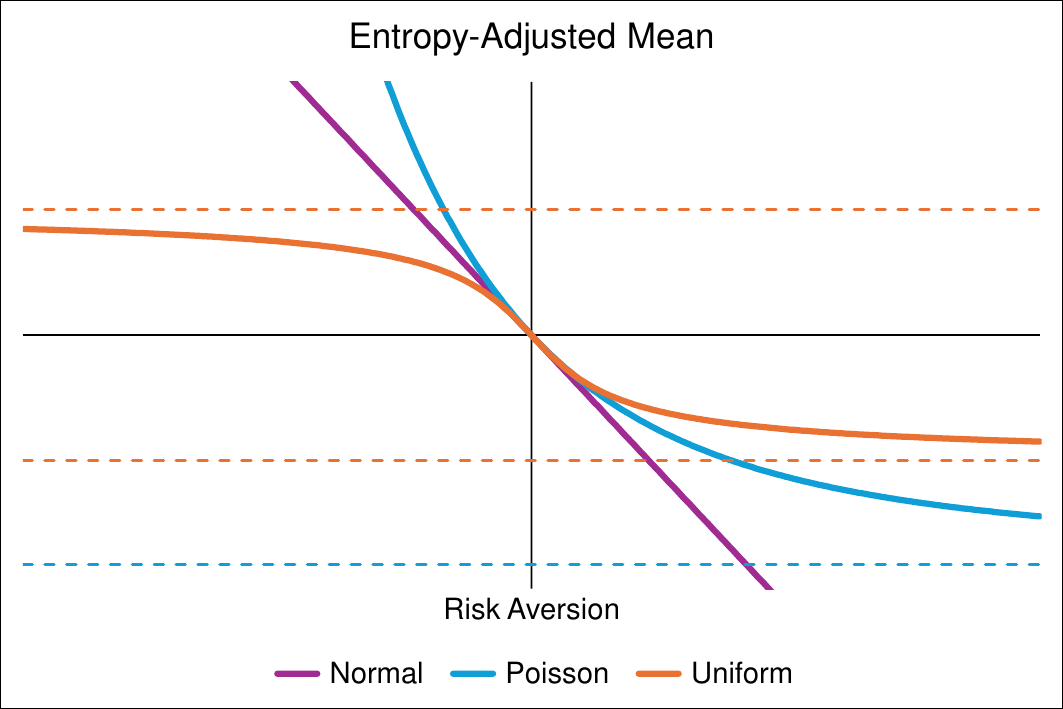}
\caption{Entropy-adjusted mean as a function of risk aversion for three random variables. In each case, the variable is standardised by linear transformation to have zero mean and unit standard deviation. The entropy-adjusted mean is monotonic decreasing; it is unbounded for the normal variable, bounded below for the Poisson variable with ten expected jumps, and bounded above and below for the uniform variable.}
\label{fig:entropyadjustedmean}
\end{figure*}

The objective of the investor is to optimise the return $dc$ on the investment strategy, a task that depends on an initial assessment of the final economy that is both subjective and indeterminate. This assessment is encapsulated in the expectation measure $\Exp$ of the investor, which maps the final variable $X$ to its initial expectation $\Exp[X]$. Use the following notation for the associated operators of variance, covariance and probability:
\begin{align}
\Var[X]&:=\Exp[XX^t]-\Exp[X]\Exp[X]^t \\
\Var[X,Y]&:=\Exp[XY^t]-\Exp[X]\Exp[Y]^t \notag \\
\Prb[\mathcal{A}]&:=\Exp[1_\mathcal{A}] \notag
\end{align}
for the random vectors $X$ and $Y$ and the event $\mathcal{A}$, where $1_\mathcal{A}$ is the indicator of the event. Investment performance is estimated using these operators, and the strategy is optimised accordingly.

\subsection{Entropic risk metric}

Maximising the mean $\Exp[dc]$ is an obvious strategic objective, though taken on its own this incentivises reckless investment and fails to reward effective risk management. Mean-variance optimisation extends the strategy by maximising the variance-adjusted mean:
\begin{equation}
\Exp[dc]-\frac{1}{2}\alpha\Var[dc]
\end{equation}
In this expression, the parameter $\alpha$ controls the risk aversion, and is interpreted as the Lagrange multiplier for the variance constraint in a strategy that maximises the mean for a fixed variance. Incorporating higher moments, entropic-risk optimisation further extends the strategy by maximising the entropy-adjusted mean:
\begin{align}
\Exp[\alpha][dc]:={}&\!-\frac{1}{\alpha}\log\Exp\exp[-\alpha\,dc] \\
={}&\Exp[dc]-\frac{1}{2}\alpha\Var[dc]+\sum_{n=3}^\infty\frac{(-\alpha)^{n-1}}{n!}\Cum_n[dc] \notag
\end{align}
where $\Cum_n$ is the $n$th cumulant operator of the investor, with $\Cum_0=0$, $\Cum_1=\Exp$ and $\Cum_2=\Var$. The change of measure implicit in this definition reweights according to the sign and magnitude of the return, with properties that make it a useful metric of value-at-risk.

\begin{theorem}[Entropic risk metric]Let $dc$ be a random variable whose cumulant generating function in the probability measure $\Exp$ is defined. The entropy-adjusted mean:
\begin{equation}
\Exp[\alpha][dc]:=-\frac{1}{\alpha}\log\Exp\exp[-\alpha\,dc]
\end{equation}
satisfies the linearity property:
\begin{equation}
\Exp[\alpha][\lambda\,dc+\mu]=\lambda\Exp[\lambda\alpha][dc]+\mu
\end{equation}
for the scalars $\lambda$ and $\mu$ and is monotonic decreasing as a function of the adjustment parameter $\alpha$ with limits:
\begin{align}
\Exp[-\infty][dc]&=\esssup[dc] \\
\Exp[0][dc]&=\Exp[dc] \notag \\
\Exp[\infty][dc]&=\essinf[dc] \notag
\end{align}
Let $\bar{\Exp}$ be an equivalent probability measure with strictly positive Radon-Nikodym weight $d\bar{\Exp}/d\Exp$. Define the entropy of $\bar{\Exp}$ relative to $\Exp$ by:
\begin{equation}
S[\bar{\Exp}|\Exp]:=\Exp\!\left[\frac{d\bar{\Exp}}{d\Exp}\log\!\left[\frac{d\bar{\Exp}}{d\Exp}\right]\right]
\end{equation}
Relative entropy is a metric of distance between equivalent probability measures: it is positive, $S[\bar{\Exp}|\Exp]\ge0$, and is zero when $\bar{\Exp}=\Exp$. Let $\Exp^\alpha$ be the equivalent probability measure with Radon-Nikodym weight $d\Exp^\alpha/d\Exp$ given by:
\begin{equation}
\bfrac{d\Exp^\alpha}{d\Exp\hfill}:=\frac{\exp[-\alpha\,dc]}{\Exp[\exp[-\alpha\,dc]]}
\end{equation}
The entropy-adjusted mean $\Exp[\alpha][dc]$ equals the mean $\bar{\Exp}[dc]$ adjusted by relative entropy:
\begin{equation}
\Exp[\alpha][dc]=\bar{\Exp}[dc]+\frac{1}{\alpha}(S[\bar{\Exp}|\Exp]-S[\bar{\Exp}|\Exp^\alpha])
\end{equation}
\end{theorem}
\begin{proof}
Linearity follows directly from the definition of the entropy-adjusted mean. Without loss of generality, assume the adjustment parameter is positive and use linearity to derive the corresponding results for the case of negative adjustment parameter.

These are standard results from information theory, and only the main elements of their proofs are presented here. Monotonicity is a consequence of H\"{o}lder's inequality in the form:
\begin{equation}
\Exp[\exp[-\alpha_1dc]]^{1/\alpha_1}\le\Exp[\exp[-\alpha_2dc]]^{1/\alpha_2}
\end{equation}
for $0<\alpha_1<\alpha_2$, so that $\Exp[\alpha_1][dc]\ge\Exp[\alpha_2][dc]$. The limit for $\alpha\to0$ follows from the expansion of the cumulant generating function. For the limit $\alpha\to\infty$, Markov's inequality:
\begin{equation}
\Prb[dc\le\gamma]\le\frac{\Exp[\exp[-\alpha\,dc]]}{\exp[-\alpha\gamma]}
\end{equation}
for the parameters $\alpha>0$ and $\gamma>\essinf[dc]$ (so that $\Prb[dc\le\gamma]>0$) is rearranged to:
\begin{equation}
\Exp[\alpha][dc]\le\gamma-\frac{1}{\alpha}\log[\Prb[dc\le\gamma]]
\end{equation}
From this, taking the limit $\alpha\to\infty$ followed by taking the limit $\gamma\to\essinf[dc]$ derives $\Exp[\infty][dc]\le\essinf[dc]$. The other side of the target inequality follows from the almost certainty of $\exp[-\alpha\,dc]\le\exp[-\alpha\essinf[dc]]$, so that $\Exp[\alpha][dc]\ge\essinf[dc]$.

Gibb's theorem states that $S[\bar{\Exp}|\Exp]\ge0$, which follows from the observation that $x\log[x]\ge x-1$ for $x>0$. The remaining result is a presentation of Donsker and Varadhan's variational formula. Let $\bar{m}:=d\bar{\Exp}/d\Exp$ and $m_\alpha:=d\Exp^\alpha/d\Exp$ be the Radon-Nikodym weights. Then:
\begin{align}
S[\bar{\Exp}&|\Exp]-S[\bar{\Exp}|\Exp^\alpha] \\
&=\Exp[\bar{m}\log[\bar{m}]]-\Exp^\alpha[(\bar{m}/m_\alpha)\log[\bar{m}/m_\alpha]] \notag \\
&=\bar{\Exp}[\log[\bar{m}]]-(\bar{\Exp}[\log[\bar{m}]]-\bar{\Exp}[\log[m_\alpha]]) \notag \\
&=\bar{\Exp}[-\alpha\,dc-\log\Exp\exp[-\alpha\,dc]] \notag \\
&=-\alpha(\bar{\Exp}[dc]-\Exp[\alpha][dc]) \notag
\end{align}
relating the entropy-adjusted mean to the relative entropies of the measures.
\end{proof}

Instances of the final equality for $\bar{\Exp}=\Exp$ and $\bar{\Exp}=\Exp^\alpha$ generate bounds for the entropy-adjusted mean:
\begin{equation}
\begin{multlined}[t][0.85\displaywidth]
\Exp^\alpha[dc]\le \\
\makebox[0pt]{$\displaystyle\frac{1}{\alpha}S[\Exp^\alpha|\Exp]+\Exp^\alpha[dc]=\Exp[\alpha][dc]=\Exp[dc]-\frac{1}{\alpha}S[\Exp|\Exp^\alpha]$}\hspace{-0.6cm} \\
\le\Exp[dc]\hspace{-0.6cm}
\end{multlined}
\end{equation}
in the case $\alpha>0$, with the inequalities reversed in the case $\alpha<0$. The lever $\alpha$ dials the entropy-adjusted mean $\Exp[\alpha][dc]$ across the span of possible values for $dc$ between $\essinf[dc]$ and $\esssup[dc]$ and no further. Entropic risk is, in this sense, more targeted than variance adjustment as a moments-based metric of value-at-risk. The equivalent measure $\Exp^\alpha$ plays a central role in this result, and the gap between the mean $\Exp[dc]$ and the entropy-adjusted mean $\Exp[\alpha][dc]$ is proportional to the distance between the measure $\Exp$ and the measure $\Exp^\alpha$, as defined by their relative entropy.

\subsection{Optimising entropic risk}

The investor has a choice of self-funded market strategies whose returns over the investment period are the components of the return vector $dq$. The return on the investment portfolio whose weights are the components of the weight vector $\omega$ is then $dc:=\omega\cdot dq$, the weighted sum of the returns on the market strategies.

Mean-variance and entropic-risk optimisation are the basis for investment strategies, leading in each case to equilibrium models of pricing and hedging. There are many similarities between these strategies, but also some critical differences. Mean-variance optimisation does not account for non-Gaussian features in the return, and the strategy may need to be moderated to avoid arbitrage. In contrast, entropic-risk optimisation incorporates all the cumulants of the return, and naturally leads to arbitrage-free equilibrium pricing.

Introduce the equivalent measure $\Exp^\phi$ parametrised by the vector $\phi$ with Radon-Nikodym weight:
\begin{equation}
\bfrac{d\Exp^\phi}{d\Exp\hfill}:=\frac{\exp[-\phi\cdot dq]}{\Exp[\exp[-\phi\cdot dq]}
\end{equation}
Price measures that are equivalent to the expectation measure are identified within this family from the equilibrium of the entropic-risk optimisation strategy. The parameter $\phi$ arises from the change of measure in the definition of the entropic risk metric, which depends on the existence of the cumulant generating function. If this does not exist, the quadratic-regularised measure is used instead, with Radon-Nikodym weight:
\begin{equation}
\bfrac{d\Exp^\phi}{d\Exp\hfill}:=\frac{\exp[-\phi\cdot dq-\frac{1}{2}dq\cdot\zeta\,dq]}{\Exp[\exp[-\phi\cdot dq-\frac{1}{2}dq\cdot\zeta\,dq]]}
\end{equation}
regularised by the positive-definite matrix $\zeta$. Regularisation is removed in the limit $\zeta\to0$; if the strategy does not survive this limit, then regularisation, which moderates the impact of large events, is retained as a parameter of the model. In the following, the cumulant generating function is assumed to exist, implying the use of quadratic regularisation wherever necessary.

In the entropic-risk optimisation strategy, the investor seeks the optimal portfolio $\omega$ of market strategies that maximises the entropy-adjusted mean:
\begin{equation}
\Exp[\alpha][\omega\cdot dq]:=-\frac{1}{\alpha}\log\Exp\exp[-\alpha\omega\cdot dq]
\end{equation}
Risk aversion is controlled by the parameter $\alpha>0$, and entropic risk is monotonic in this range:
\begin{equation}
\begin{multlined}[t][0.85\displaywidth]
\essinf[\omega\cdot dq]= \\
\Exp[\infty][\omega\cdot dq]\le\Exp[\alpha][\omega\cdot dq]\le\Exp[0][\omega\cdot dq] \\
=\omega\cdot\Exp[dq]
\end{multlined}
\end{equation}
Large values of $\alpha$ thus imply that the investor takes more account of negative outcomes in the strategic objective, corresponding to a lower risk appetite. The optimal strategy is the stationary point of entropic risk $\omega=(1/\alpha)\phi$ where the unit optimal portfolio $\phi$ satisfies the calibration condition:
\begin{equation}
0=\Exp^\phi[dq]
\end{equation}
As a practical observation, the Newton-Raphson scheme:
\begin{equation}
\partial\phi=\Var^\phi[dq]^{-1}\Exp^\phi[dq]
\end{equation}
initialised with $\phi=0$ calibrates the unit optimal portfolio to the calibration condition when the variance is invertible. This portfolio is a property of the market return, generating the price measure $\Exp^\phi$ equivalent to the expectation measure $\Exp$. Applying the entropic risk metric theorem for the case $dc=\phi\cdot dq$ and $\alpha=1$ leads to the result:
\begin{align}
S[\bar{\Exp}|\Exp]&=S[\bar{\Exp}|\Exp^\phi]+S[\Exp^\phi|\Exp]-\phi\cdot(\bar{\Exp}[dq]-\Exp^\phi[dq]) \\
&\ge S[\Exp^\phi|\Exp]-\phi\cdot(\bar{\Exp}[dq]-\Exp^\phi[dq]) \notag
\end{align}
Among the equivalent measures $\bar{\Exp}$ satisfying the calibration condition $\bar{\Exp}[dq]=0$, the price measure $\Exp^\phi$ thus has the lowest entropy relative to the expectation measure $\Exp$.

The optimal strategy is constructed on the assumption that the investor has expertise in all the available markets. In practice, markets are tiered and investors specialise in specific market sectors, and the optimal strategy across all sectors emerges as the net strategy of multiple specialist investors. This arrangement is conveniently captured in the entropic risk metric.

Suppose that the market return vector decomposes as a direct sum $dq=dq_1\oplus\cdots\oplus dq_n$ where $dq_i$ is the return vector for the $i$th market tier. The optimal portfolio is similarly decomposed as a direct sum $\phi=\phi_1\oplus\cdots\oplus\phi_n$ where $\phi_i$ is the portfolio in the $i$th market tier, and is solved from the calibration conditions:
\begin{equation}
0=\Exp^{\phi_1\oplus\cdots\oplus\phi_n}[dq_1]=\cdots=\Exp^{\phi_1\oplus\cdots\oplus\phi_n}[dq_n]
\end{equation}
This provides $n$ conditions for the $n$ components of the optimal portfolio. While these calibration conditions need to be solved simultaneously, market tiering is exploited to incrementally build the optimal portfolio from the independent contributions of sector investors, localising at each tier the risk management activity for the corresponding market expert. In this perspective, there are $n(n+1)/2$ calibration conditions for $n$ sub-markets:
\begin{align}
0&=\Exp^{\phi_1^1}[dq_1] \\
0&=\Exp^{\phi_1^2\oplus\phi_2^2}[dq_1]=\Exp^{\phi_1^2\oplus\phi_2^2}[dq_2] \notag \\
&\makebox[\widthof{$=$}][r]{$\vdots$} \notag \\
0&=\Exp^{\phi_1^n\oplus\cdots\oplus\phi_n^n}[dq_1]=\cdots=\Exp^{\phi_1^n\oplus\cdots\oplus\phi_n^n}[dq_n] \notag
\end{align}
where the calibration conditions in the $i$th row maximise the entropy-adjusted mean of the $i$th optimal portfolio $\phi^i=\phi_1^i\oplus\cdots\oplus\phi_i^i$ constituted in the lowest $i$ market tiers only. The $i$th investor then takes responsibility for determining the primary component $\psi_i$ of the $i$th optimal portfolio and its delta $\delta^i=\delta_1^i\oplus\cdots\oplus\delta_{i-1}^i$ to the $(i-1)$th optimal portfolio in the underlying sub-market, using the invertible linear change of variables:
\begin{equation}\arraycolsep=1.4pt\def\arraystretch{1.4}\abovedisplayskip=0.6\abovedisplayskip
\begin{array}{lr}
\begin{array}{rl}
\psi_i&=\phi_i^i \\
\delta_j^i&=\phi_j^{i-1}-\phi_j^i
\end{array} &
\qquad\phi_j^i=\psi_j-(\delta_j^{j+1}+\cdots+\delta_j^i)
\end{array}
\end{equation}
Stacking these deltas for $i=1,\ldots,n$ generates the optimal portfolio of the full market.

For each $i=0,\ldots,n$, define the equivalent price measure $\Exp^i$ with Radon-Nikodym weight:
\begin{equation}
\bfrac{d\Exp^i}{d\Exp\hfill}:=\frac{\exp[-(\phi_1^i\cdot dq_1+\cdots+\phi_i^i\cdot dq_i)]}{\Exp[\exp[-(\phi_1^i\cdot dq_1+\cdots+\phi_i^i\cdot dq_i)]]}
\end{equation}
and define the $i$th entropy-adjusted mean:
\begin{equation}
\mu^i:=-\log\Exp\exp[-(\phi_1^i\cdot dq_1+\cdots+\phi_i^i\cdot dq_i)]
\end{equation}
At the $i$th tier, the portfolio $\phi^i$ maximises the entropy-adjusted mean $\mu^i$ when the primary portfolio $\psi_i$ and the hedge portfolio $\delta_1^i\oplus\cdots\oplus\delta_{i-1}^i$ are calibrated to the hedge conditions:
\begin{align}
&0= \\
&\Exp^{i-1}[dq_j\exp[-(\psi_i\cdot dq_i-(\delta_1^i\cdot dq_1+\cdots+\delta_{i-1}^i\cdot dq_{i-1}))]] \notag
\end{align}
for $j=1,\ldots,i$. Starting with the expectation measure $\Exp^0=\Exp$, these calibrations are iterated to obtain the unit optimal portfolio $\phi^n$ in the full market, with its associated price measure $\Exp^n$ calibrated to the calibration conditions for the market strategies.

Since expanding the market can only increase the opportunities for the investor, the maximum entropy-adjusted mean increases, $\mu^i\ge\mu^{i-1}$, as new tiers are introduced. Furthermore, the $i$th tier introduces no additional advantage for the strategy if the optimal portfolios $\phi^i$ and $\phi^{i-1}$ satisfy the indifference condition $\mu^i=\mu^{i-1}$. This is adopted as the equilibrium condition for pricing derivative strategies that have no intrinsic value beyond their contractual link with underlying strategies.

\section{Entropic pricing and hedging}

\begin{figure*}[!p]
\centering\normalsize\setlength{\tabcolsep}{0.5\columnsep}
\begin{tabular}{@{}C{\textwidth}@{}}
Entropic Pricing and Hedging \\[0.5ex] \hline
\\[-1ex]
{\begin{minipage}[t]{\textwidth}\useparinfo
Fix a finite investment horizon and let $\Exp$ be the expectation measure mapping variables at the final time $T=t+dt$ to their expectations at the initial time $t$. Assume the market is tiered into a set of underlying securities with initial price vector $q$ and final price vector $Q=q+dq$ and a derivative security with initial price $p$ and final price $P=p+dp$. Consider two self-funded strategies: one comprising underlying securities only; and one that also includes the derivative security.
\end{minipage}} \\\\[-1ex]
{\begin{minipage}[t]{0.5\textwidth-\tabcolsep}\useparinfo
{\bf Underlying only:} In this strategy, the investment comprises an underlying portfolio $\omega$ satisfying the initial condition:
\begin{equation}
\omega\cdot q=0
\end{equation}
so that the investment is self-funded. The return on the investment has entropy-adjusted mean:
\begin{align}
\Exp[\alpha]&[\omega\cdot dq]= \\
&-\frac{1}{\alpha}\log\Exp\exp[-\alpha\omega\cdot dq] \notag
\end{align}
at the risk aversion $\alpha>0$. The optimal portfolio $\omega$ that maximises this satisfies the {\bf calibration condition}:
\begin{equation}
\Exp[(dq-qr\,dt)\exp[-\alpha\omega\cdot dq]]=0
\end{equation}
where $r\,dt$ is the Lagrange multiplier used to implement the self-funding condition.
\end{minipage}}
\hfill\vline\hfill
{\begin{minipage}[t]{0.5\textwidth-\tabcolsep}\useparinfo
{\bf Underlying and derivative:} In this strategy, the investment comprises the derivative security and an underlying portfolio $(\omega-\delta)$ satisfying the initial condition:
\begin{equation}
p+(\omega-\delta)\cdot q=0
\end{equation}
so that the investment is self-funded. The return on the investment has entropy-adjusted mean:
\begin{align}
\Exp[\alpha]&[dp+(\omega-\delta)\cdot dq]= \\
&-\frac{1}{\alpha}\log\Exp\exp[-\alpha(dp+(\omega-\delta)\cdot dq)] \notag
\end{align}
at the risk aversion $\alpha>0$. The hedge portfolio $\delta$ that maximises this satisfies the {\bf hedge condition}:
\begin{equation}
\Exp[(dq-q(r+\alpha s)\,dt)\exp[-\alpha(dp+(\omega-\delta)\cdot dq)]]=0
\end{equation}
where $(r+\alpha s)\,dt$ is the Lagrange multiplier used to implement the self-funding condition.
\end{minipage}} \\\\[-1ex]
{\begin{minipage}[t]{\textwidth}\useparinfo
These two strategies offer different entropy-adjusted mean returns to the investor. When the derivative security has no intrinsic value beyond its contractual link with the underlying securities, the efficient market equilibrates to the indifferent state with matching entropy-adjusted means, leading to the {\bf price condition}:
\begin{equation}
\Exp[\alpha][\omega\cdot dq]=\Exp[\alpha][dp+(\omega-\delta)\cdot dq]
\end{equation}
The calibration condition is solved to derive the optimal portfolio $\omega$ of underlying securities, and the hedge condition is then solved to derive the hedge portfolio $\delta$ of underlying securities that re-establishes optimality for the combined portfolio. With these portfolios, the fair initial price $p$ of the derivative security is derived from the price condition.
\end{minipage}} \\\\[-1ex]
{\begin{minipage}[t]{0.5\textwidth-\tabcolsep}\useparinfo
Define the price measure $\Exp^\phi$ equivalent to the expectation measure $\Exp$ with Radon-Nikodym weight:
\begin{equation}
\bfrac{d\Exp^\phi}{d\Exp^{\vphantom{\phi}}\hfill}:=\bfrac{\exp[-\phi\cdot dq]}{\Exp^{\vphantom{\phi}}[\exp[-\phi\cdot dq]]}
\end{equation}
While there is no explicit expression for the unit optimal portfolio $\phi:=\alpha\omega$ that solves the calibration condition, it is efficiently discovered by the Newton-Raphson scheme initialised with $\phi=0$ and repeating the step:
\begin{align}
r\,dt:={}& \\
&\hspace{-1cm}\bfrac{(\phi+\Var^\phi[dq]^{-1}\Exp^\phi[dq])\cdot q}{\Var^\phi[dq]^{-1}q\cdot q} \notag \\
\partial\phi={}&\Var^\phi[dq]^{-1}(\Exp^\phi[dq]-qr\,dt) \notag
\end{align}
until the required accuracy is achieved.
\end{minipage}}
\hfill\vline\hfill
{\begin{minipage}[t]{0.5\textwidth-\tabcolsep}\useparinfo
Define the price measure $\Exp^{\phi\delta}:=\Exp^{(\phi-\alpha\delta)\oplus\alpha}$ equivalent to the price measure $\Exp^\phi$ with Radon-Nikodym weight:
\begin{equation}
\bfrac{d\Exp^{\phi\delta}}{d\Exp^\phi\hfill}:=\bfrac{\exp[-\alpha(P-\delta\cdot dq)]}{\Exp^\phi[\exp[-\alpha(P-\delta\cdot dq)]]}
\end{equation}
While there is no explicit expression for the hedge portfolio $\delta$ that solves the hedge and price conditions, it is efficiently discovered by the Newton-Raphson scheme initialised with $\delta=0$ and repeating the step:
\begin{align}
(r+\alpha s)\,dt:={}& \\
&\hspace{-1cm}\bfrac{\alpha\Exp^\phi[\alpha][P-\delta\cdot Q]+\Var^{\phi\delta}[dq]^{-1}\Exp^{\phi\delta}[dq]\cdot\Exp^{\phi\delta}[Q]}{\Var^{\phi\delta}[dq]^{-1}q\cdot\Exp^{\phi\delta}[Q]} \notag \\
\partial\delta={}&-\Var^{\phi\delta}[dq]^{-1}(\Exp^{\phi\delta}[dq]-q(r+\alpha s)\,dt)/\alpha \notag
\end{align}
until the required accuracy is achieved.
\end{minipage}} \\
\\[-1ex] \hline
\end{tabular}
\caption{In entropic pricing, economic principles are established that determine the optimal portfolio of underlying securities and the hedge portfolio and price of the derivative security, based on the maximisation of the entropy-adjusted mean return.}
\label{fig:entropicpricingprinciples}
\vspace{100cm}
\end{figure*}

\begin{figure*}[!pt]
\centering\normalsize\setlength{\tabcolsep}{0.5\columnsep}
\begin{tabular}{@{}C{\textwidth}@{}}
Entropic Pricing and Hedging: Normally Distributed Returns \\[0.5ex] \hline
\\[-1ex]
{\begin{minipage}[t]{\textwidth}\useparinfo
Fix a finite investment horizon and let $\Exp$ be the expectation measure mapping variables at the final time $T=t+dt$ to their expectations at the initial time $t$. Assume the market is tiered into a set of underlying securities with initial price vector $q$ and final price vector $Q=q+dq$ and a derivative security with initial price $p$ and final price $P=p+dp$. Also assume that the underlying and derivative returns have joint normal distribution in the expectation measure. Consider two self-funded strategies: one comprising underlying securities only; and one that also includes the derivative security.
\end{minipage}} \\\\[-1ex]
{\begin{minipage}[t]{0.5\textwidth-\tabcolsep}\useparinfo
{\bf Underlying only:} In this strategy, the investment comprises an underlying portfolio $\omega$ satisfying the initial condition:
\begin{equation}
\omega\cdot q=0
\end{equation}
so that the investment is self-funded. The return on the investment has entropy-adjusted mean:
\begin{align}
\Exp[\alpha]&[\omega\cdot dq]= \\
&\Exp[\omega\cdot dq]-\frac{1}{2}\alpha\Var[\omega\cdot dq] \notag
\end{align}
at the risk aversion $\alpha>0$. The optimal portfolio $\omega$ that maximises this is given by the {\bf calibration solution}:
\begin{align}
r\,dt&=\bfrac{\Var[dq]^{-1}q\cdot\Exp[dq]}{\Var[dq]^{-1}q\cdot q} \\
\alpha\omega&=\Var[dq]^{-1}(\Exp[dq]-qr\,dt) \notag
\end{align}
where $r\,dt$ is the Lagrange multiplier used to implement the self-funding condition.
\end{minipage}}
\hfill\vline\hfill
{\begin{minipage}[t]{0.5\textwidth-\tabcolsep}\useparinfo
{\bf Underlying and derivative:} In this strategy, the investment comprises the derivative security and an underlying portfolio $(\omega-\delta)$ satisfying the initial condition:
\begin{equation}
p+(\omega-\delta)\cdot q=0
\end{equation}
so that the investment is self-funded. The return on the investment has entropy-adjusted mean:
\begin{align}
\Exp[\alpha]&[dp+(\omega-\delta)\cdot dq]= \\
&\Exp[dp+(\omega-\delta)\cdot dq]-\frac{1}{2}\alpha\Var[dp+(\omega-\delta)\cdot dq] \notag
\end{align}
at the risk aversion $\alpha>0$. The hedge portfolio $\delta$ that maximises this is given by the {\bf hedge solution}:
\begin{align}
s\,dt&=\bfrac{p-\Var[dq]^{-1}\Var[dq,dp]\cdot q}{\Var[dq]^{-1}q\cdot q} \\
\delta&=\Var[dq]^{-1}(\Var[dq,dp]+qs\,dt) \notag
\end{align}
where $(r+\alpha s)\,dt$ is the Lagrange multiplier used to implement the self-funding condition.
\end{minipage}} \\\\[-1ex]
{\begin{minipage}[t]{\textwidth}\useparinfo
These two strategies offer different entropy-adjusted mean returns to the investor. When the derivative security has no intrinsic value beyond its contractual link with the underlying securities, the efficient market equilibrates to the indifferent state with matching entropy-adjusted means, leading to the {\bf price solution} $p=\delta\cdot q$ with hedge portfolio:
\begin{equation}
\delta=\Var[dq]^{-1}\left(\Var[dq,P]+\frac{1}{\alpha}\left(\sqrt{1+2\alpha\bfrac{\Exp[\alpha][P-\Var[dq]^{-1}\Var[dq,P]\cdot Q]}{\Var[dq]^{-1}\bar{q}\cdot\bar{q}}}-1\right)\bar{q}\right)
\end{equation}
where $\bar{q}:=q(1+r\,dt)$ inflates the initial underlying price $q$ by the implied funding rate $r$. The hedge portfolio of underlying securities simultaneously funds and hedges the derivative security, and the price solution decomposes the initial derivative price as the weighted sum of its component initial underlying prices.
\end{minipage}} \\
\\[-1ex] \hline
\end{tabular}
\caption{When the underlying and derivative returns have joint normal distribution in the expectation measure, the calibration, hedge and price conditions are solved explicitly in terms of their means and covariances.}
\label{fig:entropicpricingnormal}
\end{figure*}

\begin{figure*}[!pt]
\centering\normalsize\setlength{\tabcolsep}{0.5\columnsep}
\begin{tabular}{@{}C{\textwidth}@{}}
Entropic Pricing and Hedging: Continuous Price Diffusion \\[0.5ex] \hline
\\[-1ex]
{\begin{minipage}[t]{\textwidth}\useparinfo
Fix a finite investment horizon and let $\Exp$ be the expectation measure mapping variables at the final time $T=t+dt$ to their expectations at the initial time $t$. Assume the market is tiered into a set of underlying securities with initial price vector $q$ and final price vector $Q=q+dq$ and a derivative security with initial price $p$ and final price $P=p+dp$. Also assume that the prices diffuse continuously, so that their covariances are $O[dt]$ and their higher cumulants are $O[dt^2]$. Consider two self-funded strategies: one comprising underlying securities only; and one that also includes the derivative security.
\end{minipage}} \\\\[-1ex]
{\begin{minipage}[t]{0.5\textwidth-\tabcolsep}\useparinfo
{\bf Underlying only:} In this strategy, the investment comprises an underlying portfolio $\omega$ satisfying the initial condition:
\begin{equation}
\omega\cdot q=0
\end{equation}
so that the investment is self-funded. The return on the investment has entropy-adjusted mean:
\begin{align}
\Exp[\alpha]&[\omega\cdot dq]= \\
&\Exp[\omega\cdot dq]-\frac{1}{2}\alpha\Var[\omega\cdot dq]+O[dt^2] \notag
\end{align}
at the risk aversion $\alpha>0$. The optimal portfolio $\omega$ that maximises this is given by the {\bf calibration solution}:
\begin{align}
r\,dt&=\bfrac{\Var[dq]^{-1}q\cdot\Exp[dq]}{\Var[dq]^{-1}q\cdot q}+O[dt^2] \\
\alpha\omega&=\Var[dq]^{-1}(\Exp[dq]-qr\,dt)+O[dt] \notag
\end{align}
where $r\,dt$ is the Lagrange multiplier used to implement the self-funding condition.
\end{minipage}}
\hfill\vline\hfill
{\begin{minipage}[t]{0.5\textwidth-\tabcolsep}\useparinfo
{\bf Underlying and derivative:} In this strategy, the investment comprises the derivative security and an underlying portfolio $(\omega-\delta)$ satisfying the initial condition:
\begin{equation}
p+(\omega-\delta)\cdot q=0
\end{equation}
so that the investment is self-funded. The return on the investment has entropy-adjusted mean:
\begin{align}
\Exp&[\alpha][dp+(\omega-\delta)\cdot dq]= \\
&\Exp[dp+(\omega-\delta)\cdot dq]-\frac{1}{2}\alpha\Var[dp+(\omega-\delta)\cdot dq]+O[dt^2] \notag
\end{align}
at the risk aversion $\alpha>0$. The hedge portfolio $\delta$ that maximises this is given by the {\bf hedge solution}:
\begin{align}
s\,dt&=\bfrac{p-\Var[dq]^{-1}\Var[dq,dp]\cdot q}{\Var[dq]^{-1}q\cdot q}+O[dt^2] \\
\delta&=\Var[dq]^{-1}(\Var[dq,dp]+qs\,dt)+O[dt] \notag
\end{align}
where $(r+\alpha s)\,dt$ is the Lagrange multiplier used to implement the self-funding condition.
\end{minipage}} \\\\[-1ex]
{\begin{minipage}[t]{\textwidth}\useparinfo
These two strategies offer different entropy-adjusted mean returns to the investor. When the derivative security has no intrinsic value beyond its contractual link with the underlying securities, the efficient market equilibrates to the indifferent state with matching entropy-adjusted means, leading to the {\bf price solution}:
\begin{equation}
p=\frac{1}{1+r\,dt}\Exp^\phi[P]-\frac{1}{2}\alpha\Var[P-\Var[dq]^{-1}\Var[dq,P]\cdot dq]+O[dt^2]
\end{equation}
The first term expresses the discounted derivative price as a martingale in the price measure. The second term then adjusts the price to account for the residual risk after hedging.
\end{minipage}} \\
\\[-1ex] \hline
\end{tabular}
\caption{When the underlying and derivative prices diffuse continuously in the expectation measure, the calibration, hedge and price conditions are solved explicitly to $O[dt]$ in terms of their means and covariances. These solutions can be used when the portfolio is rebalanced continuously.}
\label{fig:entropicpricingcontinuous}
\end{figure*}

The market of underlying securities accessible to the investor performs the two key functions of {\em funding} and {\em hedging} for the derivative security. Fix a finite investment horizon with initial time $t$ and final time $T=t+dt$, and assume the market is tiered into a set of underlying securities with initial price vector $q$ and final price vector $Q=q+dq$ and a derivative security with initial price $p$ and final price $P=p+dp$. The investor quantifies the performance of the investment strategy using an initial assessment of the final economic state, encapsulated in the expectation measure $\Exp$ that maps final variables to their initial expectations. Following the strategy developed in the previous section, the optimal portfolio then maximises the entropy-adjusted mean return at a specified risk aversion $\alpha$ for the investor.

There are two self-funded investment strategies to consider: one that includes only underlying securities; and one that also includes the derivative security. In the first strategy, the portfolio $\omega$ of underlying securities is optimal when it maximises the entropy-adjusted mean return $\Exp[\alpha][\omega\cdot dq]$. The optimal portfolio thus satisfies the {\bf calibration condition}:
\begin{equation}
\Exp[(dq-qr\,dt)\exp[-\alpha\omega\cdot dq]]=0
\end{equation}
where $r\,dt$ is the Lagrange multiplier used to implement the self-funding condition $\omega\cdot q=0$. Adjusting for the opportunity and risk that adding the derivative security presents, in the second strategy the optimal portfolio $\omega$ is modified by the hedge portfolio $\delta$, and optimality is restored in the combined portfolio when it maximises the entropy-adjusted mean return $\Exp[\alpha][dp+(\omega-\delta)\cdot dq]$. The hedge portfolio thus satisfies the {\bf hedge condition}:
\begin{equation}
\Exp[(dq-q(r+\alpha s)\,dt)\exp[-\alpha(dp+(\omega-\delta)\cdot dq)]]=0
\end{equation}
where $(r+\alpha s)\,dt$ is the Lagrange multiplier used to implement the self-funding condition $p+(\omega-\delta)\cdot q=0$.

Returns from correlated changes in the derivative and underlying prices are cancelled by hedging, with the residual return on the combined portfolio maximising the entropy-adjusted mean. When its price additionally satisfies the {\bf price condition}:
\begin{equation}
\Exp[\alpha][\omega\cdot dq]=\Exp[\alpha][dp+(\omega-\delta)\cdot dq]
\end{equation}
the derivative security neither advantages nor disadvantages the investment strategy. Since it has no value beyond its contractual links, this indifference relationship is assumed as the equilibrium condition for the derivative price. The market mechanism that achieves equilibrium is considered in more detail later in the essay.

The calibration, hedge and price conditions, together with the two self-funding conditions, are the foundations of entropic pricing. By avoiding the unrealistic assumptions that enable risk-neutral theory, entropic pricing can be consistently applied to pricing and hedging across a range spanning everything from algorithmic trading to XVA and capital modelling.

\subsection{Calibration}

Investments are self-funded within the underlying market, and the entropic model does not assume the existence of a risk-free funding rate for investments that are not self-funded. Calibration identifies the price measure $\Exp^\phi$ equivalent to the expectation measure $\Exp$ with Radon-Nikodym weight:
\begin{equation}
\bfrac{d\Exp^\phi}{d\Exp\hfill}:=\bfrac{\exp[-\phi\cdot dq]}{\Exp[\exp[-\phi\cdot dq]]}
\end{equation}
where the unit optimal portfolio $\phi$ and implied funding rate $r$ are solved from the calibration and self-funding conditions:
\begin{align}
qr\,dt&=\Exp^\phi[dq] \\
\phi\cdot q&=0 \notag
\end{align}
The implied funding rate defined here originates in the optimisation as the Lagrange multiplier for the self-funding condition.

When the portfolio is rebalanced on the discrete schedule $t_0<t_1<\cdots<t_n$, at each time $t_i$ the unit optimal portfolio $\phi_i$ is determined at the start of the investment interval $t_i<t_{i+1}$ and the calibration condition is compounded to generate the price measure $\Exp_i^\phi$ equivalent to the expectation measure $\Exp_i$:
\begin{align}
\bfrac{d\Exp_i^\phi}{d\Exp_i\hfill}:={}&\prod_{j=i}^{n-1}\bfrac{\exp[-\phi_j\cdot dq_j]}{\Exp_j[\exp[-\phi_j\cdot dq_j]]} \\
={}&\exp\!\left[\sum_{j=i}^{n-1}(-\phi_j\cdot dq_j-\log\Exp_j\exp[-\phi_j\cdot dq_j])\right] \notag
\end{align}
for variables on the discrete schedule. In the limit of continuous rebalancing, the sum in the exponent is replaced by the integral.

\begin{figure*}[!pt]
\centering
\begin{tabular}{@{}C{0.33\textwidth-\tabcolsep}C{0.33\textwidth-\tabcolsep}C{0.33\textwidth-\tabcolsep}@{}}
\includegraphics[width=0.33\textwidth-\tabcolsep]{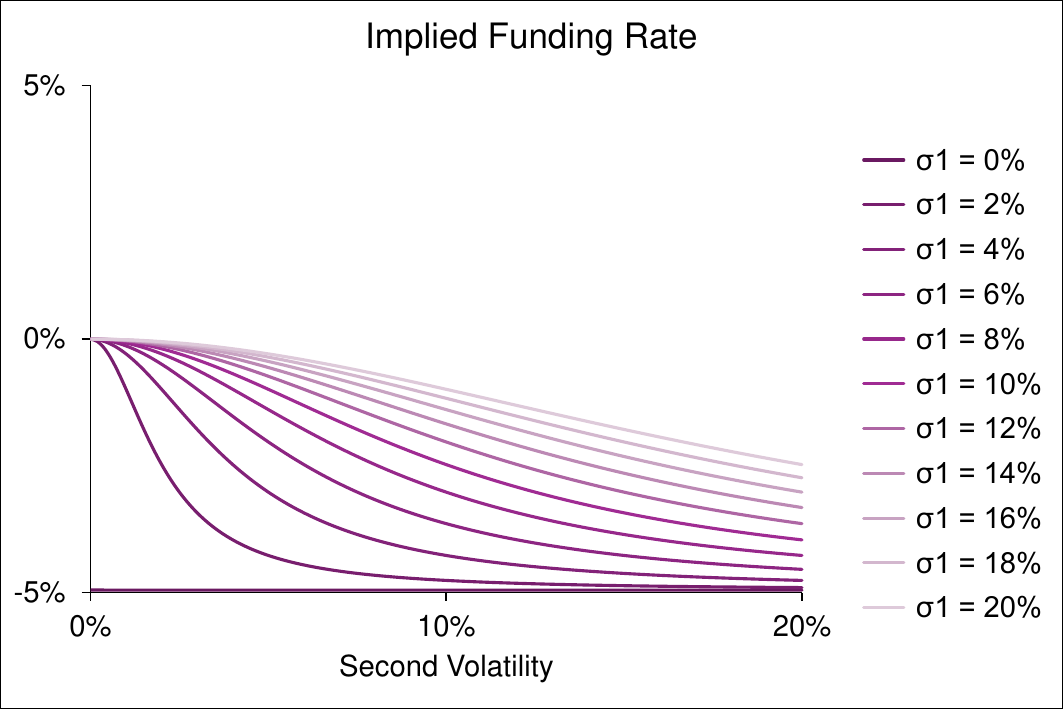}&
\includegraphics[width=0.33\textwidth-\tabcolsep]{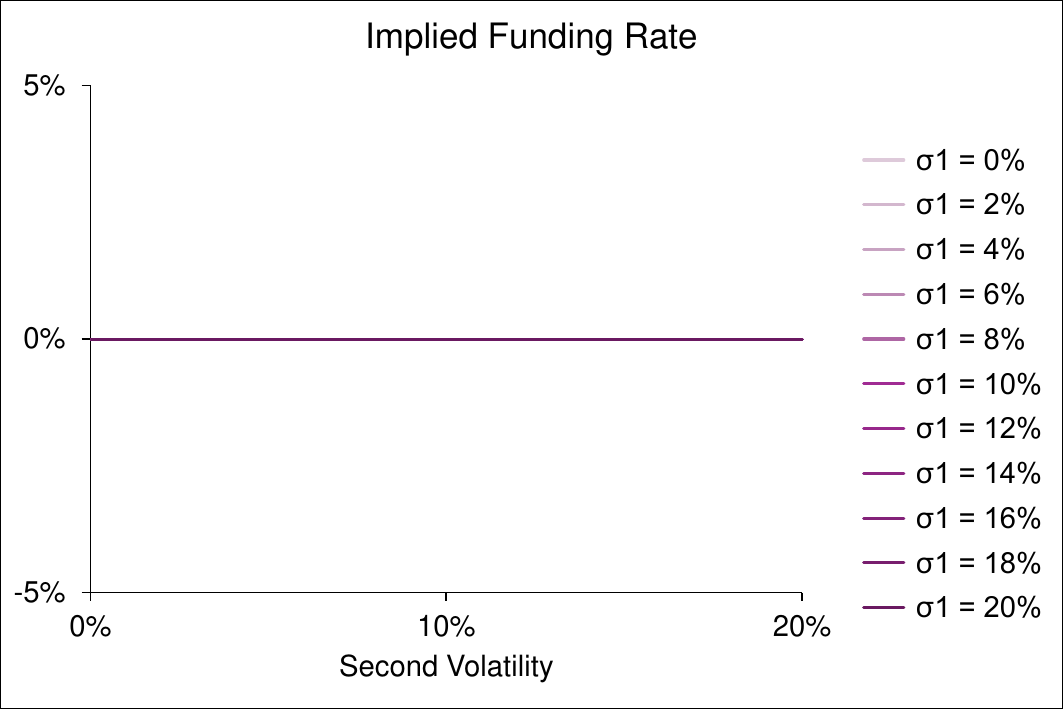}&
\includegraphics[width=0.33\textwidth-\tabcolsep]{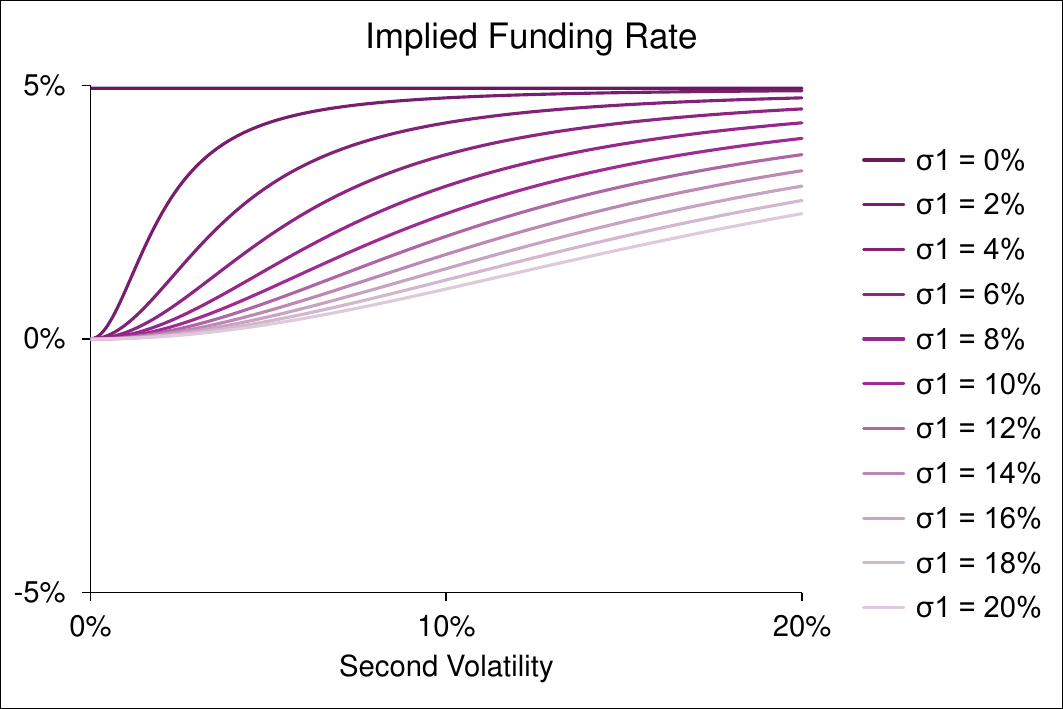}\\
Normal ($\mu_1=-5\%$)&Normal ($\mu_1=0\%$)&Normal ($\mu_1=5\%$)\\
\includegraphics[width=0.33\textwidth-\tabcolsep]{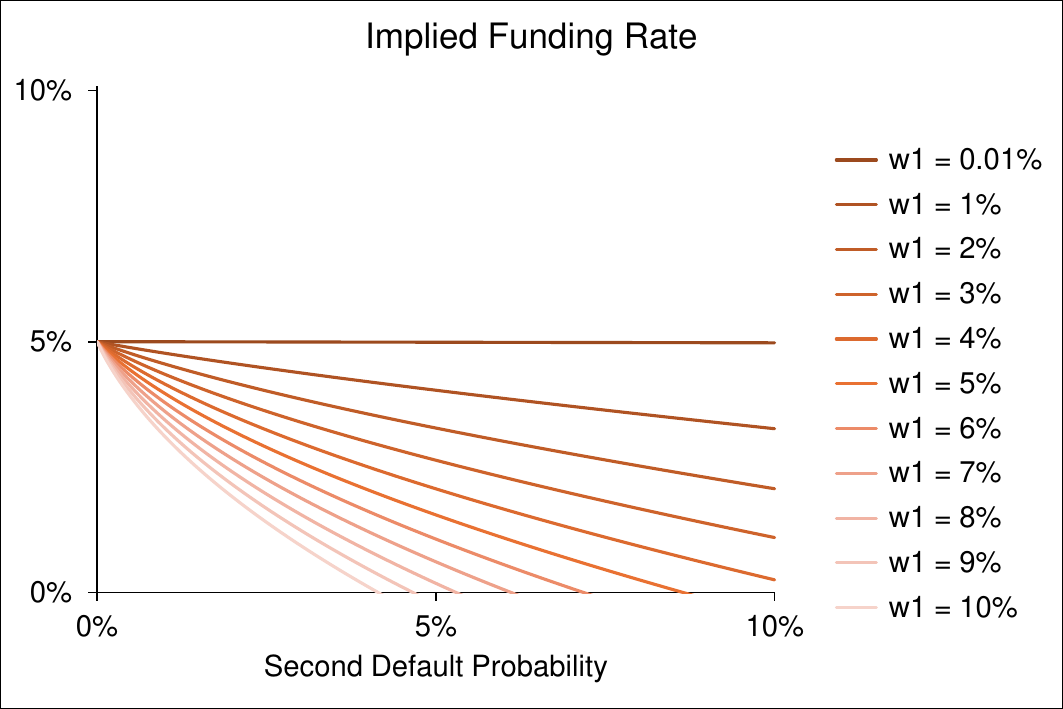}&
\includegraphics[width=0.33\textwidth-\tabcolsep]{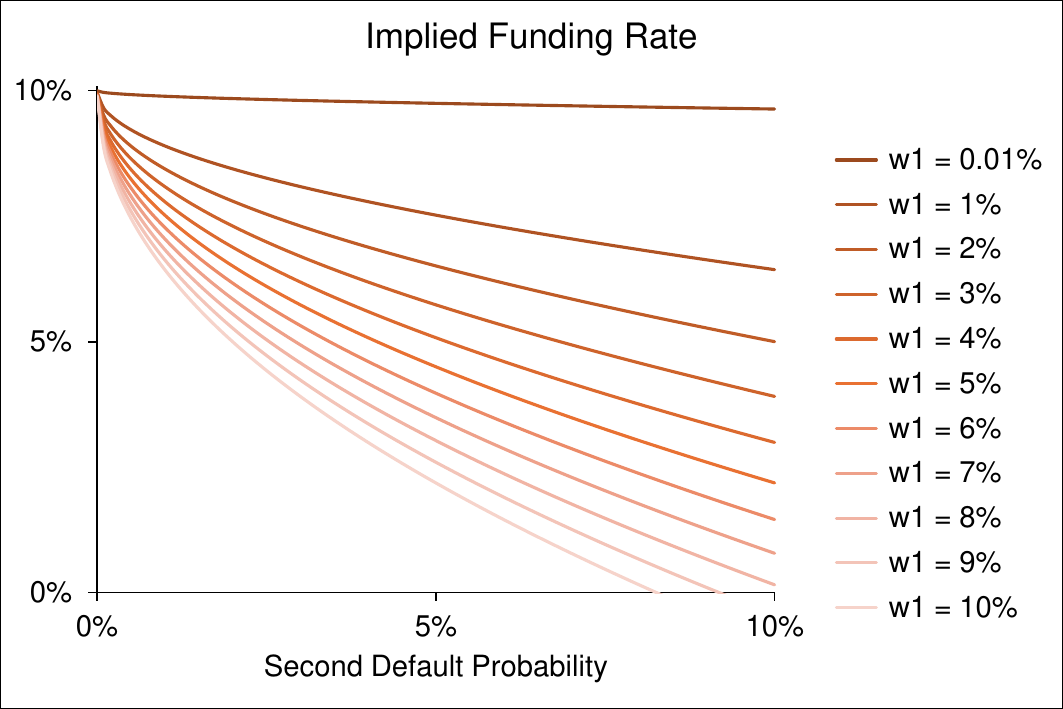}&
\includegraphics[width=0.33\textwidth-\tabcolsep]{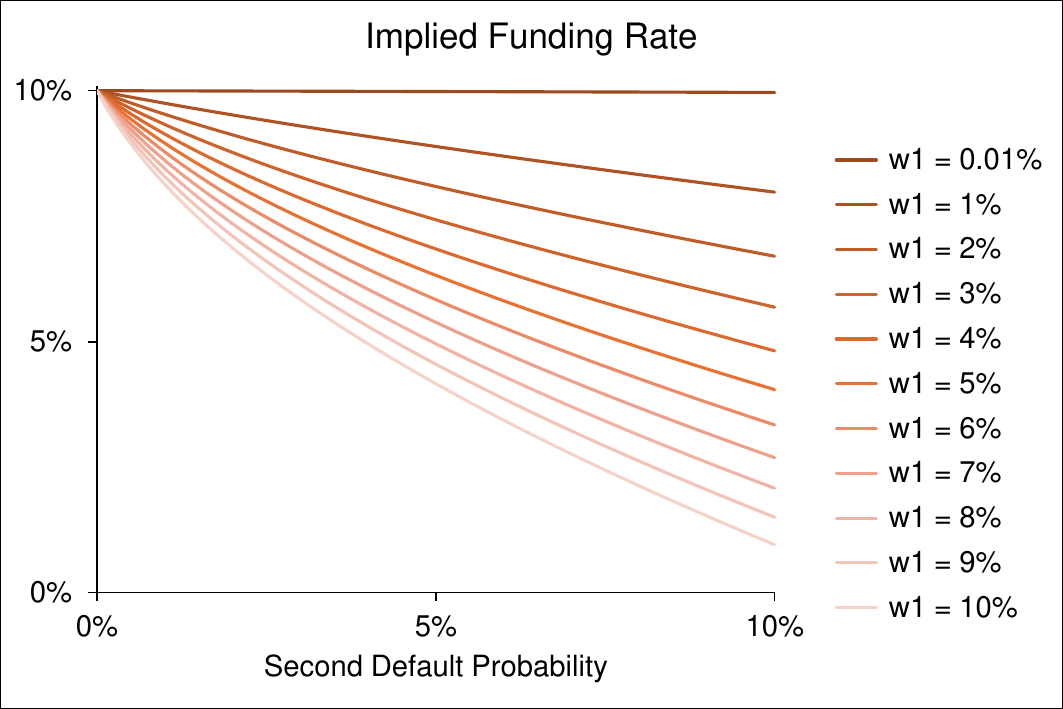}\\
Binomial ($c_1=5\%$)&Binomial ($c_1=10\%$)&Binomial ($c_1=15\%$)
\end{tabular}
\caption{Discounting is an emergent feature in entropic pricing, where the funding rate is implied as the Lagrange multiplier that implements the self-funding condition. In these studies the investment horizon is $dt=1$ and the underlying market comprises two securities with unit initial prices and independent final prices.
\\\hspace*{\savedindent}In the first study, the final prices are modelled as independent normal variables with mean $q(1+\mu\,dt)$ and variance $q^2\sigma^2\,dt$. Examples are plotted for the first underlying with drift $\mu_1=-5\%,0\%,5\%$ and volatility $\sigma_1=0\%,\ldots,20\%$. The second underlying has drift $\mu_2=0\%$ and the implied funding rate is plotted as a function of its volatility $\sigma_2$. The implied funding rate is inbetween the two drifts and gravitates to the drift of the underlying with the lowest volatility.
\\\hspace*{\savedindent}In the second study, the final prices are modelled as independent binomial variables with possible values $(0,1+c\,dt)$ and probabilities $(w,1-w)$. Examples are plotted for the first underlying with coupon $c_1=5\%,10\%,15\%$ and default probability $w_1=0.01\%,\ldots,10\%$. The second underlying has coupon $c_2=10\%$ and the implied funding rate is plotted as a function of its default probability $w_2$. The implied funding rate is capped by the coupons and decreases as the default probability increases.}
\label{fig:fundingrate}
\end{figure*}

While there is no explicit solution to the calibration condition, the unit optimal portfolio and implied funding rate can be identified in some scenarios, and more generally are efficiently discovered by iterating the Newton-Raphson scheme. In each of these scenarios, the expectation in the calibration condition is simplified or approximated, which facilitates its solution.
\begin{description}[leftmargin=0\parindent]
\item[Normal model]When the underlying prices are joint normally distributed in the expectation measure, their cumulant generating function is:
\begin{align}
\log\Exp\exp[\psi\cdot dq]&= \\
&\hspace{-2cm}\psi\cdot\Exp[dq]+\frac{1}{2}\psi\cdot\Var[dq]\psi \notag \\
\log\Exp^\phi\exp[\psi\cdot dq]&= \notag \\
&\hspace{-2cm}\psi\cdot(\Exp[dq]-\Var[dq]\phi)+\frac{1}{2}\psi\cdot\Var[dq]\psi \notag
\end{align}
The calibration condition becomes:
\begin{equation}
qr\,dt=\Exp[dq]-\Var[dq]\phi
\end{equation}
which is solved with the unit optimal portfolio:
\begin{equation}
\phi=\Var[dq]^{-1}(\Exp[dq]-qr\,dt)
\end{equation}
The implied funding rate is then determined from the self-funding condition:
\begin{equation}
r\,dt=\bfrac{\Var[dq]^{-1}q\cdot\Exp[dq]}{\Var[dq]^{-1}q\cdot q}
\end{equation}
\item[Continuous model]When the underlying prices diffuse continuously in the expectation measure, their cumulant generating function is:
\begin{align}
\log\Exp\exp[\psi\cdot dq]&= \\
&\hspace{-2cm}\psi\cdot\Exp[dq]+\frac{1}{2}\psi\cdot\Var[dq]\psi+O[dt^2] \notag \\
\log\Exp^\phi\exp[\psi\cdot dq]&= \notag \\
&\hspace{-2cm}\psi\cdot(\Exp[dq]-\Var[dq]\phi)+\frac{1}{2}\psi\cdot\Var[dq]\psi+O[dt^2] \notag
\end{align}
where the first-order cumulants are $O[1]$ or higher, the second-order cumulants are $O[dt]$ and the higher-order cumulants are $O[dt^2]$ or higher. The calibration condition becomes:
\begin{equation}
qr\,dt=\Exp[dq]-\Var[dq]\phi+O[dt^2]
\end{equation}
which is solved with the unit optimal portfolio:
\begin{equation}
\phi=\Var[dq]^{-1}(\Exp[dq]-qr\,dt)+O[dt]
\end{equation}
The implied funding rate is then determined from the self-funding condition:
\begin{equation}
r\,dt=\bfrac{\Var[dq]^{-1}q\cdot\Exp[dq]}{\Var[dq]^{-1}q\cdot q}+O[dt^2]
\end{equation}
In addition to the assumption that the underlying prices diffuse continuously, this solution is applicable only when the optimal portfolio is rebalanced with sufficient frequency so that the $O[dt]$ corrections to the strategy can be neglected.
\item[Numerical scheme]At each step the Newton-Raphson scheme improves the approximate unit optimal portfolio $\phi$ by the correction $\partial\phi$ which solves the calibration and self-funding conditions to $O[\partial\phi]$. Neglecting $O[\partial\phi^2]$ and higher terms, the calibration condition becomes:
\begin{equation}
qr\,dt=\Exp^\phi[dq]-\Var^\phi[dq]\,\partial\phi
\end{equation}
which is solved with the correction:
\begin{equation}
\partial\phi=\Var^\phi[dq]^{-1}(\Exp^\phi[dq]-qr\,dt)
\end{equation}
The implied funding rate $r$ is then determined from the self-funding condition:
\begin{equation}
r\,dt=\bfrac{(\phi+\Var^\phi[dq]^{-1}\Exp^\phi[dq])\cdot q}{\Var^\phi[dq]^{-1}q\cdot q}
\end{equation}
The Newton-Raphson scheme is initialised with the choice $\phi=0$ and repeats these steps until the calibration condition is satisfied to the required level of accuracy.
\end{description}

The implied funding rate incorporates contributions from all the underlying returns, weighted by their predictabilities so that the least risky returns dominate. This property is demonstrated in the studies of \cref{fig:fundingrate}, which simplify the model with the assumption that the final underlying prices are independent. The calibration conditions decouple to generate a set of independent conditions, with the $i$th underlying security satisfying:
\begin{equation}
r=\frac{1}{dt}\left(\bfrac{\Exp[Q_i\exp[-\phi_iQ_i]]}{q_i\Exp[\exp[-\phi_iQ_i]]}-1\right)
\end{equation}
These conditions are solved separately to express the components of the unit optimal portfolio $\phi[r]$ as functions of the implied funding rate $r$. The self-funding condition $\phi[r]\cdot q=0$ is then solved for $r$.

In the first study, the final underlying price $Q_i$ is modelled as a normal variable with mean $q_i(1+\mu_i\,dt)$ and variance $q_i^2\sigma_i^2\,dt$, where $\mu_i$ is the drift and $\sigma_i$ is the volatility. The calibration condition derives:
\begin{equation}
\phi_i=\frac{\mu_i-r}{q_i^{\vphantom{2}}\sigma_i^2}
\end{equation}
and the self-funding condition then derives:
\begin{equation}
r=\sum\nolimits_i\bfrac{1/\sigma_i^2}{\sum\nolimits_j1/\sigma_j^2}\,\mu_i
\end{equation}
The implied funding rate is the weighted average of the drifts, with weight inversely proportional to the square of the volatility.

In the second study, the final underlying price $Q_i$ is modelled as a binomial variable with possible values $(0,1+c_i\,dt)$ occurring with probabilities $(w_i,1-w_i)$, where $c_i$ is the coupon and $w_i$ is the default probability. The calibration condition derives:
\begin{equation}
\phi_i=\frac{1}{1+c_i\,dt}\log\!\left[\frac{1-w_i}{w_i}\left(\frac{1+c_i\,dt}{q_i(1+r\,dt)}-1\right)\right]
\end{equation}
and the self-funding condition:
\begin{equation}
0=\sum\nolimits_i\frac{q_i}{1+c_i\,dt}\log\!\left[\frac{1-w_i}{w_i}\left(\frac{1+c_i\,dt}{q_i(1+r\,dt)}-1\right)\right]
\end{equation}
is solved for the implied funding rate.

\subsection{Hedging and pricing}

The price condition determines the initial price $p$ from the final price $P$ of the derivative security via the relationship:
\begin{equation}
0=\Exp^\phi[\alpha][dp-\delta\cdot dq]
\end{equation}
where the investment is funded by the hedge portfolio $\delta$ in the underlying market, chosen to maximise the entropy-adjusted mean return at risk aversion $\alpha$ subject to the self-funding condition $\delta\cdot q=p$.

When $\alpha=0$ the investor neglects the impact of residual risk from the hedging strategy and the derivative price is a discounted martingale in the price measure:
\begin{equation}
p=\frac{1}{1+r\,dt}\Exp^\phi[P]
\end{equation}
Over the discrete schedule $t_0<t_1<\cdots<t_n$ the price condition is compounded to derive:
\begin{equation}
p_i=\Exp_i^\phi\!\left[\prod_{j=i}^{n-1}\frac{1}{1+r_j\,dt_j}p_n\right]
\end{equation}
The model bears all the hallmarks of risk-neutral theory, including a kernel that performs stochastic discounting and a price measure that is equivalent to the expectation measure, but applies more generally. Entropic risk optimisation uniquely identifies both these ingredients from the optimal strategy in the underlying market, and does not need the additional assumptions of market completeness and the existence of risk-free funding.

When $\alpha\ne0$ the investor reserves against residual risk from the hedging strategy and the derivative price is a discounted log-martingale in the price measure:
\begin{equation}
p=\Exp^\phi[\alpha][P-\delta\cdot dq]
\end{equation}
Over the discrete schedule $t_0<t_1<\cdots<t_n$ the price condition is compounded to derive:
\begin{equation}
p_i=\Exp_i^\phi[\alpha]\!\left[p_n-\sum_{j=i}^{n-1}\delta_j\cdot dq_j\right]
\end{equation}
The model adjusts the initial price to protect the hedging strategy from the negative consequences of unhedged risks. If the underlying market is complete for the derivative security, so that the return on the hedged portfolio has zero variance, then no buffer is required and the price simplifies to the risk-neutral solution.

Hedging identifies the price measure $\Exp^{\phi\delta}:=\Exp^{(\phi-\alpha\delta)\oplus\alpha}$ equivalent to the price measure $\Exp^\phi$ with Radon-Nikodym weight:
\begin{equation}
\bfrac{d\Exp^{\phi\delta}}{d\Exp^\phi\hfill}:=\bfrac{\exp[-\alpha(P-\delta\cdot dq)]}{\Exp^\phi[\exp[-\alpha(P-\delta\cdot dq)]]}
\end{equation}
where the hedge portfolio $\delta$ and implied funding spread $s$ are solved from the hedge and self-funding conditions:
\begin{align}
q(r+\alpha s)\,dt&=\Exp^{\phi\delta}[dq] \\
\delta\cdot q&=\Exp^\phi[\alpha][P-\delta\cdot dq] \notag
\end{align}
The implied funding spread accounts for the incremental funding cost or benefit arising from the inclusion of the derivative security. Given the hedge portfolio $\delta$, the initial derivative price $p=\delta\cdot q$ decomposes as the weighted sum of its component initial underlying prices. The hedge portfolio of underlying securities thus simultaneously funds and hedges the derivative security, and is optimally selected for both objectives.

While there is no explicit solution to the hedge condition, the hedge portfolio and implied funding spread can be identified in some scenarios, and more generally are efficiently discovered by iterating the Newton-Raphson scheme. In each of these scenarios, the expectations in the hedge and self-funding conditions are simplified or approximated, which facilitates their solution.
\begin{description}[leftmargin=0\parindent]
\item[Normal model]When the underlying and derivative prices are joint normally distributed in the expectation measure, their cumulant generating function is:
\begin{align}
\log\Exp\exp[\beta\,dp+\psi\cdot dq]&= \\
&\hspace{-3cm}\beta\Exp[dp]+\psi\cdot\Exp[dq] \notag \\
&\hspace{-3cm}+\frac{1}{2}\beta^2\Var[dp]+\beta\psi\cdot\Var[dq,dp]+\frac{1}{2}\psi\cdot\Var[dq]\psi \notag \\
\log\Exp^\phi\exp[\beta\,dp+\psi\cdot dq]&= \notag \\
&\hspace{-3cm}\beta(\Exp[dp]-\phi\cdot\Var[dq,dp])+\psi\cdot(\Exp[dq]-\Var[dq]\phi) \notag \\
&\hspace{-3cm}+\frac{1}{2}\beta^2\Var[dp]+\beta\psi\cdot\Var[dq,dp]+\frac{1}{2}\psi\cdot\Var[dq]\psi \notag
\end{align}
With this assumption, the calibration condition is solved with unit optimal portfolio $\phi=\Var[dq]^{-1}(\Exp[dq]-qr\,dt)$ where the implied funding rate $r$ is tuned to the condition $\phi\cdot q=0$. The hedge condition becomes:
\begin{equation}
q(r+\alpha s)\,dt=qr\,dt-\alpha\Var[dq,P]+\alpha\Var[dq]\delta
\end{equation}
which is solved with the hedge portfolio:
\begin{equation}
\delta=\Var[dq]^{-1}(\Var[dq,P]+qs\,dt)
\end{equation}
The self-funding condition becomes:
\begin{align}
&\delta\cdot q= \\
&\quad\Exp[P]-\Var[dq]^{-1}(\Exp[dq]-qr\,dt)\cdot\Var[dq,P]-\frac{1}{2}\alpha\Var[P] \notag \\
&\quad-\delta\cdot(qr\,dt-\alpha\Var[dq,P])-\frac{1}{2}\alpha\delta\cdot\Var[dq]\delta \notag
\end{align}
which is solved with the implied funding spread:
\begin{align}
s\,dt={}&\frac{1}{\alpha}(1+r\,dt)\times \\
&\left(\sqrt{1+2\alpha\bfrac{\Exp[\alpha][P-\Var[dq]^{-1}\Var[dq,P]\cdot Q]}{\Var[dq]^{-1}q\cdot q(1+r\,dt)^2}}-1\right) \notag
\end{align}
The initial derivative price is then given by $p=\delta\cdot q$. This solution applies when the underlying and derivative prices are joint normally distributed, and is a useful approximation in near-normal scenarios.
\item[Continuous model]When the underlying and derivative prices diffuse continuously in the expectation measure, their cumulant generating function is:
\begin{align}
\log\Exp\exp[\beta\,dp+\psi\cdot dq]&= \\
&\hspace{-3cm}\beta\Exp[dp]+\psi\cdot\Exp[dq] \notag \\
&\hspace{-3cm}+\frac{1}{2}\beta^2\Var[dp]+\beta\psi\cdot\Var[dq,dp]+\frac{1}{2}\psi\cdot\Var[dq]\psi+O[dt^2] \notag \\
\log\Exp^\phi\exp[\beta\,dp+\psi\cdot dq]&= \notag \\
&\hspace{-3cm}\beta(\Exp[dp]-\phi\cdot\Var[dq,dp])+\psi\cdot(\Exp[dq]-\Var[dq]\phi) \notag \\
&\hspace{-3cm}+\frac{1}{2}\beta^2\Var[dp]+\beta\psi\cdot\Var[dq,dp]+\frac{1}{2}\psi\cdot\Var[dq]\psi+O[dt^2] \notag
\end{align}
where the first-order cumulants are $O[1]$ or higher, the second-order cumulants are $O[dt]$ and the higher-order cumulants are $O[dt^2]$ or higher. With this assumption, the calibration condition is solved with unit optimal portfolio $\phi=\Var[dq]^{-1}(\Exp[dq]-qr\,dt)+O[dt]$ where the implied funding rate $r$ is tuned to the condition $\phi\cdot q=0$. The hedge condition becomes:
\begin{equation}
q(r+\alpha s)\,dt=qr\,dt-\alpha\Var[dq,P]+\alpha\Var[dq]\delta+O[dt^2]
\end{equation}
which is solved with the hedge portfolio:
\begin{equation}
\delta=\Var[dq]^{-1}(\Var[dq,P]+qs\,dt)+O[dt]
\end{equation}
The self-funding condition becomes:
\begin{equation}
\delta\cdot q=\Exp[P]-\Var[dq]^{-1}\Exp[dq]\cdot\Var[dq,P]+O[dt]
\end{equation}
which is solved with the implied funding spread:
\begin{equation}
s\,dt=\bfrac{\Exp[P]-\Var[dq]^{-1}\Var[dq,P]\cdot\Exp[Q]}{\Var[dq]^{-1}q\cdot q}+O[dt^2]
\end{equation}
The initial derivative price is then given by:
\begin{align}
p={}&\frac{1}{1+r\,dt}\Exp^\phi[P] \\
&-\frac{1}{2}\alpha\Var[P-\Var[dq]^{-1}\Var[dq,P]\cdot dq]+O[dt^2] \notag
\end{align}
In addition to the assumption that the underlying and derivative prices diffuse continuously, this solution is applicable only when the hedge portfolio is rebalanced with sufficient frequency so that the $O[dt]$ corrections to the strategy can be neglected.
\item[Numerical scheme]At each step the Newton-Raphson scheme improves the approximate hedge portfolio $\delta$ by the correction $\partial\delta$ which solves the hedge and self-funding conditions to $O[\partial\delta]$. Neglecting $O[\partial\delta^2]$ and higher terms, the hedge condition becomes:
\begin{equation}
q(r+\alpha s)\,dt=\Exp^{\phi\delta}[dq]+\alpha\Var^{\phi\delta}[dq]\,\partial\delta
\end{equation}
which is solved with the correction:
\begin{equation}
\partial\delta=-\frac{1}{\alpha}\Var^{\phi\delta}[dq]^{-1}(\Exp^{\phi\delta}[dq]-q(r+\alpha s)\,dt)
\end{equation}
Neglecting $O[\partial\delta^2]$ and higher terms, the self-funding condition becomes:
\begin{equation}
(\delta+\partial\delta)\cdot q=\Exp^\phi[\alpha][P-\delta\cdot dq]-\partial\delta\cdot\Exp^{\phi\delta}[dq]
\end{equation}
which is solved with the implied funding spread:
\begin{align}
(r+\alpha s)\,dt={}& \\
&\hspace{-1cm}\bfrac{\alpha\Exp^\phi[\alpha][P-\delta\cdot Q]+\Var^{\phi\delta}[dq]^{-1}\Exp^{\phi\delta}[dq]\cdot\Exp^{\phi\delta}[Q]}{\Var^{\phi\delta}[dq]^{-1}q\cdot\Exp^{\phi\delta}[Q]} \notag
\end{align}
The Newton-Raphson scheme is initialised with the choice $\delta=0$ and repeats these steps until the hedge condition is satisfied to the required level of accuracy. The initial derivative price is then given by $p=\delta\cdot q$.
\end{description}

\begin{figure*}[!pt]
\centering
\begin{tabular}{@{}C{0.33\textwidth-\tabcolsep}C{0.33\textwidth-\tabcolsep}C{0.33\textwidth-\tabcolsep}@{}}
\includegraphics[width=0.33\textwidth-\tabcolsep]{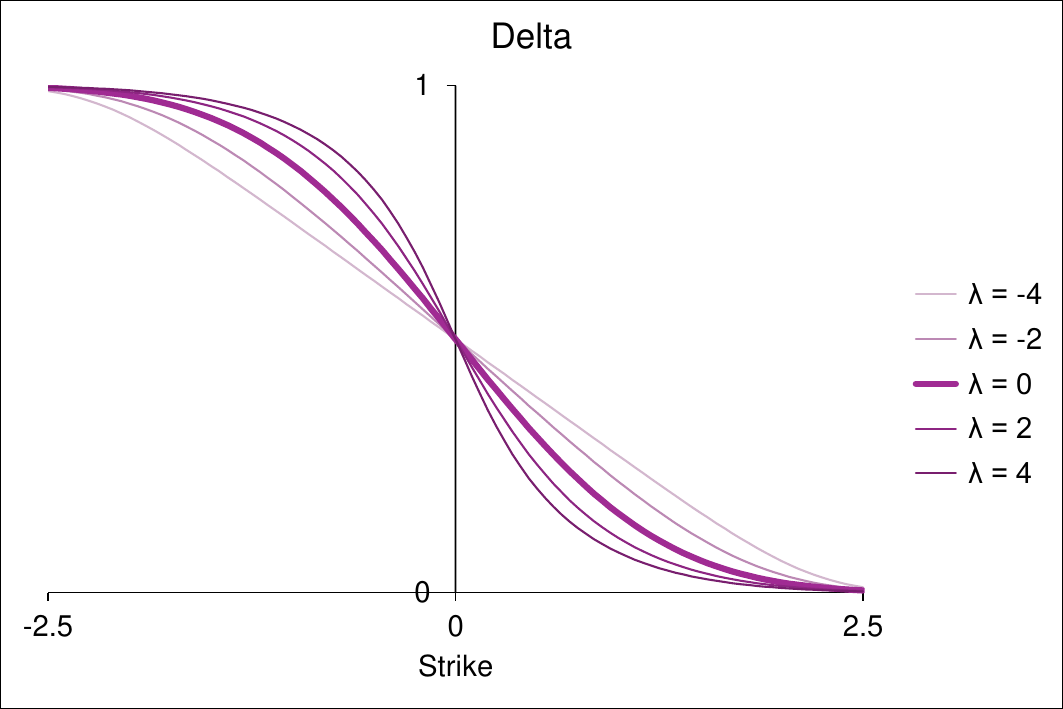}&
\includegraphics[width=0.33\textwidth-\tabcolsep]{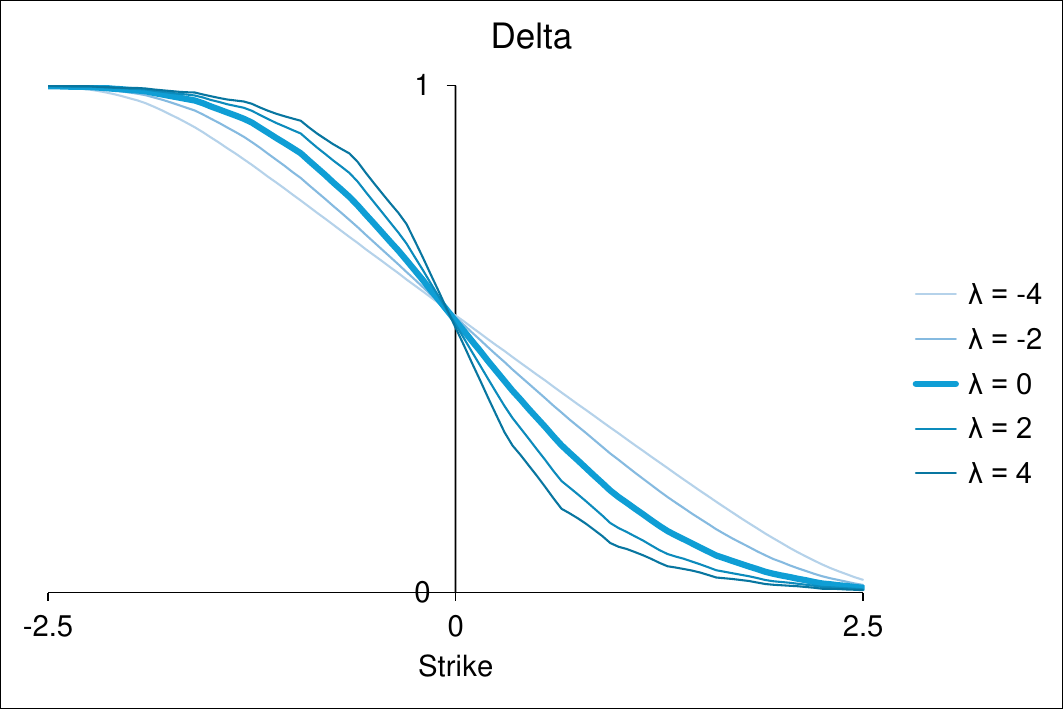}&
\includegraphics[width=0.33\textwidth-\tabcolsep]{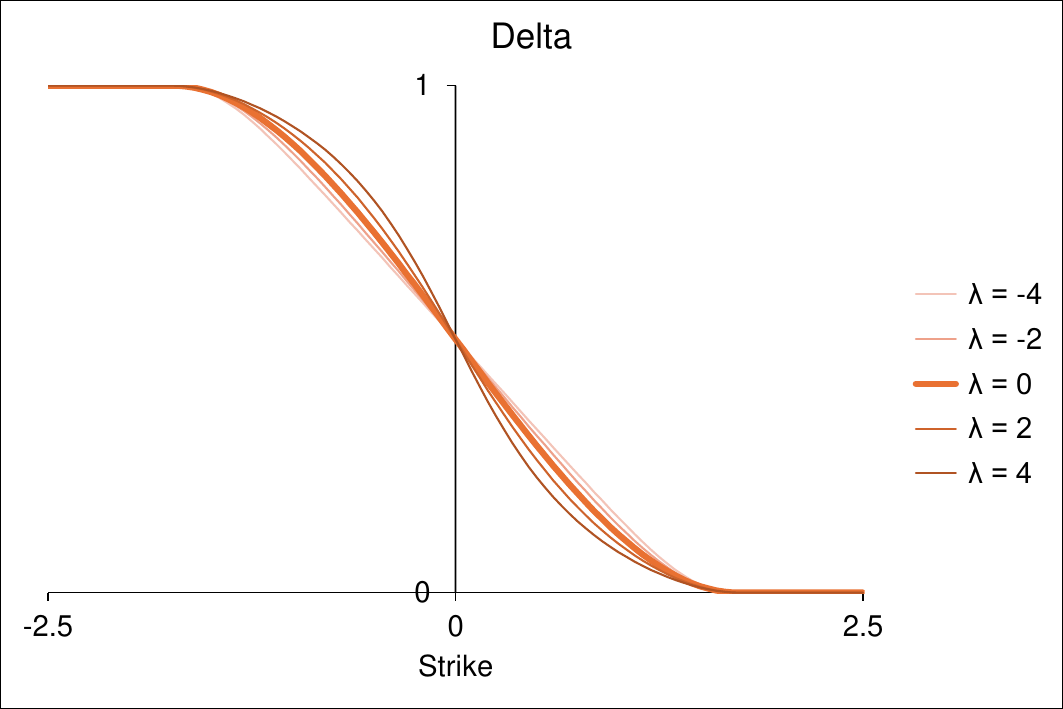}\\
\includegraphics[width=0.33\textwidth-\tabcolsep]{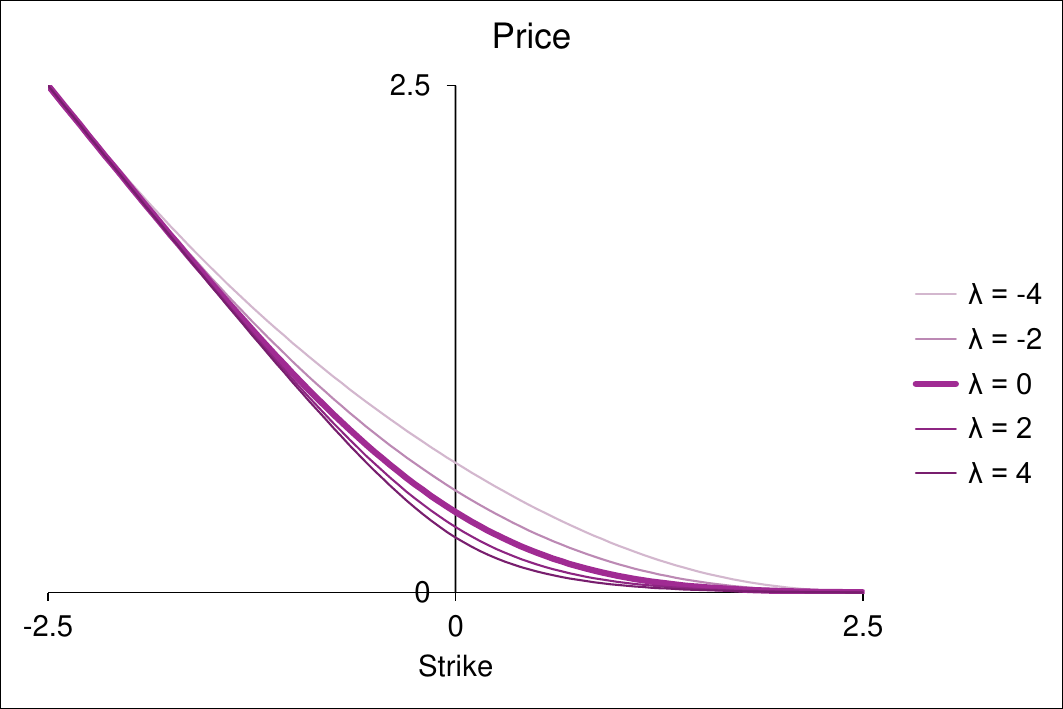}&
\includegraphics[width=0.33\textwidth-\tabcolsep]{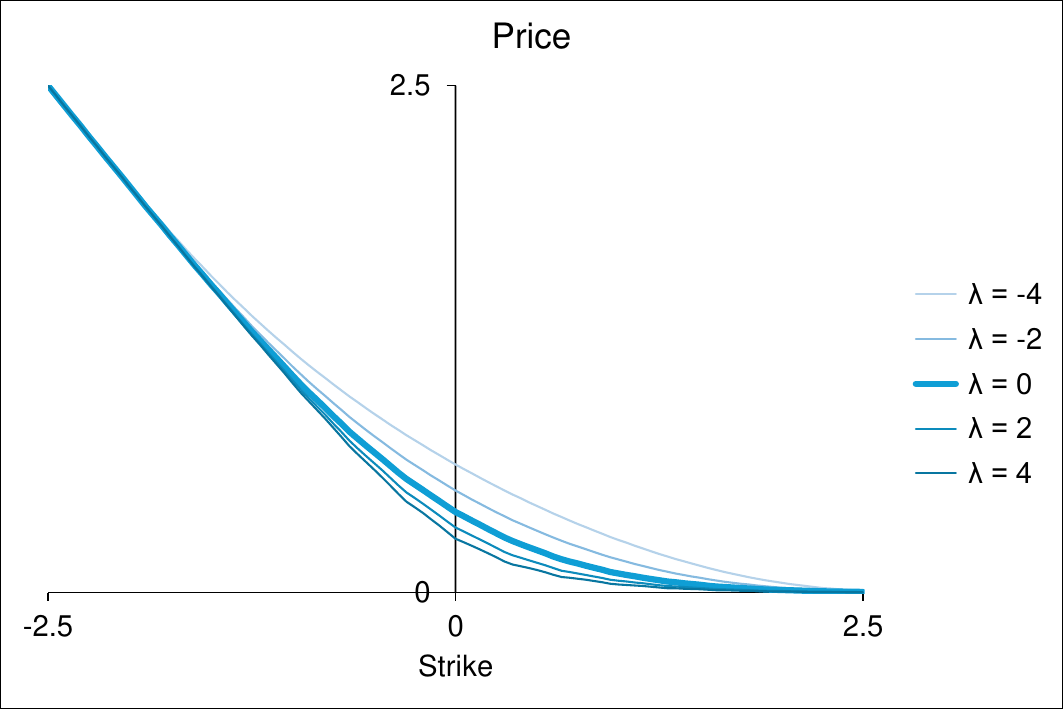}&
\includegraphics[width=0.33\textwidth-\tabcolsep]{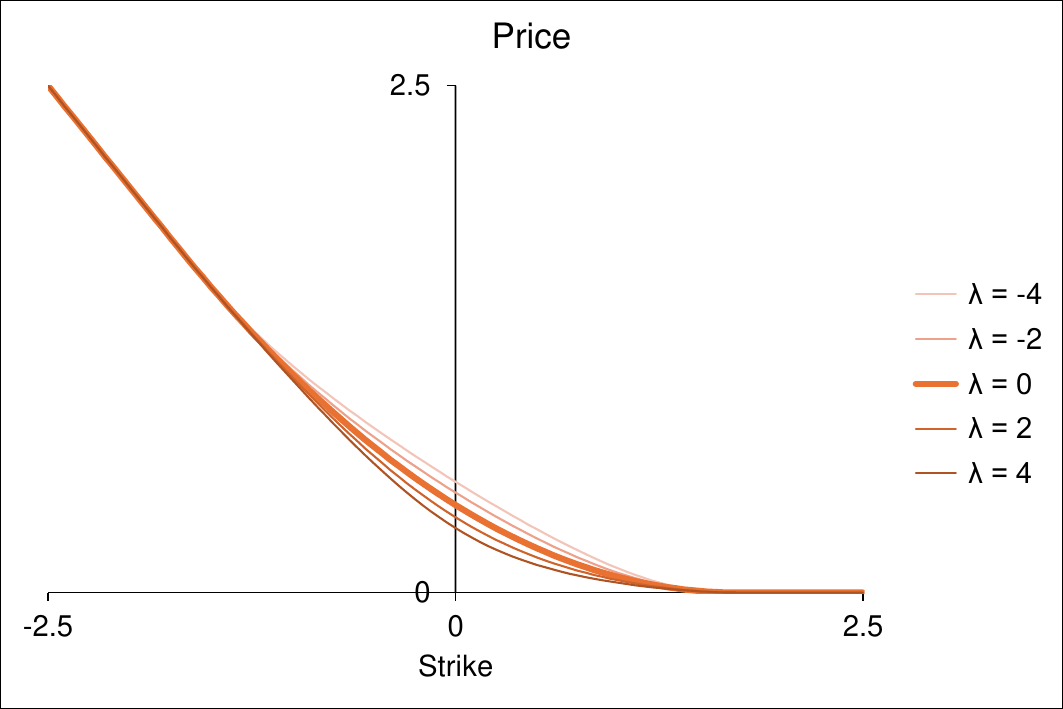}\\
\includegraphics[width=0.33\textwidth-\tabcolsep]{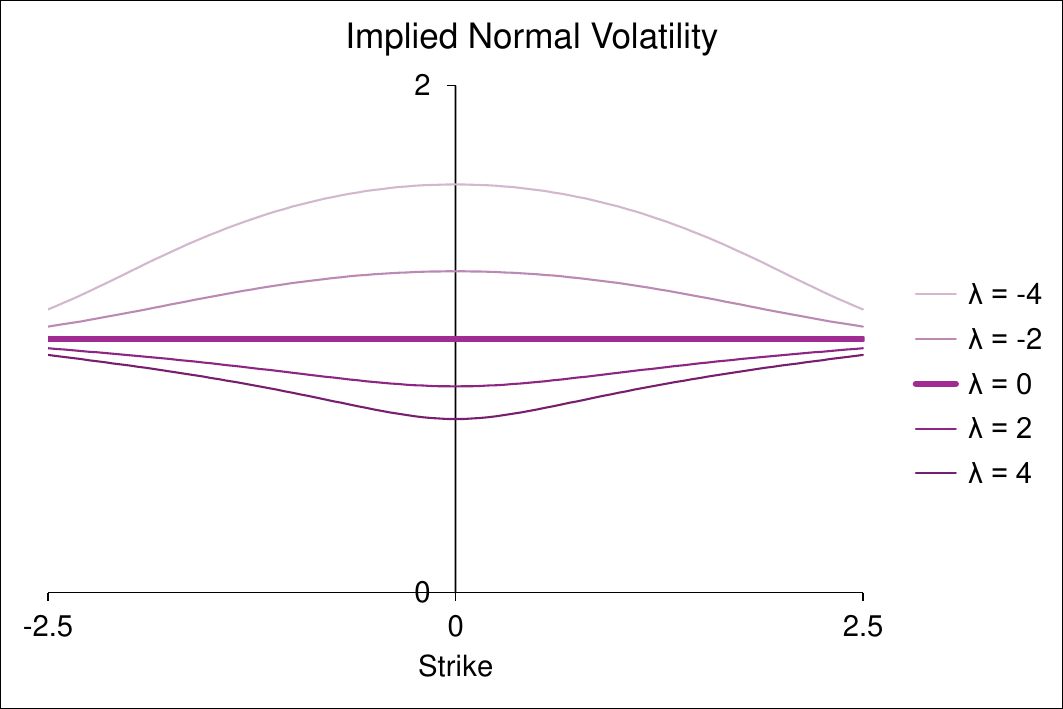}&
\includegraphics[width=0.33\textwidth-\tabcolsep]{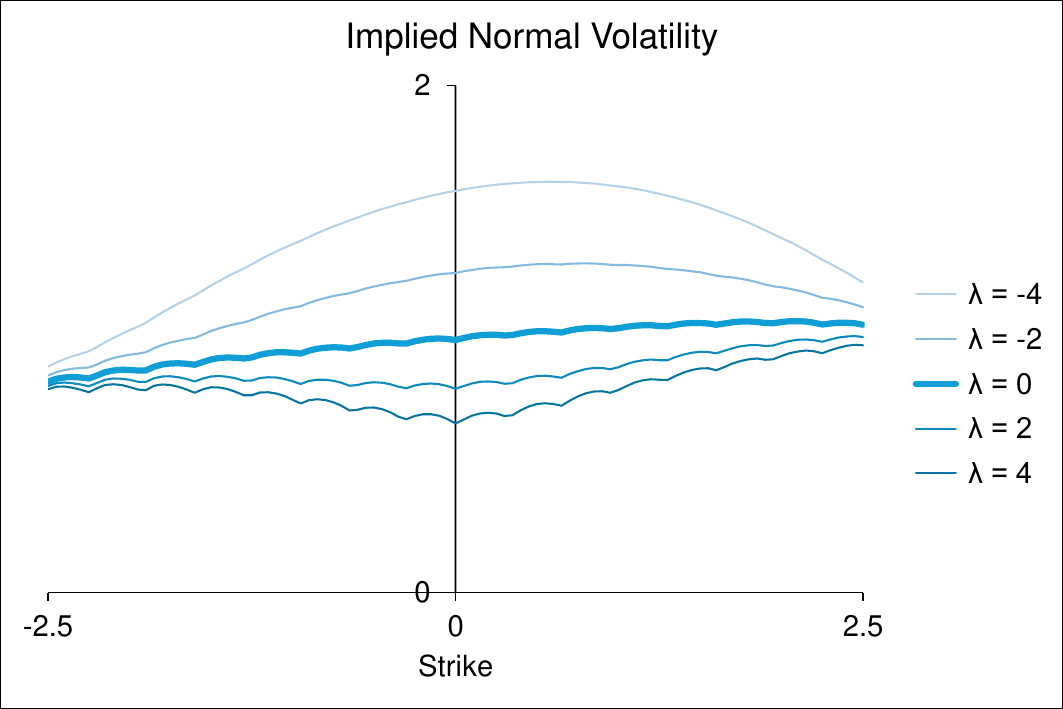}&
\includegraphics[width=0.33\textwidth-\tabcolsep]{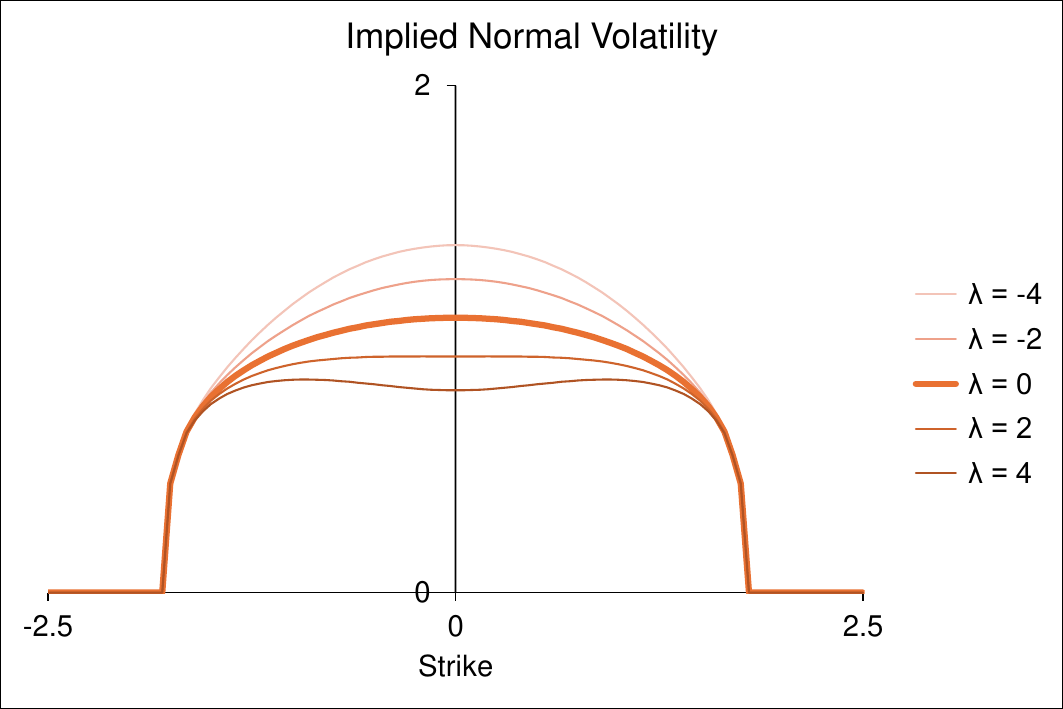}\\
Call Option (Normal)&Call Option (Poisson)&Call Option (Uniform)\\
\includegraphics[width=0.33\textwidth-\tabcolsep]{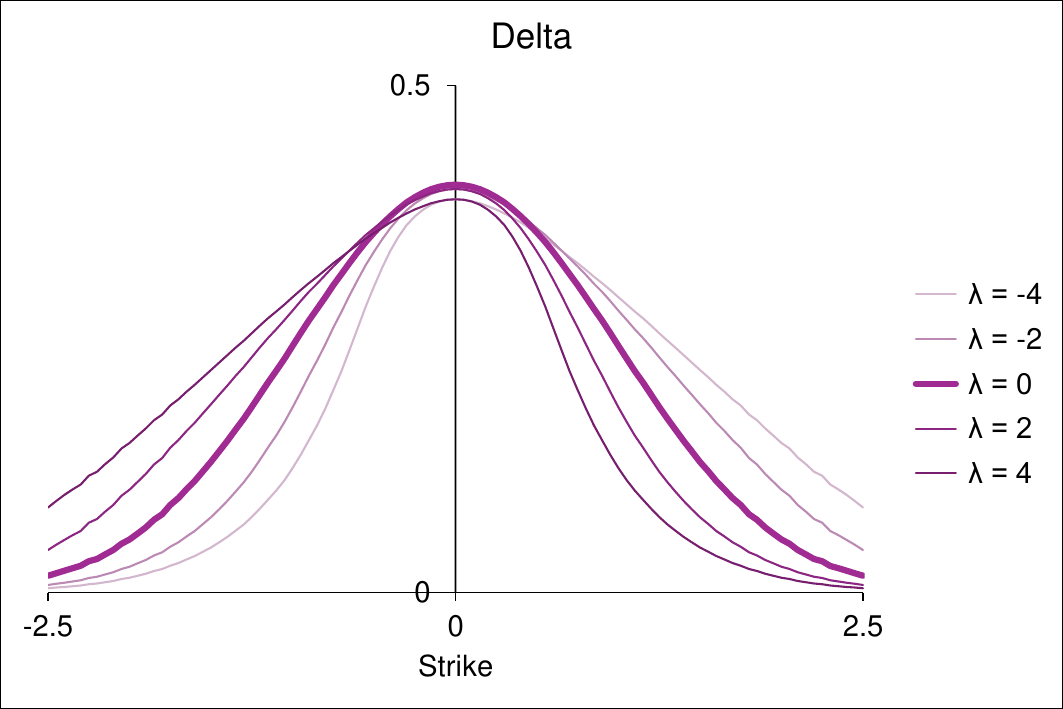}&
\includegraphics[width=0.33\textwidth-\tabcolsep]{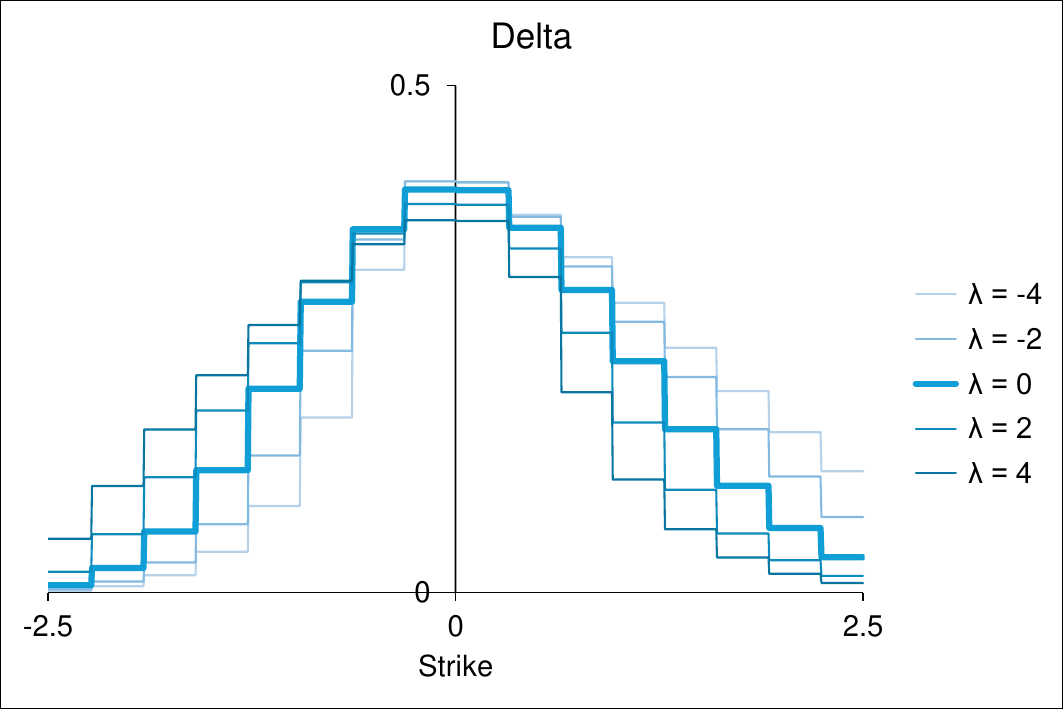}&
\includegraphics[width=0.33\textwidth-\tabcolsep]{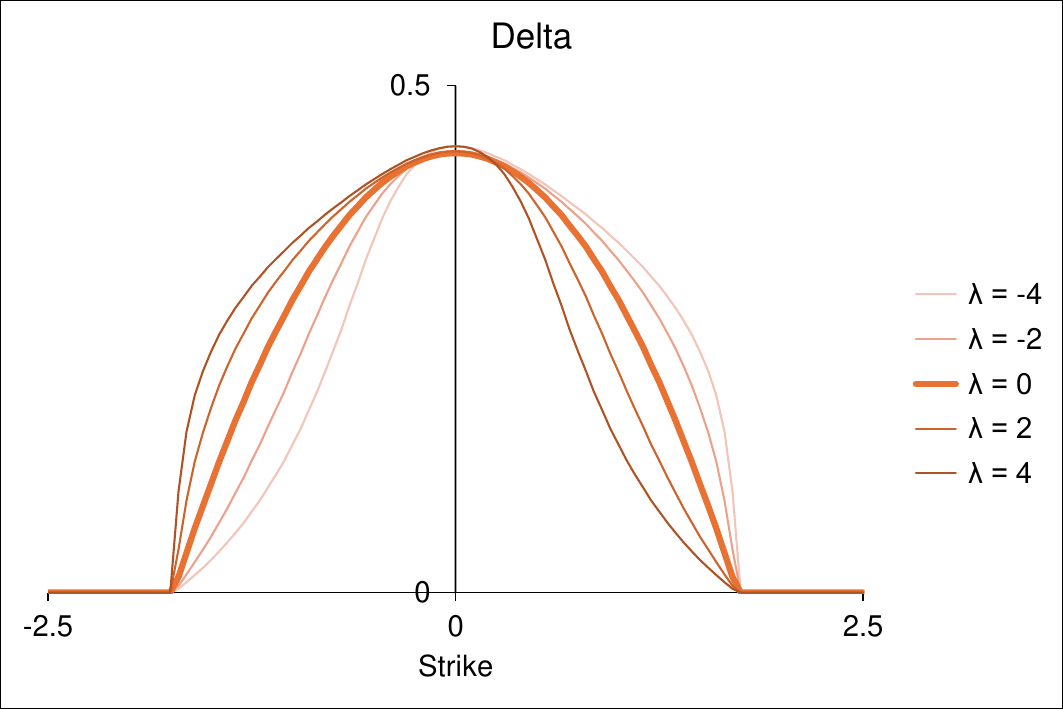}\\
\includegraphics[width=0.33\textwidth-\tabcolsep]{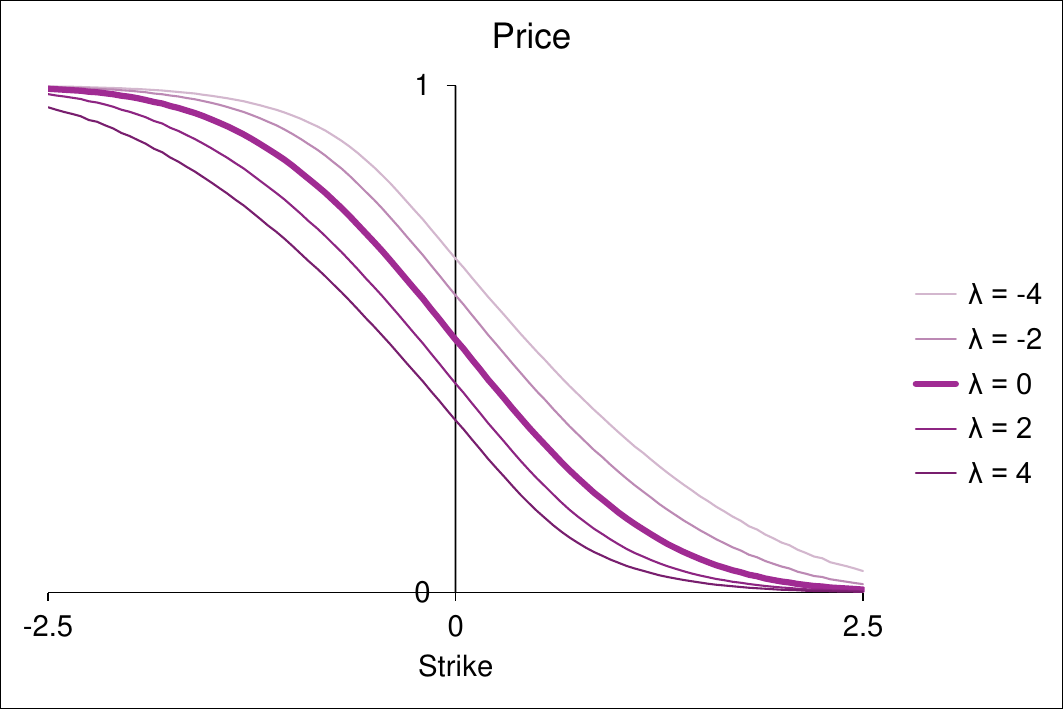}&
\includegraphics[width=0.33\textwidth-\tabcolsep]{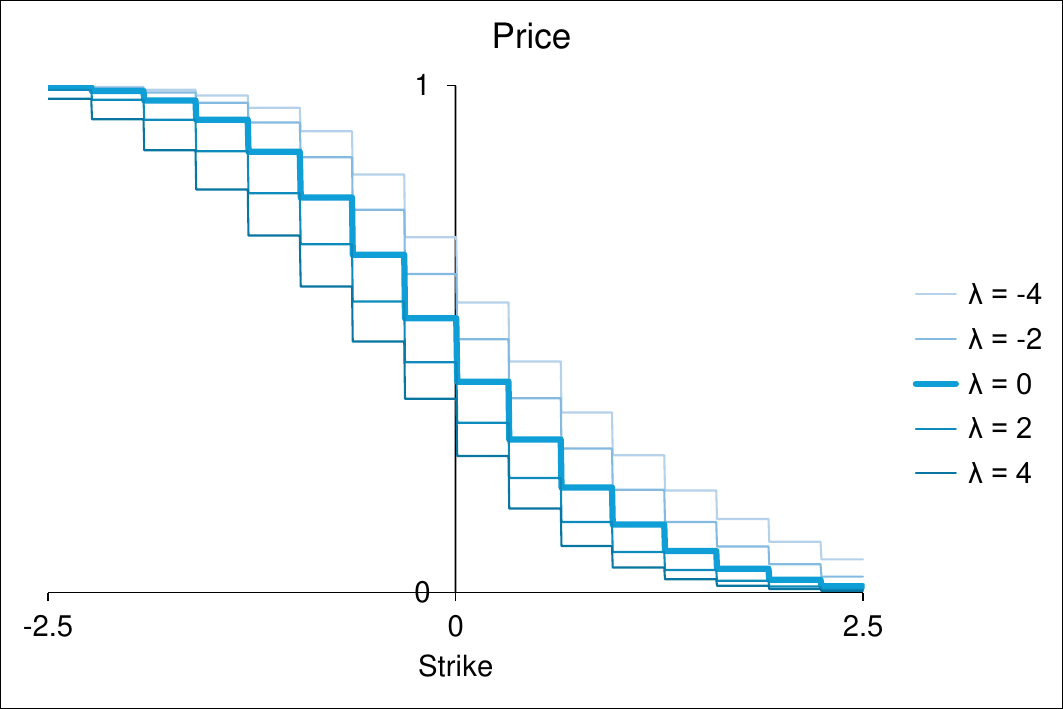}&
\includegraphics[width=0.33\textwidth-\tabcolsep]{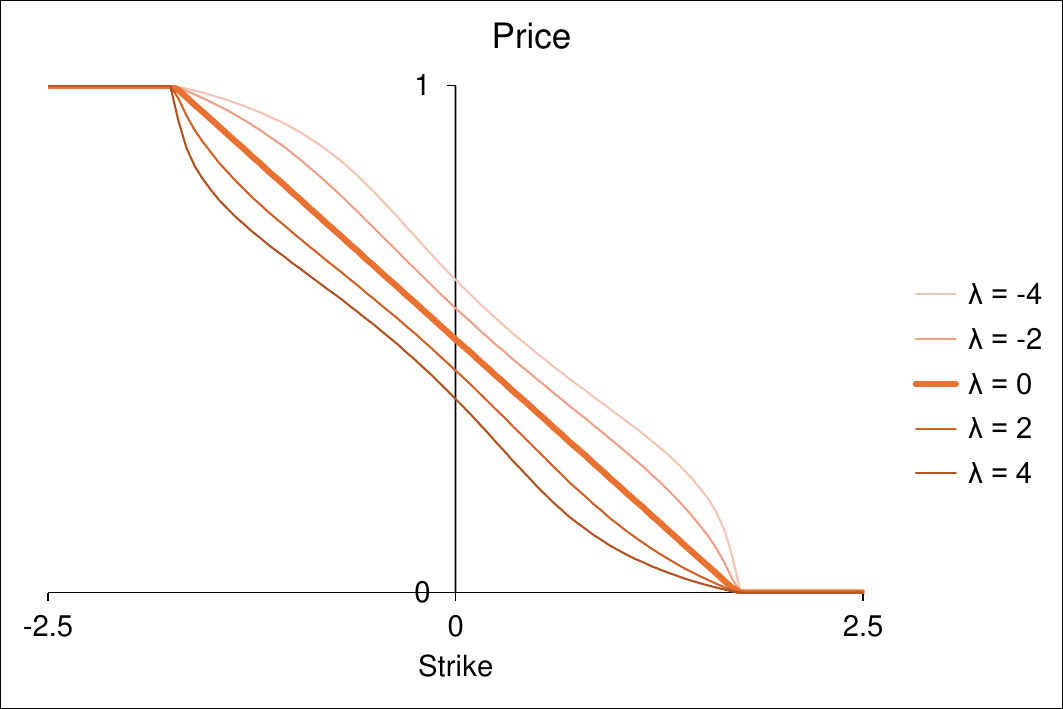}\\
Digital Option (Normal)&Digital Option (Poisson)&Digital Option (Uniform)
\end{tabular}
\caption{Three studies of entropic hedging and pricing for a call option (top three rows) with final price $P=(Q-k)^+$ and a digital option (bottom two rows) with final price $P=1_{Q>k}$ at strike $k$ on an underlying security with final price $Q$ modelled as a normal variable, a Poisson variable with ten expected jumps, and a uniform variable, all standardised by linear transformation to have zero mean and unit standard deviation. Risk aversion is fixed at one, $\alpha=1$, and the implied funding rate is set to zero, $r=0$. The graphs plot the delta and price per unit notional as a function of the strike for a range of notionals $\lambda$, highlighting the central plot for the mid-pricing case $\lambda=0$. Residual risk after hedging creates a bid-offer around the mid price and skews the optimal hedge ratio for positions with large notional.}
\label{fig:entropicpricing}
\end{figure*}

\subsection{Limit order book}

Using the price measure, the hedge and price conditions are rewritten as:
\begin{align}
q(r+\alpha s)\,dt&=\bfrac{\Exp^\phi[dq\,\exp[-\alpha(P-\delta\cdot dq)]]}{\Exp^\phi[\exp[-\alpha(P-\delta\cdot dq)]]} \\
p&=\Exp^\phi[\alpha][P-\delta\cdot dq] \notag
\end{align}
These conditions express a non-linear relationship between the initial and final derivative prices. If the derivative security with notional $\lambda$ has final price $\lambda P$, then the corresponding initial price $p$, hedge portfolio $\delta$ and implied funding spread $s$, quoted per unit notional of the derivative security, are determined from the hedge and price conditions:
\begin{align}
q(r+\lambda\alpha s)\,dt&=\bfrac{\Exp^\phi[dq\,\exp[-\lambda\alpha(P-\delta\cdot dq)]]}{\Exp^\phi[\exp[-\lambda\alpha(P-\delta\cdot dq)]]} \\
p&=\Exp^\phi[\lambda\alpha][P-\delta\cdot dq] \notag
\end{align}
and the self-funding condition. Scaling the derivative notional is thus equivalent to scaling the risk aversion. This is the scale invariance property of entropic pricing.

When $\lambda\alpha$ is small, the hedge and price conditions are expanded to discover the solutions:
\begin{align}
s&=\bfrac{\Exp^\phi[P]/(1+r\,dt)-\Var^\phi[dq]^{-1}\Var^\phi[dq,P]\cdot q}{\Var^\phi[dq]^{-1}q\cdot q\,dt}+O[\lambda\alpha] \notag \\
\delta&=\Var^\phi[dq]^{-1}(\Var^\phi[dq,P]+qs\,dt)+O[\lambda\alpha] \\
p&=\frac{1}{1+r\,dt}\left(\Exp^\phi[P]-\frac{1}{2}\alpha\Var^\phi[P-\delta\cdot dq]\right)+O[(\lambda\alpha)^2] \notag
\end{align}
This expansion in terms of $\lambda\alpha$ approximates the price when the risk appetite of the investor is large or when the derivative security has small notional. In the marginal market, the model thus implies a limit order book for the derivative security with mid price and bid-offer:
\begin{align}
p_0&=\frac{1}{1+r\,dt}\Exp^\phi[P] \\
p_\pm&=\frac{\alpha}{1+r\,dt}\Var^\phi[P-\delta\cdot dq] \notag
\end{align}
The discounted mid price of the derivative security is a martingale in the price measure and the bid-offer prices incompleteness of the underlying market. Higher-order adjustments then accommodate the extended features of the price distributions.

This model is extensive in scope, consistently encompassing arbitrary economic assumptions, predictable and unpredictable funding, complete and incomplete hedge markets, continuous and discontinuous evolution, continuous and discrete hedging, and accounting for the risk appetite of the investor. The studies in \cref{fig:entropicpricing} demonstrate entropic hedging and pricing for call and digital options over a single investment horizon, and show the bid-offer around the mid model that prices the residual risk after hedging.

\section{Market and model risk}

Derivative price models support trading activity, and their risks manifest as poor decision making in the trade lifecycle. The model must provide entry and exit prices that have a reasonable chance of transacting, and a funding and hedging strategy that has a reasonable chance of retaining the initial premium for the trade. The market and model risks in these objectives are distinguished.
\begin{description}[leftmargin=0\parindent]
\item[Market risk]An optimal funding and hedging strategy is derived on the assumption of the expectation measure. Residual profit and loss is generated if this strategy is not faithfully executed.
\item[Model risk]If the expectation measure does not accurately capture realised market dynamics over the trade lifecycle, then the strategy is suboptimal and will also generate residual profit and loss.
\end{description}
Analysis of these risks considers two expectation measures: the before measure $\Exp$ assumed initially as the basis for decision making and used to quantify the predicted performance of the strategy; and the after measure $\bar{\Exp}$ fitted to realised market dynamics and used to quantify the actual performance of the strategy.

Starting with the before expectation measure $\Exp$, the calibration conditions:
\begin{align}
0&=\phi\cdot q \\
0&=\Exp^\phi[dq]-qr\,dt \notag
\end{align}
are solved to generate the implied funding rate $r$ and the unit optimal portfolio $\phi$ for the underlying securities, and the hedge and price conditions:
\begin{align}
p&=\delta\cdot q \\
0&=\Exp^\phi[(dq-q(r+\alpha s)\,dt)\exp[-\alpha(P-\delta\cdot dq)]] \notag \\
p&=\Exp^\phi[\alpha][P-\delta\cdot dq] \notag
\end{align}
are solved to generate the implied funding spread $s$, the hedge portfolio $\delta$ and the initial price $p$ for the derivative security. While the hedge neutralises the underlying component of the derivative risk, when the market is incomplete there are unhedged risks that are sensitive to the choice of expectation measure. This is the origin of model risk.

Suppose that the after expectation measure $\bar{\Exp}$ is equivalent to the before expectation measure $\Exp$ with Radon-Nikodym weight:
\begin{equation}
\bfrac{d\bar{\Exp}}{d\Exp}=\bfrac{\exp[\epsilon W]}{\Exp[\exp[\epsilon W]]}=1+\epsilon(W-\Exp[W])+O[\epsilon^2]
\end{equation}
for a bounded weight $W$ and a scale parameter $\epsilon$. With the benefit of hindsight, the implied funding rate and unit optimal portfolio are adjusted, $\bar{r}=r+\epsilon\,\partial r$ and $\bar{\phi}=\phi+\epsilon\,\partial\phi$, with:
\begin{align}
\partial r&=\bfrac{\Var^\phi[dq]^{-1}\Var^\phi[dq,W]\cdot q}{\Var^\phi[dq]^{-1}q\cdot q\,dt}+O[\epsilon] \\
\partial\phi&=\Var^\phi[dq]^{-1}(\Var^\phi[dq,W]-q\,\partial r\,dt)+O[\epsilon] \notag
\end{align}
Applying these adjustments, the before and after price measures are related by:
\begin{align}
&\bar{\Exp}^{\bar{\phi}}[X]-\Exp^\phi[X]= \\
&\qquad\epsilon\Var^\phi[W-\partial\phi\cdot dq,X]+O[\epsilon^2] \notag \\
&\bar{\Exp}^{\bar{\phi}}[\alpha][X]-\Exp^\phi[\alpha][X]= \notag \\
&\qquad\epsilon\Var^\phi\!\left[W-\partial\phi\cdot dq,-\frac{1}{\alpha}\bfrac{\exp[-\alpha X]}{\Exp^\phi[\exp[-\alpha X]]}\right]+O[\epsilon^2] \notag
\end{align}
for the indeterminate variable $X$.

Similar adjustments apply to the implied funding spread, hedge portfolio and initial price, derived from the hedge and price conditions in the after expectation measure:
\begin{align}
\bar{p}&=\bar{\delta}\cdot q \\
0&=\bar{\Exp}^{\bar{\phi}}[(dq-q(\bar{r}+\alpha\bar{s})\,dt)\exp[-\alpha(P-\bar{\delta}\cdot dq)]] \notag \\
\bar{p}&=\bar{\Exp}^{\bar{\phi}}[\alpha][P-\bar{\delta}\cdot dq] \notag
\end{align}
Focusing on the price condition, the adjustment to the initial price is $\bar{p}=p+\epsilon\,\partial p$ with:
\begin{align}
&\partial p=\frac{1}{1+(r+\alpha s)\,dt}\times \\
&\quad\Var^\phi\!\left[W-\partial\phi\cdot dq,-\frac{1}{\alpha}\bfrac{\exp[-\alpha(P-\delta\cdot dq)]}{\Exp^\phi[\exp[-\alpha(P-\delta\cdot dq)]]}\right]+O[\epsilon] \notag
\end{align}
The hedge portfolio is the metric of market risk in the entropic framework, and this adjustment is the metric of model risk as it measures the sensitivity of the initial price to the choice of expectation measure. If the underlying market is complete, the variance of the hedged return is zero and the sensitivity to the expectation measure is eliminated. More generally, model risk increases in proportion with the degree of market incompleteness.

\section{The Black-Scholes model}

As an illustration of the entropic method, consider an underlying market comprising a funding security with price $b>0$ and an underlying security with price $q$ satisfying the stochastic differential equations:
\begin{align}
db&=br\,dt \\
dq&=q(\mu\,dt+\sqrt{\nu}\,dy) \notag
\end{align}
where the increment $dy$ is Brownian in the expectation measure. The model is parametrised by the funding rate $r$ of the funding security and the mean rate $\mu$ and variance rate $\nu>0$ of the underlying security.

The self-funded portfolio in the underlying market comprises $\omega$ units of the underlying security funded by $-\omega(q/b)$ units of the funding security. The portfolio return has entropy-adjusted mean:
\begin{align}
\frac{1}{dt}\Exp[&\alpha][\omega(dq-(q/b)db)]= \\
&\frac{1}{2}\frac{(\mu-r)^2}{\alpha\nu}-\frac{1}{2}\alpha q^2\nu\left(\omega-\frac{\mu-r}{\alpha q\nu}\right)^{\!2}+O[dt] \notag
\end{align}
This is maximised by the optimal portfolio:
\begin{equation}
\omega=\frac{\mu-r}{\alpha q\nu}
\end{equation}

When the derivative security is added, the self-funded portfolio comprises one unit of the derivative security and $(\omega-\delta)$ units of the underlying security funded by $-(p/b)-(\omega-\delta)(q/b)$ units of the funding security. The underlying market is complete when the model parameters are fixed but incomplete when the model parameters are volatile. In the first scenario, risk can be neutralised by hedging and the derivative price is determined from the no-arbitrage condition. In the second scenario, risk cannot be neutralised by hedging and the no-arbitrage condition does not determine the derivative price. By replacing risk neutralisation with risk optimisation as the objective of hedging, the entropic method extends to this case. The derivative price is then resolved by replacing no-arbitrage with indifference as the equilibrium condition for pricing.

\begin{description}[leftmargin=0\parindent]
\item[Complete market]Assume the funding rate $r$, mean rate $\mu$ and variance rate $\nu$ are fixed. The portfolio return has entropy-adjusted mean:
\begin{align}
\frac{1}{dt}\Exp[&\alpha][(dp-(p/b)db)+(\omega-\delta)(dq-(q/b)db)]= \\
&\frac{\partial p}{\partial t}-pr+\frac{\partial p}{\partial q}qr+\frac{1}{2}\frac{\partial^2p}{\partial q^2}q^2\nu+\frac{1}{2}\frac{(\mu-r)^2}{\alpha\nu} \notag \\
&-\frac{1}{2}\alpha q^2\nu\left(\left(\omega-\frac{\mu-r}{\alpha q\nu}\right)-\left(\delta-\frac{\partial p}{\partial q}\right)\right)^{\!2}+O[dt] \notag
\end{align}
when the derivative security with price $p[t,q]$ is added. This is maximised by the hedge portfolio:
\begin{equation}
\delta=\frac{\partial p}{\partial q}
\end{equation}
The indifference condition dictates that the addition of the derivative security neither increases nor decreases the maximum entropy-adjusted mean of the portfolio, leading to the equation:
\begin{equation}
0=\frac{\partial p}{\partial t}-(p-\delta q)r+\frac{1}{2}\frac{\partial^2p}{\partial q^2}q^2\nu
\end{equation}
for the derivative price. The expression on the right is the sum of time decay, funding and gamma contributions to the profit and loss of the hedged derivative. Setting this to zero eliminates the possibility of arbitrage.
\item[Incomplete market]Assume the funding rate $r$ and mean rate $\mu$ are fixed and the variance rate $\nu$ satisfies the stochastic differential equation:
\begin{equation}
d\nu=\gamma[t,q,\nu]\,dt+\epsilon[t,q,\nu]\sqrt{\nu}\,dz
\end{equation}
where the increments $dy$ and $dz$ are joint Brownian with correlation $\rho[t,q,\nu]$ in the expectation measure. The portfolio return has entropy-adjusted mean:
\begin{align}
\frac{1}{dt}\Exp[&\alpha][(dp-(p/b)db)+(\omega-\delta)(dq-(q/b)db)]= \\
&\frac{\partial p}{\partial t}-pr+\frac{\partial p}{\partial q}qr+\frac{\partial p}{\partial\nu}(\gamma-(\mu-r)\rho\epsilon) \notag \\
&+\frac{1}{2}\frac{\partial^2p}{\partial q^2}q^2\nu+\frac{\partial^2p}{\partial q\partial\nu}q\nu\rho\epsilon+\frac{1}{2}\frac{\partial^2p}{\partial\nu^2}\nu\epsilon^2 \notag \\
&-\frac{1}{2}\alpha\left(\frac{\partial p}{\partial\nu}\right)^{\!2}\!\nu(1-\rho^2)\epsilon^2+\frac{1}{2}\frac{(\mu-r)^2}{\alpha\nu} \notag \\
&-\frac{1}{2}\alpha q^2\nu\left(\left(\omega-\frac{\mu-r}{\alpha q\nu}\right)-\left(\delta-\frac{\partial p}{\partial q}-\frac{\partial p}{\partial\nu}\frac{\rho\epsilon}{q}\right)\right)^{\!2} \notag \\
&+O[dt] \notag
\end{align}
when the derivative security with price $p[t,q,\nu]$ is added. This is maximised by the hedge portfolio:
\begin{equation}
\delta=\frac{\partial p}{\partial q}+\frac{\partial p}{\partial\nu}\frac{\rho\epsilon}{q}
\end{equation}
The indifference condition dictates that the addition of the derivative security neither increases nor decreases the maximum entropy-adjusted mean of the portfolio, leading to the equation:
\begin{align}
0={}&\frac{\partial p}{\partial t}-(p-\delta q)r+\frac{\partial p}{\partial\nu}(\gamma-\mu\rho\epsilon) \\
&+\frac{1}{2}\frac{\partial^2p}{\partial q^2}q^2\nu+\frac{\partial^2p}{\partial q\partial\nu}q\nu\rho\epsilon+\frac{1}{2}\frac{\partial^2p}{\partial\nu^2}\nu\epsilon^2 \notag \\
&-\frac{1}{2}\alpha\left(\frac{\partial p}{\partial\nu}\right)^{\!2}\!\nu(1-\rho^2)\epsilon^2 \notag
\end{align}
for the derivative price.  The expression on the right is the sum of time decay, funding, vega, gamma and value-at-risk contributions in a conservative estimate of the profit and loss of the hedged derivative. Setting this to zero generates the fair price of the derivative security, including a reserve for the residual risk that remains after hedging.
\end{description}

\section{Markovian diffusion}

Generalising the example of the Black-Scholes model, in this section assume that there is a vector $x$ of state variables whose cumulant generating function in the expectation measure $\Exp$ satisfies:
\begin{align}
\frac{1}{dt}\log{}&\Exp\exp[\eta\cdot dx]= \\
&\eta\cdot\mu+\frac{1}{2}\eta\cdot\nu\eta+\int_j(\exp[\eta\cdot j]-1)\,\theta[dj]+O[dt] \notag
\end{align}
where the vector $\mu\equiv\mu[t,x]$ and positive-definite matrix $\nu\equiv\nu[t,x]$ are the mean and covariance rate for the continuous component of the diffusion, and the positive measure $\theta[dj]\equiv\theta[t,x][dj]$ is the frequency of jumps in the discontinuous component of the diffusion. The underlying price $q\equiv q[t,x]$ and derivative price $p\equiv p[t,x]$ are then assumed to be functions of time $t$ and state $x$.

This L\'evy-Khintchine form for the cumulant generating function of the state implies its distribution is divisible, which enables the transition to the limit $dt\to0$ of continuous hedging. Repeated differentiation with respect to $\eta$ derives It\^o's lemma for the function $f\equiv f[t,x]$ of time $t$ and state $x$ in the expectation measure $\Exp$:
\begin{align}
\frac{1}{dt}\Exp[df]={}&\frac{\partial f}{\partial t}+\frac{\partial f}{\partial x}\cdot\mu+\frac{1}{2}\frac{\partial^2f}{\partial x^2}\cdot{}\!\cdot\nu \\
&+\int_j\partial_jf\,\theta[dj]+O[dt] \notag
\end{align}
where $\partial_jf[t,x]:=f[t,x+j]-f[t,x]$. In the price measure $\Exp^\phi$, It\^o's lemma is modified by the drift adjustments of Girsanov's theorem:
\begin{align}
\frac{1}{dt}\Exp^\phi[df]={}&\frac{\partial f}{\partial t}+\frac{\partial f}{\partial x}\cdot\mu^\phi+\frac{1}{2}\frac{\partial^2f}{\partial x^2}\cdot{}\!\cdot\nu \\
&+\int_j\partial_jf\,\theta^\phi[dj]+O[dt] \notag
\end{align}
where the mean rate and frequency of jumps are adjusted:
\begin{align}
\mu^\phi&:=\mu-\nu\left(\phi\cdot\frac{\partial q}{\partial x}\right) \\
\theta^\phi[dj]&:=\exp[-\phi\cdot\partial_jq]\,\theta[dj] \notag
\end{align}

With this assumption, the price model is established on the divisible process for the state vector $x$ and the underlying price vector $q[t,x]$ given as a function of time and state. Using the expression for It\^o's lemma in the price measure, the calibration, hedge and price conditions are transformed into partial differential-integral form. The calibration condition becomes:
\begin{align}
0={}&\frac{\partial q}{\partial t}-qr+\frac{\partial q}{\partial x}\cdot\mu^\phi+\frac{1}{2}\frac{\partial^2q}{\partial x^2}\cdot{}\!\cdot\nu \\
&+\int_j\partial_jq\,\theta^\phi[dj]+O[dt] \notag
\end{align}
which is solved simultaneously with the self-funding condition $\phi\cdot q=0$ for the implied funding rate $r$ and the unit optimal portfolio $\phi$. The hedge and price conditions become:
\begin{align}
0={}&\frac{\partial p}{\partial x}\cdot\nu\frac{\partial q}{\partial x}+qs-\left(\frac{\partial q}{\partial x}\cdot\nu\frac{\partial q}{\partial x}\right)\delta \notag \\
&\!\!\!\!+\int_j\frac{1}{\alpha}(1-\exp[-\alpha(\partial_jp-\delta\cdot\partial_jq)])\,\partial_jq\,\theta^\phi[dj]+O[dt] \notag \\
0={}&\frac{\partial p}{\partial t}-pr+\frac{\partial p}{\partial x}\cdot\mu^\phi+\frac{1}{2}\frac{\partial^2p}{\partial x^2}\cdot{}\!\cdot\nu \\
&-\frac{1}{2}\alpha\left(\frac{\partial p}{\partial x}-\delta\cdot\frac{\partial q}{\partial x}\right)\cdot\nu\left(\frac{\partial p}{\partial x}-\delta\cdot\frac{\partial q}{\partial x}\right) \notag \\
&+\int_j\frac{1}{\alpha}(1-\exp[-\alpha(\partial_jp-\delta\cdot\partial_jq)])\,\theta^\phi[dj]+O[dt] \notag
\end{align}
which are solved simultaneously with the self-funding condition $\delta\cdot q=p$ for the implied funding spread $s$, hedge portfolio $\delta$ and derivative price $p$. Incompleteness means that the hedge portfolio does not exactly offset the risks of the derivative security for all possible market movements over the hedging interval, and this residual risk is accounted for with a compensating adjustment to the derivative price.

Starting with its terminal contractual settlement, the numerical solution of these conditions for the derivative price repeats a Newton-Raphson scheme that determines the initial price $p$ from the final price $p+dp$ over each interval in a discretised schedule. The high dimensional complexity of this procedure is a consequence of the non-linearity of the hedge portfolio in the hedge condition. There are two useful scenarios where this condition is linear: the mid-pricing case $\alpha=0$; and the continuous diffusion case $\theta[dj]=0$. Pricing simplifies in these cases.

\begin{description}[leftmargin=0\parindent]
\item[Mid-pricing]When $\alpha=0$, the hedge and price conditions simplify:
\begin{align}
s={}&\bfrac{p-V^{-1}\left(\frac{\partial p}{\partial x}\cdot\nu\frac{\partial q}{\partial x}+\int_j\partial_jp\,\partial_jq\,\theta^\phi[dj]\right)\cdot q}{V^{-1}q\cdot q}+O[dt] \notag \\
\delta={}&V^{-1}\left(\frac{\partial p}{\partial x}\cdot\nu\frac{\partial q}{\partial x}+\int_j\partial_jp\,\partial_jq\,\theta^\phi[dj]+qs\right)+O[dt] \notag \\
0={}&\frac{\partial p}{\partial t}-pr+\frac{\partial p}{\partial x}\cdot\mu^\phi+\frac{1}{2}\frac{\partial^2p}{\partial x^2}\cdot{}\!\cdot\nu \\
&+\int_j(\partial_jp-\delta\cdot\partial_jq)\,\theta^\phi[dj]+O[dt] \notag
\end{align}
where:
\begin{align}
V&:=\frac{\partial q}{\partial x}\cdot\nu\frac{\partial q}{\partial x}+\int_j\partial_jq\,\partial_jq\,\theta^\phi[dj]
\end{align}
Embedding the expressions for the implied funding spread and hedge portfolio, the price condition becomes a linear parabolic differential-integral equation for the derivative price. This equation balances the contributions to the return from the continuous and discontinuous components of the diffusion.
\item[Continuous diffusion]When $\theta[dj]=0$, the hedge and price conditions simplify:
\begin{align}
s={}&\bfrac{p-V^{-1}\left(\frac{\partial p}{\partial x}\cdot\nu\frac{\partial q}{\partial x}\right)\cdot q}{V^{-1}q\cdot q}+O[dt] \notag \\
\delta={}&V^{-1}\left(\frac{\partial p}{\partial x}\cdot\nu\frac{\partial q}{\partial x}+qs\right)+O[dt] \notag \\
0={}&\frac{\partial p}{\partial t}-pr+\frac{\partial p}{\partial x}\cdot\mu^\phi+\frac{1}{2}\frac{\partial^2p}{\partial x^2}\cdot{}\!\cdot\nu \\
&-\frac{1}{2}\alpha\left(\frac{\partial p}{\partial x}-\delta\cdot\frac{\partial q}{\partial x}\right)\cdot\nu\left(\frac{\partial p}{\partial x}-\delta\cdot\frac{\partial q}{\partial x}\right)+O[dt] \notag
\end{align}
where:
\begin{align}
V&:=\frac{\partial q}{\partial x}\cdot\nu\frac{\partial q}{\partial x}
\end{align}
Embedding the expressions for the implied funding spread and hedge portfolio, the price condition becomes a non-linear parabolic differential equation for the derivative price. This equation modifies the expression for the price to adjust for residual risk.
\end{description}

As a concrete example of a continuous Markovian diffusion, the Heston model illustrates the key features of the more general framework. In this model, the state variables are the funding price $b$, the underlying price $q$ and the variance rate $\nu$ satisfying the stochastic differential equations:
\begin{align}
db&=br\,dt \\
dq&=q(\mu\,dt+\sqrt{\nu}\,dy) \notag \\
d\nu&=\kappa(\bar{\nu}-\nu)\,dt+\epsilon\sqrt{\nu}\,dz \notag
\end{align}
where $dy$ and $dz$ are joint-Brownian increments in the expectation measure with correlation $\rho$. This is equivalent to the cumulant generating function:
\begin{align}
\frac{1}{dt}\log{}&\Exp\exp[\beta\,db+\psi\,dq+\eta\,d\nu]= \\
&\beta\,br+\psi\,q\mu+\eta\,\kappa(\bar{\nu}-\nu) \notag \\
&+\frac{1}{2}\psi^2\,q^2\nu+\psi\eta\,q\nu\rho\epsilon+\frac{1}{2}\eta^2\,\nu\epsilon^2+O[dt] \notag
\end{align}
The funding price has predictable increment with mean $br\,dt$ where $r$ is the funding rate. The underlying price has increment with mean $q\mu\,dt$ and variance $q^2\nu\,dt$ where $\mu$ is the mean rate. The variance rate has increment with mean $\kappa(\bar{\nu}-\nu)\,dt$ and variance $\nu\epsilon^2\,dt$ where $\kappa$ is the mean reversion rate, $\bar{\nu}$ is the mean reversion level and $\epsilon$ is the volatility of variance. The final parameter $\rho$ is the correlation between the increments of the underlying price and the variance rate.

Since the funding security has predictable return, the implied funding rate is $r$ and the implied funding spread is $s=0$. The self-funded portfolio with underlying weight $\phi$ and funding weight $-\phi(q/b)$ creates the price measure $\Exp^\phi$ with cumulant generating function:
\begin{align}
\frac{1}{dt}\log{}&\Exp^\phi\exp[\beta\,db+\psi\,dq+\eta\,d\nu]= \\
&\beta\,br+\psi\,q\mu^\phi+\eta\,\kappa(\bar{\nu}^\phi-\nu) \notag \\
&+\frac{1}{2}\psi^2\,q^2\nu+\psi\eta\,q\nu\rho\epsilon+\frac{1}{2}\eta^2\,\nu\epsilon^2+O[dt] \notag
\end{align}
where the mean rate of the underlying price and the mean reversion level of the variance rate are adjusted:
\begin{align}
\mu^\phi&:=\mu-\phi q\nu \\
\bar{\nu}^\phi&:=\bar{\nu}-\phi q\nu\frac{\rho\epsilon}{\kappa} \notag
\end{align}
The calibration condition then implies that the unit optimal portfolio has $\phi=(\mu-r)/(q\nu)$, so that:
\begin{align}
\mu^\phi&=r \\
\bar{\nu}^\phi&=\bar{\nu}-(\mu-r)\frac{\rho\epsilon}{\kappa} \notag
\end{align}
The state is also a Heston model in the price measure, albeit with these two drift parameters adjusted to account for the change of measure. In this approach, the Heston parameters in the price measure can be estimated from statistical analysis of the state in the expectation measure, applying the measure adjustment post-estimation to calibrate to the underlying price.

The hedge and price conditions in the Heston model generate the expressions:
\begin{align}
\delta={}&\frac{\partial p}{\partial q}+\frac{\partial p}{\partial\nu}\frac{\rho\epsilon}{q}+O[dt] \\
0={}&\frac{\partial p}{\partial t}-pr+\frac{\partial p}{\partial q}qr+\frac{\partial p}{\partial\nu}(\kappa(\bar{\nu}-\nu)-(\mu-r)\rho\epsilon) \notag \\
&+\frac{1}{2}\frac{\partial^2p}{\partial q^2}\,q^2\nu+\frac{\partial^2p}{\partial q\,\partial\nu}\,q\nu\rho\epsilon+\frac{1}{2}\frac{\partial^2p}{\partial\nu^2}\,\nu\epsilon^2 \notag \\
&-\frac{1}{2}\alpha\left(\frac{\partial p}{\partial\nu}\right)^{\!2}\!\nu(1-\rho^2)\epsilon^2+O[dt] \notag
\end{align}
for the hedge portfolio $\delta[t,q,\nu]$ and the derivative price $p[t,q,\nu]$. In the first expression, the hedge ratio is adjusted when the increment of the variance rate is correlated with the increment of the underlying price. The derivative price then satisfies a linear parabolic differential equation with a non-linear correction term that prices the unhedgeable risks when the increments are not perfectly correlated.

\section{Discrete models}

As the example of the Heston model demonstrates, entropic risk optimisation can be consistently applied in theoretical models that include unhedgeable factors, naturally adjusting for market incompleteness in pricing and hedging. Data-driven models take a different approach, eschewing theoretical foundations for a statistical method that extracts information directly and exclusively from the data. In practice, this means that the model is supported on a finite set of eigenstates generated from a statistical model trained on historical data. The expectation measure is then assumed to be uniformly distributed on the eigenstates. With no theory to support development beyond the information expressed in the data, a more general foundation is needed that does not depend on unrealistic assumptions such as complete markets, continuous hedging or risk-free funding. Advocates of deep hedging address this by replacing risk neutralisation with risk optimisation, using convex risk metrics to assess strategic performance. While they suffer from lack of explainability, data-driven models can uncover complex relationships between economic variables that may be hard to capture theoretically.

Entropy is a universal concept in information theory that is ideally suited to a wide range of applications. In this section, the entropic risk metric is applied to risk optimisation in discrete models of classical and quantum information. The core premise of the discrete model is that the indeterminate variable has a spectrum of possible eigenvalues, which may include duplicates, that are uniformly distributed in the expectation measure. Classical or quantum models are then characterised by the existence or absence of a common basis of eigenstates supporting the eigenvalues of all the variables.

The indeterminate variable $X$ is associated with its spectrum $(X_1,\ldots,X_n)$ and the expectation measure assumes these outcomes are equally likely. Represent the variable as a diagonal matrix with its eigenvalues arrayed on the diagonal. The mean of $X$ is the average eigenvalue, expressed as the normalised trace of its matrix:
\begin{equation}
\Exp[X]:=\frac{1}{n}\tr[X]
\end{equation}
Significantly, this definition extends to arbitrary self-adjoint matrices, where the mean, defined to be the average eigenvalue, is equivalently evaluated as the normalised trace. Classical information is characterised by the commutativity of its indeterminate variables; without loss of generality, the variables are simultaneously diagonalised and the expectation measure is uniformly distributed on the common set of diagonalising eigenstates. This restriction is lifted for quantum information, where non-commutativity is embraced as an additional source of uncertainty with remarkable consequences for the quality of model fitting.

Functional calculus extends to self-adjoint matrices by applying the transformation directly to the eigenvalues in the diagonalised representation. The entropy-adjusted mean of $X$ is then defined:
\begin{equation}
\Exp[\alpha][X]=-\frac{1}{\alpha}\log\!\left[\frac{1}{n}\tr[\exp[-\alpha X]]\right]
\end{equation}
When optimising the strategy with this target function, care must be taken to manage the potential complications of non-commutativity. For example, the perturbation $\partial(X^2)=(\partial X)X+X(\partial X)$ cannot be simplified further when $X$ and $\partial X$ do not commute. Fortunately, the trace of this expression can be simplified, thanks to cyclicity, and in general the perturbation has the handy simplification:
\begin{equation}
\partial\tr[f[X]]=\tr[f'[X]\,\partial X]
\end{equation}
Using this result, the portfolios are optimised at the stationary points of the entropy-adjusted mean:
\begin{align}
0&=\frac{\partial}{\partial\omega}(\Exp[\alpha][\omega\cdot dq]-\omega\cdot qr\,dt) \\
&=\bfrac{\tr[dq\,\exp[-\alpha\omega\cdot dq]]}{\tr[\exp[-\alpha\omega\cdot dq]]}-qr\,dt \notag \\
0&=\frac{\partial}{\partial\delta}(\Exp[\alpha][dp+(\omega-\delta)\cdot dq]-(p-\delta\cdot q)(r+\alpha s)\,dt) \notag \\
&=-\bfrac{\tr[dq\,\exp[-\alpha(P+(\omega-\delta)\cdot dq)]]}{\tr[\exp[-\alpha(P+(\omega-\delta)\cdot dq)]]}+q(r+\alpha s)\,dt \notag
\end{align}
which generates the calibration and hedge conditions. The price condition for the initial derivative price is then obtained by equating the entropy-adjusted means. With the unit optimal portfolio defined by $\phi:=\alpha\omega$, these three conditions are expressed as:
\begin{align}
0&=\tr[(dq-qr\,dt)\exp[-\phi\cdot dq]] \\
0&=\tr[(dq-q(r+\alpha s)\,dt)\exp[-\phi\cdot dq-\alpha(P-\delta\cdot dq)]] \notag \\
p&=-\frac{1}{\alpha}\log\!\left[\frac{\tr[\exp[-\phi\cdot dq-\alpha(P-\delta\cdot dq)]]}{\tr[\exp[-\phi\cdot dq]]}\right] \notag
\end{align}
The calibration condition is solved together with the self-funding condition $\phi\cdot q=0$ for the implied funding rate $r$ and the unit optimal portfolio $\phi$, and the hedge and price conditions are solved together with the self-funding condition $\delta\cdot q=p$ for the implied funding spread $s$, the hedge portfolio $\delta$ and the initial derivative price $p$.

\subsection{Mid pricing}

There are, in general, no explicit solutions to the hedge and price conditions, though they are efficiently solved by the Newton-Raphson scheme. Mid pricing is derived in the limit $\alpha=0$, and for this special case the conditions are explicitly solved to derive the hedge portfolio and the initial derivative price. The solution depends on the small-$\alpha$ expansion of the conditions, which is complicated by the non-commutativity of the underlying returns. Fortunately, the expansion can be generated with the aid of Poincar\'e's formula for the differential of the matrix exponential.

Define the adjoint action $\ad_X[Y]:=XY-YX$ of the matrix $X$ on the matrix $Y$. Poincar\'e's formula is derived in the form:
\begin{align}
1-{}&\exp[-(X+\alpha Y)]\exp[X] \\
&=\int_{s=0}^1\frac{\partial}{\partial s}(1-\exp[-s(X+\alpha Y)]\exp[sX])\,ds \notag \\
&=\alpha\int_{s=0}^1\exp[-s(X+\alpha Y)]Y\exp[sX]\,ds \notag \\
&=\alpha\int_{s=0}^1\exp[-sX]Y\exp[sX]\,ds+O[\alpha^2] \notag \\
&=\alpha\left(\int_{s=0}^1\exp[-s\ad_X]\,ds\right)\![Y]+O[\alpha^2] \notag \\
&=\alpha\bfrac{1-\exp[-\ad_X]}{\ad_X}[Y]+O[\alpha^2] \notag
\end{align}
where the foundational relationship between Lie algebras and Lie groups:
\begin{equation}
\exp[\ad_X][Y]=\exp[X]Y\exp[-X]
\end{equation}
is applied at the fourth step. Combining this result with the elementary property $\tr[\ad_X[Y]Z]=-\tr[Y\ad_X[Z]]$ of the adjoint action derives the expansion:
\begin{align}
\tr[Z\exp&[-(X+\alpha Y)]]=\tr[Z\exp[-X]] \\
&-\alpha\tr\!\left[\bfrac{\exp[\ad_X]-1}{\ad_X}[Z]Y\exp[-X]\right]+O[\alpha^2] \notag
\end{align}

This result is used to solve the hedge and price conditions in the mid-pricing case. Define the price measure:
\begin{equation}
\Exp^\phi[X]:=\bfrac{\tr[X\exp[-\phi\cdot dq]]}{\tr[\exp[-\phi\cdot dq]]}
\end{equation}
Poincar\'e's formula derives the small-$\alpha$ expansions:
\begin{align}
&\frac{\tr[\exp[-\phi\cdot dq-\alpha(P-\delta\cdot dq)]]}{\tr[\exp[-\phi\cdot dq]]}= \\
&\qquad\qquad1-\alpha\Exp^\phi[P-\delta\cdot dq]+O[\alpha^2] \notag \\
&\frac{\tr[dq\,\exp[-\phi\cdot dq-\alpha(P-\delta\cdot dq)]]}{\tr[\exp[-\phi\cdot dq]]}= \notag \\
&\qquad\qquad qr\,dt-\alpha\Exp^\phi[d\hat{q}\,(P-\delta\cdot dq)]+O[\alpha^2] \notag
\end{align}
adjusting to account for non-commutativity in the underlying returns with the definition:
\begin{equation}
d\hat{q}:=\bfrac{\exp[\ad_{\phi\cdot dq}]-1}{\ad_{\phi\cdot dq}}[dq]
\end{equation}
The expansions generate the following solutions for the hedge and price conditions:
\begin{align}
p&=\frac{1}{1+r\,dt}\Exp^\phi[P]+O[\alpha] \\
s\,dt&=\bfrac{p-\Var^\phi[d\hat{q},dq]^{-1}\Var^\phi[d\hat{q},P]\cdot q}{\Var^\phi[d\hat{q},dq]^{-1}q\cdot q}+O[\alpha] \notag \\
\delta&=\Var^\phi[d\hat{q},dq]^{-1}(\Var^\phi[d\hat{q},P]+qs\,dt)+O[\alpha] \notag
\end{align}
Quantum solutions are similar to the corresponding classical solutions, but include a correction to the hedge portfolio when there is non-commutativity among the components of the underlying return vector.

\subsection{Option pricing}

Consider a market that contains a funding security whose final price is represented by the positive-definite matrix $B$ and an underlying security whose final price is represented by the self-adjoint matrix $Q$. For the spread portfolio with strike $k$ whose final price is represented by the self-adjoint matrix $Q-kB$, let $u$ be an eigenstate and let $\zeta$ be the corresponding eigenvalue satisfying:
\begin{equation}
(Q-kB)u=\zeta u
\end{equation}
The sensitivity of the eigenvalue to the strike is negative:
\begin{equation}
\frac{d\zeta}{dk}=-\bfrac{u^\dagger Bu}{u^\dagger u}
\end{equation}
The eigenstate may also be sensitive to the strike, and this creates higher-order convexity in the relationship between strike and eigenvalue. If the matrices $Q$ and $B$ commute then the eigenstates of the matrix $Q-kB$ are fixed as the strike varies and the corresponding eigenvalues are linear in the strike. If the matrices $Q$ and $B$ do not commute then the eigenstates of the matrix $Q-kB$ are not fixed as the strike varies and the corresponding eigenvalues are non-linear in the strike.

Convexity between strike and eigenvalue can be utilised in option pricing. Assume for convenience that, for each $i=1,\ldots,n$, the $i$th eigenvalue $\zeta_i$ is strictly decreasing as a function of the strike and is zero at the unique strike $x_i$. If the calibration condition is solved with unit optimal portfolio $\phi=0$ and implied funding rate $r=0$, so that funding adjustments can be ignored, then the mid price for the option is:
\begin{align}
&\frac{1}{n}\tr[(Q-kB)^+]=\frac{1}{n}\sum_{i=1}^n\zeta_i[k]^+ \\
&\qquad=\sum_{i=1}^n(x_i-k)^+w_i+\int_{x=-\infty}^\infty(x-k)^+w[x]\,dx \notag
\end{align}
where the second expression is derived from the first expression using integration by parts with:
\begin{align}
w_i&:=-\frac{1}{n}\zeta_i'[x_i] \\
w[x]&:=\frac{1}{n}\sum_{i=1}^n\zeta_i''[x](1_{x<x_i}) \notag
\end{align}
The option price is expressed as an integral of the payoff $(x-k)^+$ with discrete weight $w_i$ at the point $x_i$ and continuous density $w[x]$ at the point $x$. In the classical model with commuting matrices, the eigenvalues satisfy $\zeta_i''[x]=0$ and the implied distribution is discrete. In the quantum model with non-commuting matrices, the eigenvalues satisfy $\zeta_i''[x]\ne0$ and the implied distribution has discrete and continuous components.

This ability of the model to extract an implied distribution with continuous support from discrete variables is the source of its efficacy. Heisenberg uncertainty can be calibrated to realistic market volatility smiles with as few as two eigenstates, demonstrating convergence in the quantum model that is exponentially faster than classical Monte-Carlo estimation. The eigenvalue convexity this exploits is precisely the feature that enables violation of the classical Bell's inequality, and is a uniquely quantum property.

The quantum contribution to the model is captured in the unitary matrix $U$ that rotates from the eigenbasis of the matrix $Q$ to the eigenbasis of the matrix $B$. Starting with diagonal matrices for $Q$ and $B$, the option price is then expressed as:
\begin{equation}
\frac{1}{n}\tr[(UQU^{-1}-kB)^+]
\end{equation}
No rotation between the eigenbases is generated when $U$ is given by the identity:
\begin{equation}
U_{jk}=1_{j=k}
\end{equation}
For the simplest case $n=2$ this derives the option price:
\begin{equation}
\frac{1}{2}(Q_1-kB_1)^++\frac{1}{2}(Q_2-kB_2)^+
\end{equation}
Maximum rotation between the eigenbases is generated when $U$ is given by the discrete Fourier transform:
\begin{equation}
U_{jk}=\frac{1}{\sqrt{n}}\exp\!\left[2\pi i\frac{jk}{n}\right]
\end{equation}
For the simplest case $n=2$ this derives the option price:
\begin{align}
&\frac{1}{2}\!\left((\bar{Q}-k\bar{B})+\sqrt{\hat{Q}^2+k^2\hat{B}^2}\right)^{\!+} \\
&+\frac{1}{2}\!\left((\bar{Q}-k\bar{B})-\sqrt{\hat{Q}^2+k^2\hat{B}^2}\right)^{\!+} \notag
\end{align}
where:
\begin{alignat}{2}
\bar{B}&:=\frac{1}{2}(B_1+B_2) &\qquad \hat{B}&:=\frac{1}{2}(B_1-B_2) \\
\bar{Q}&:=\frac{1}{2}(Q_1+Q_2) &\qquad \hat{Q}&:=\frac{1}{2}(Q_1-Q_2) \notag
\end{alignat}
The additional convexity in the relationship between strike and eigenvalue creates an implied distribution with continuous support.

\begin{figure*}[!pt]
\centering
\begin{tabular}{@{}C{0.33\textwidth-\tabcolsep}C{0.33\textwidth-\tabcolsep}C{0.33\textwidth-\tabcolsep}@{}}
\includegraphics[width=0.33\textwidth-\tabcolsep]{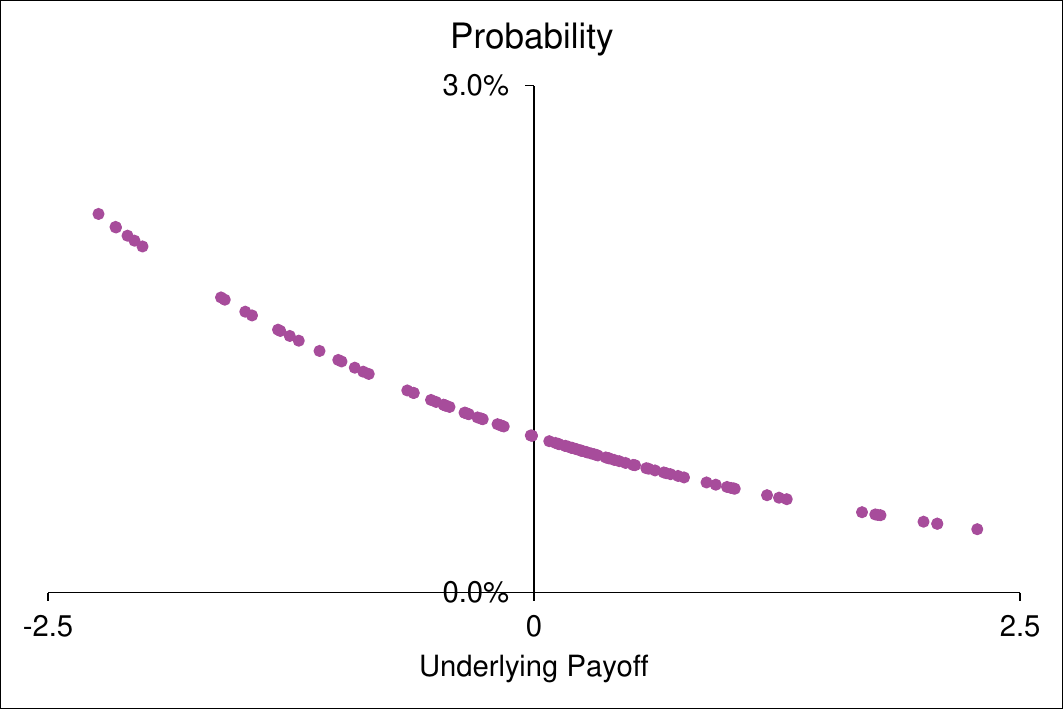}&
\includegraphics[width=0.33\textwidth-\tabcolsep]{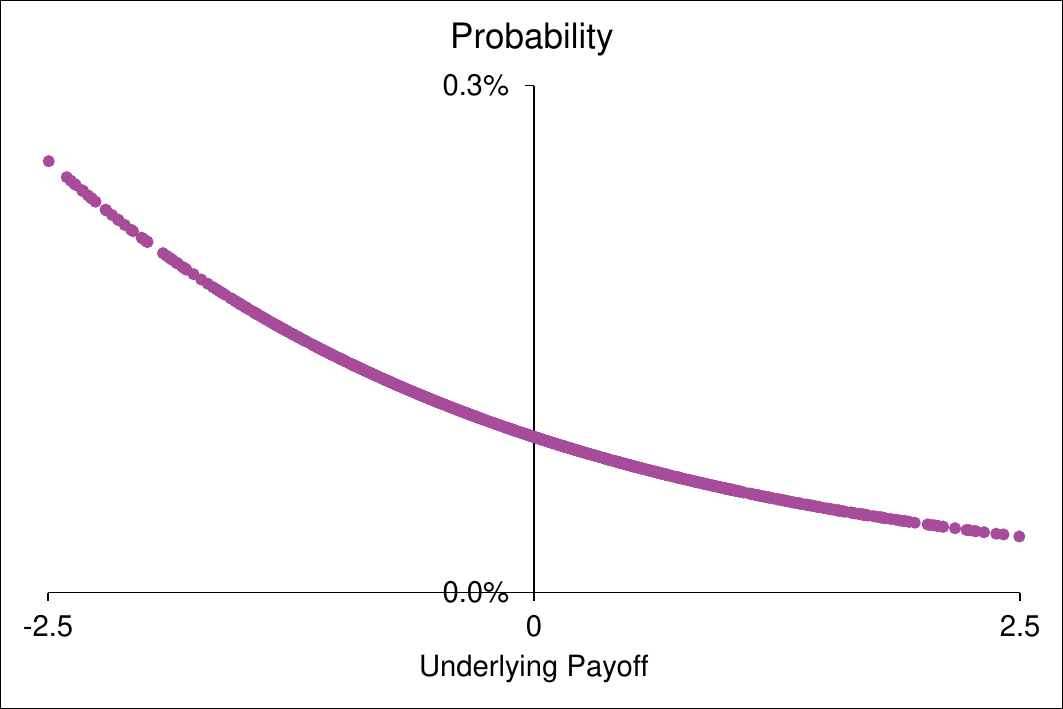}&
\includegraphics[width=0.33\textwidth-\tabcolsep]{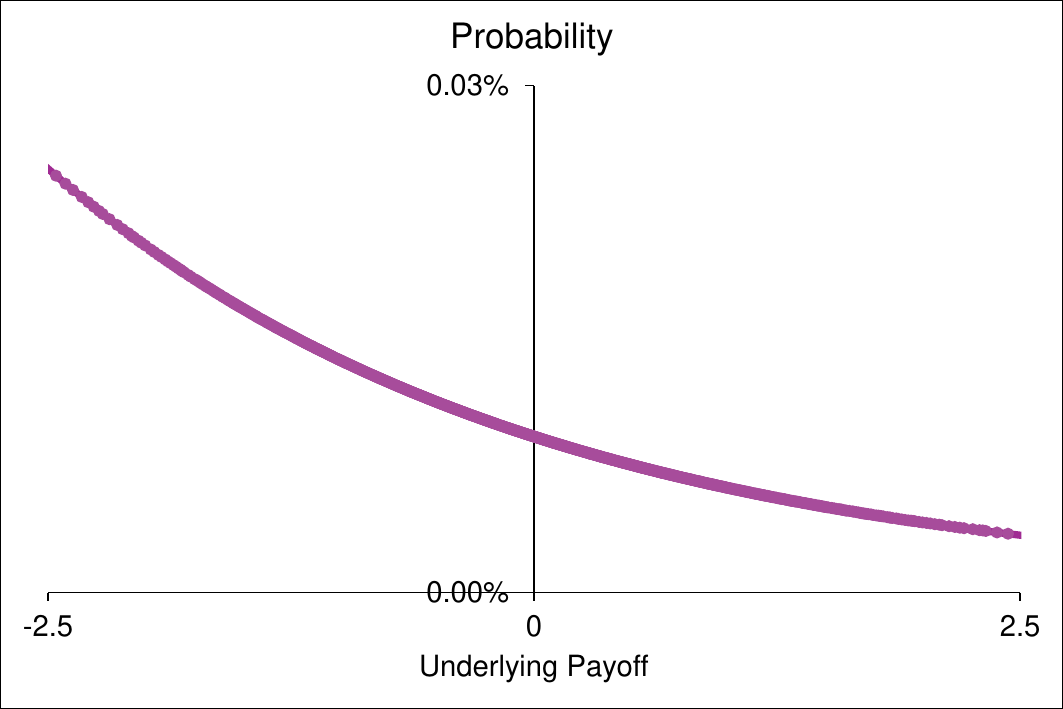}\\
\includegraphics[width=0.33\textwidth-\tabcolsep]{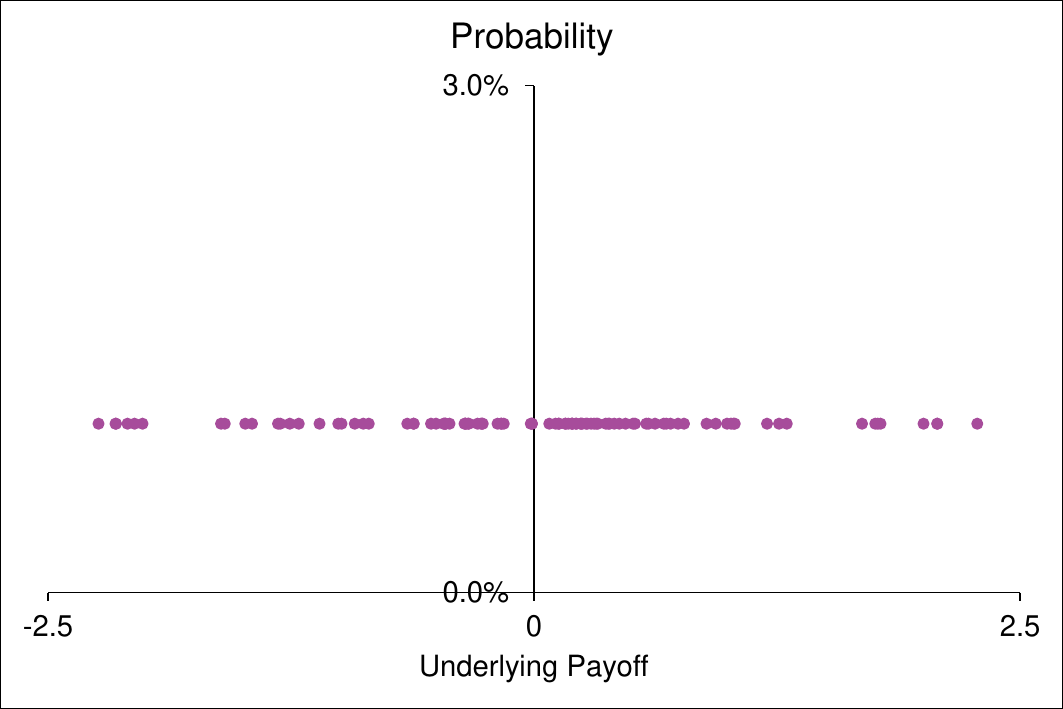}&
\includegraphics[width=0.33\textwidth-\tabcolsep]{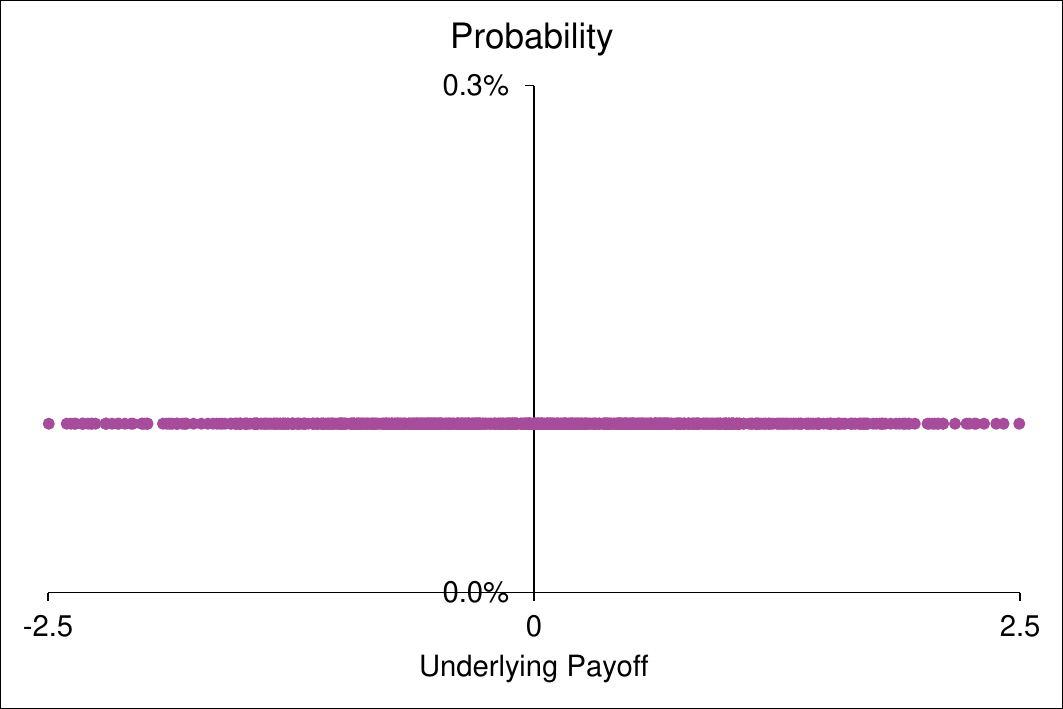}&
\includegraphics[width=0.33\textwidth-\tabcolsep]{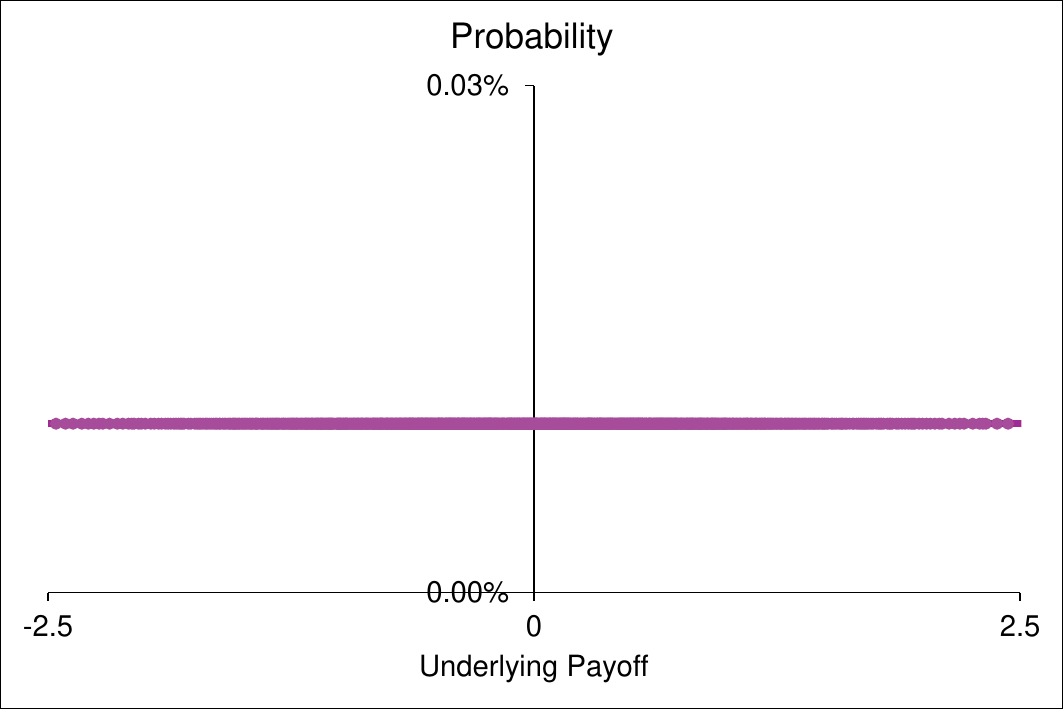}\\
\includegraphics[width=0.33\textwidth-\tabcolsep]{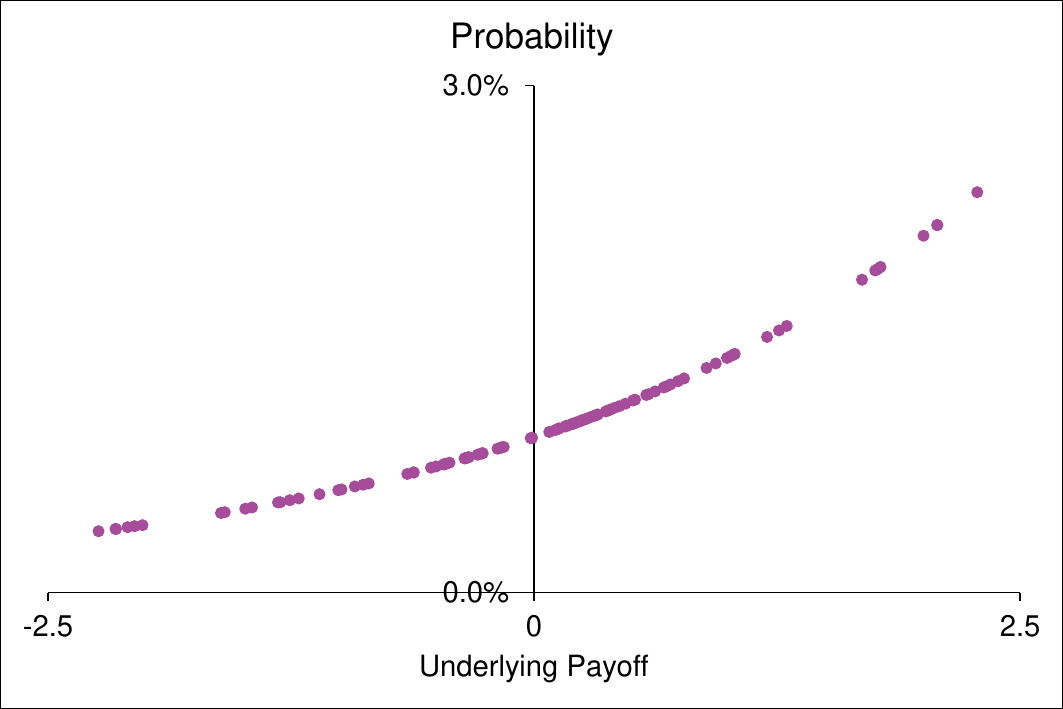}&
\includegraphics[width=0.33\textwidth-\tabcolsep]{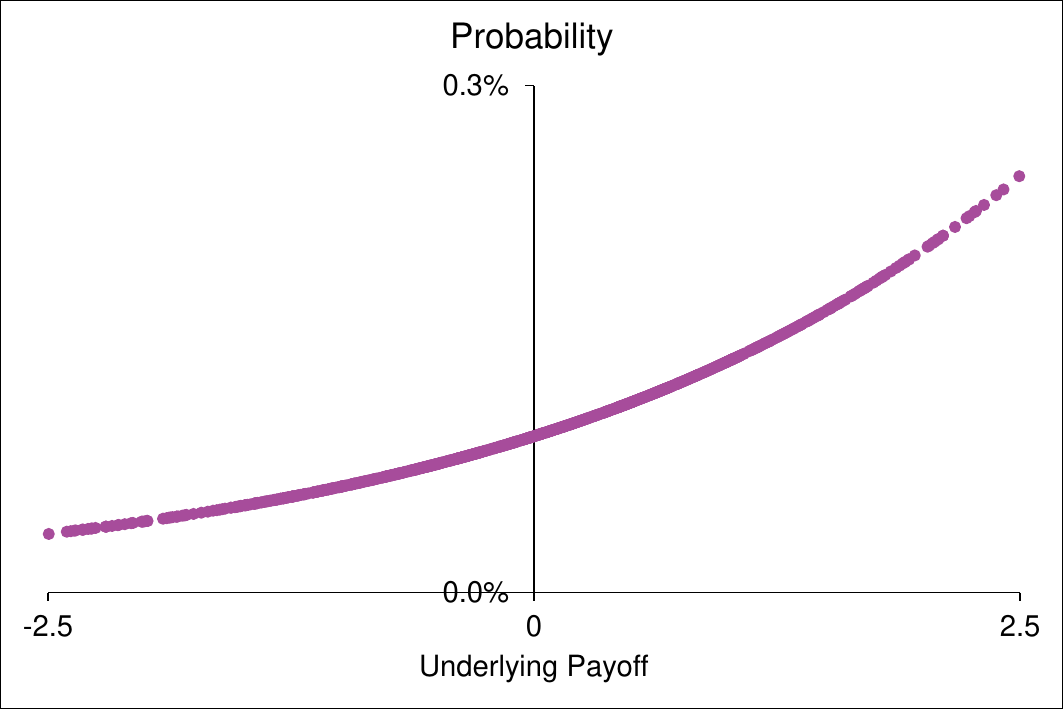}&
\includegraphics[width=0.33\textwidth-\tabcolsep]{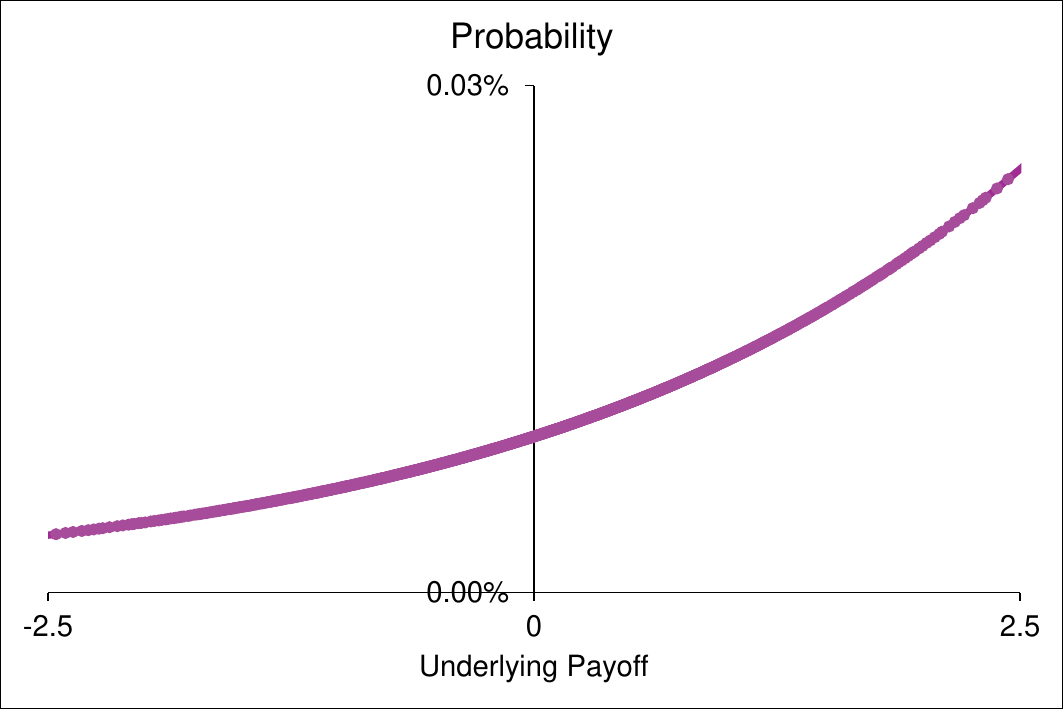}\\
Distribution ($n=\num{100}$)&Distribution ($n=\num{1000}$)&Distribution ($n=\num{10000}$)
\end{tabular}
\caption{In these studies of discrete models on classical information, the final underlying price is uniformly distributed in the expectation measure on $n=\num{100}$, $\num{1000}$ and $\num{10000}$ samples from a standardised normal variable, centralised to have zero mean. The price measure is supported on the same samples with probabilities adjusted to match the target initial underlying price $q=-0.4$ (top row), $q=0$ (middle row) and $q=0.4$ (bottom row). In the transformation from expectation to price measure, the uniform probabilities are biased using an exponential function of the final underlying price. Among equivalent measures that satisfy the calibration condition, the price measure thus has minimum entropy relative to the expectation measure.}
\label{fig:classicalprobs}
\end{figure*}

\begin{figure*}[!pt]
\centering
\begin{tabular}{@{}C{0.33\textwidth-\tabcolsep}C{0.33\textwidth-\tabcolsep}C{0.33\textwidth-\tabcolsep}@{}}
\includegraphics[width=0.33\textwidth-\tabcolsep]{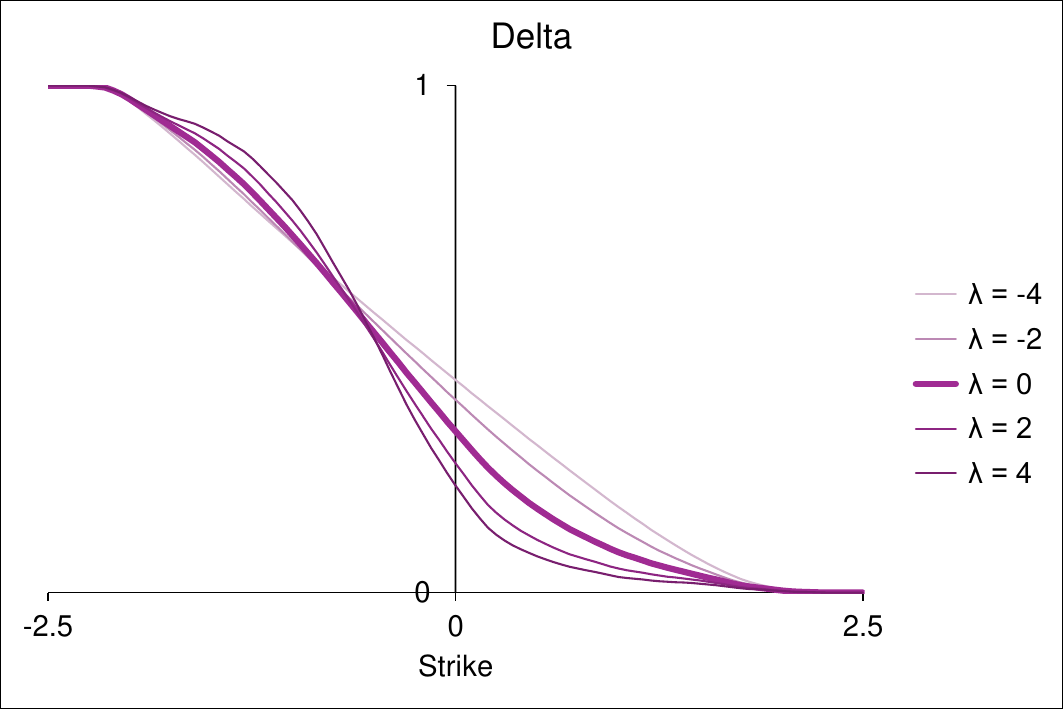}&
\includegraphics[width=0.33\textwidth-\tabcolsep]{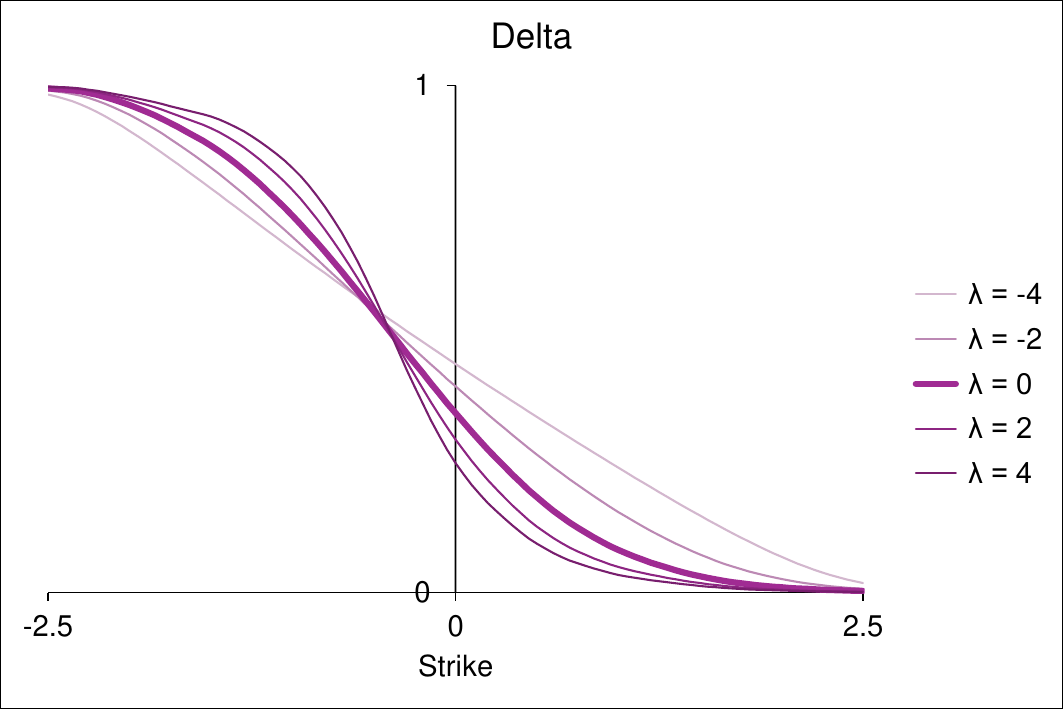}&
\includegraphics[width=0.33\textwidth-\tabcolsep]{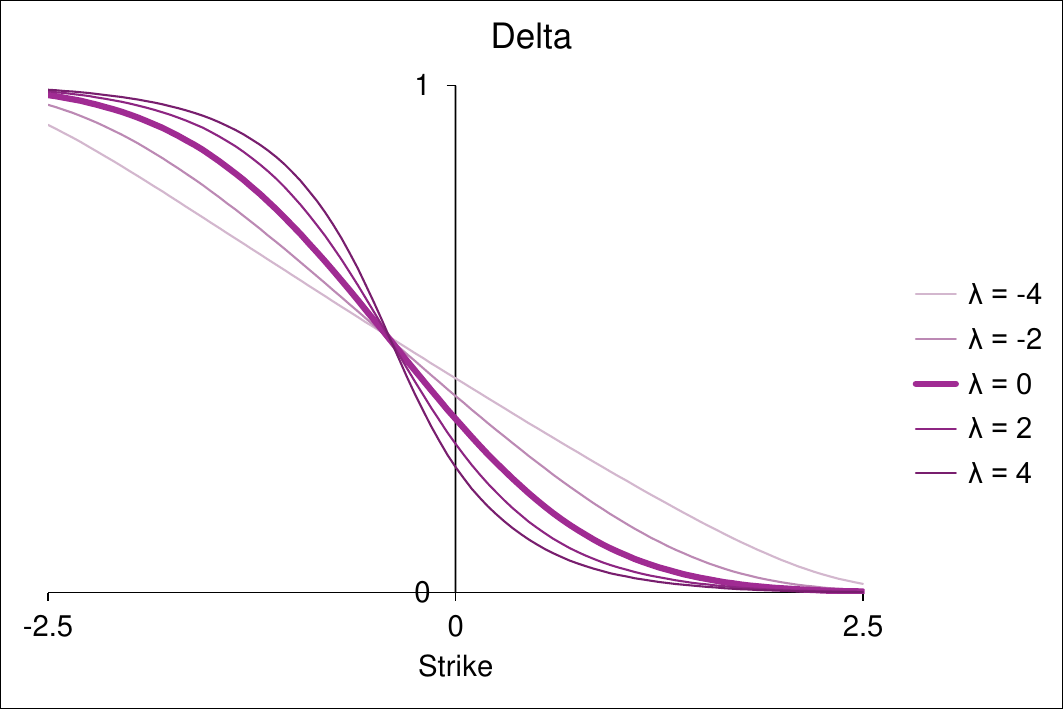}\\
\includegraphics[width=0.33\textwidth-\tabcolsep]{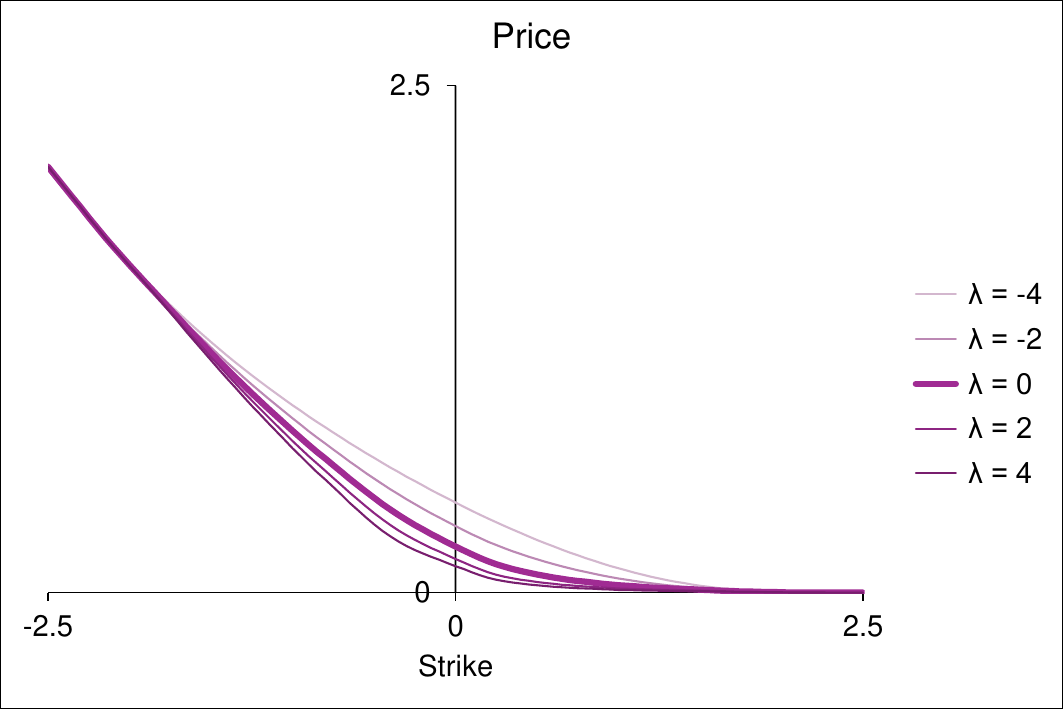}&
\includegraphics[width=0.33\textwidth-\tabcolsep]{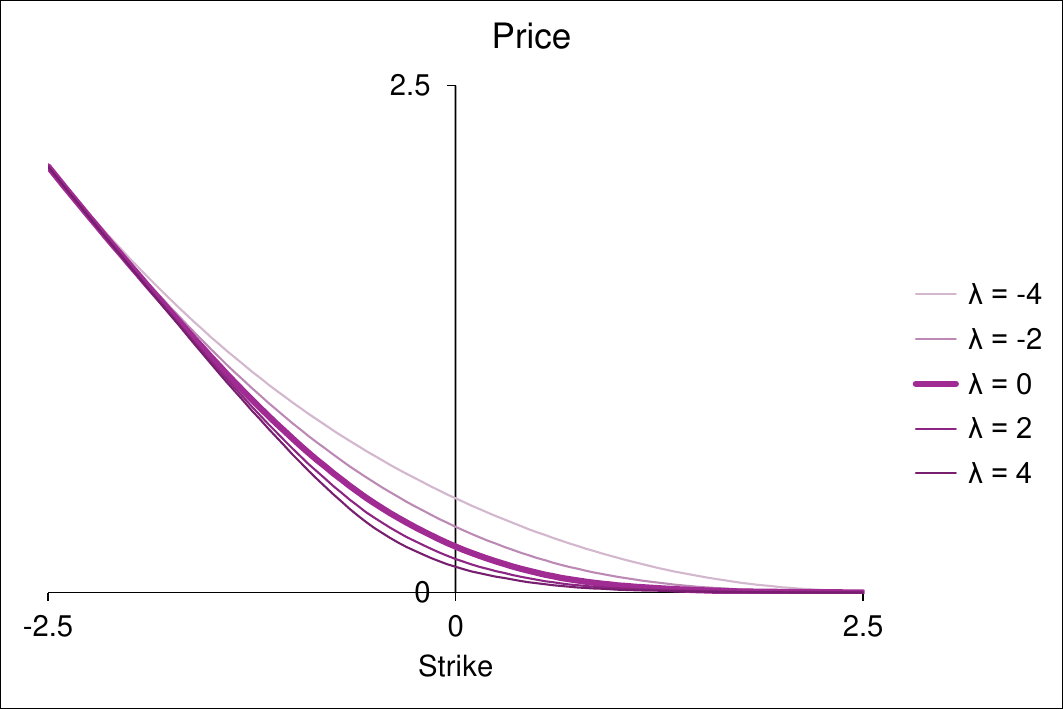}&
\includegraphics[width=0.33\textwidth-\tabcolsep]{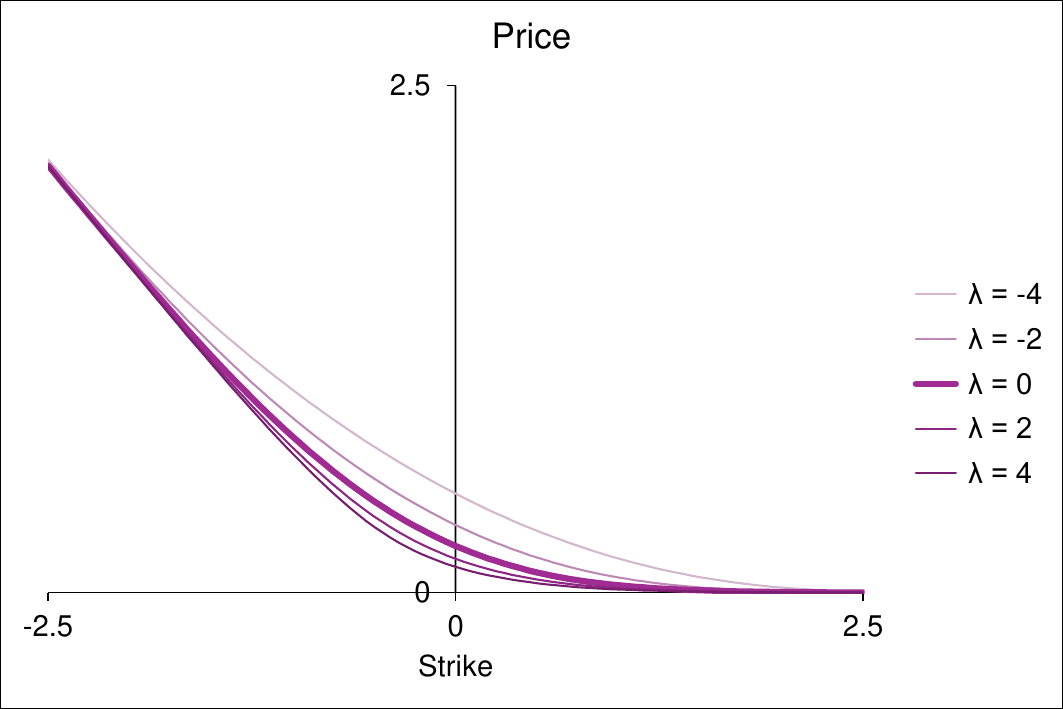}\\
\includegraphics[width=0.33\textwidth-\tabcolsep]{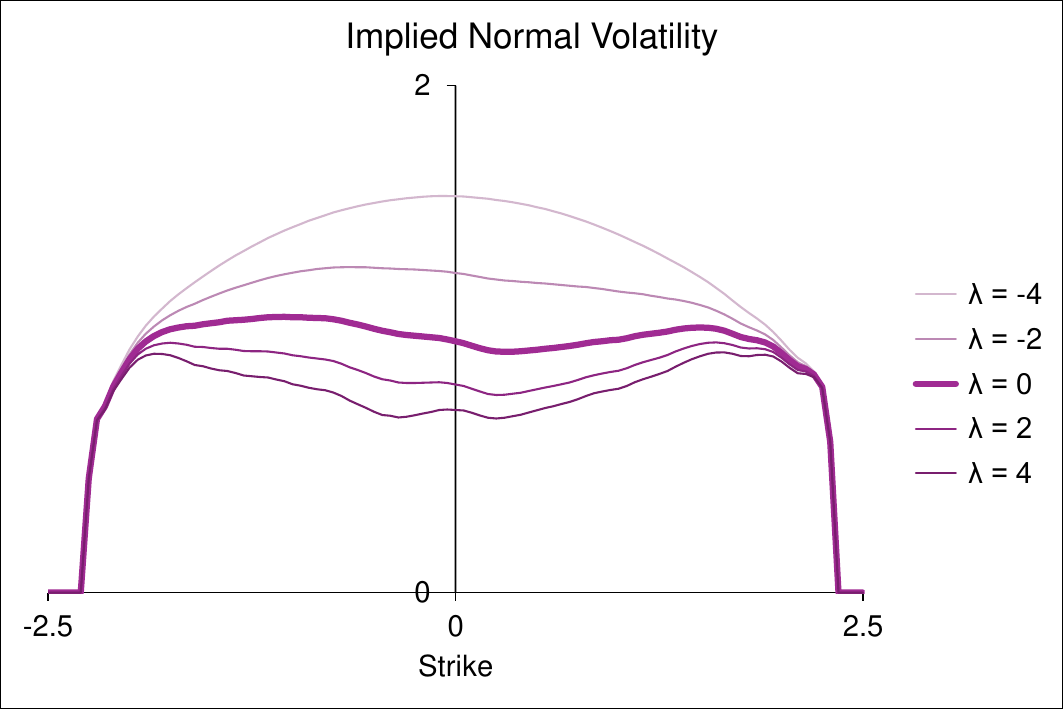}&
\includegraphics[width=0.33\textwidth-\tabcolsep]{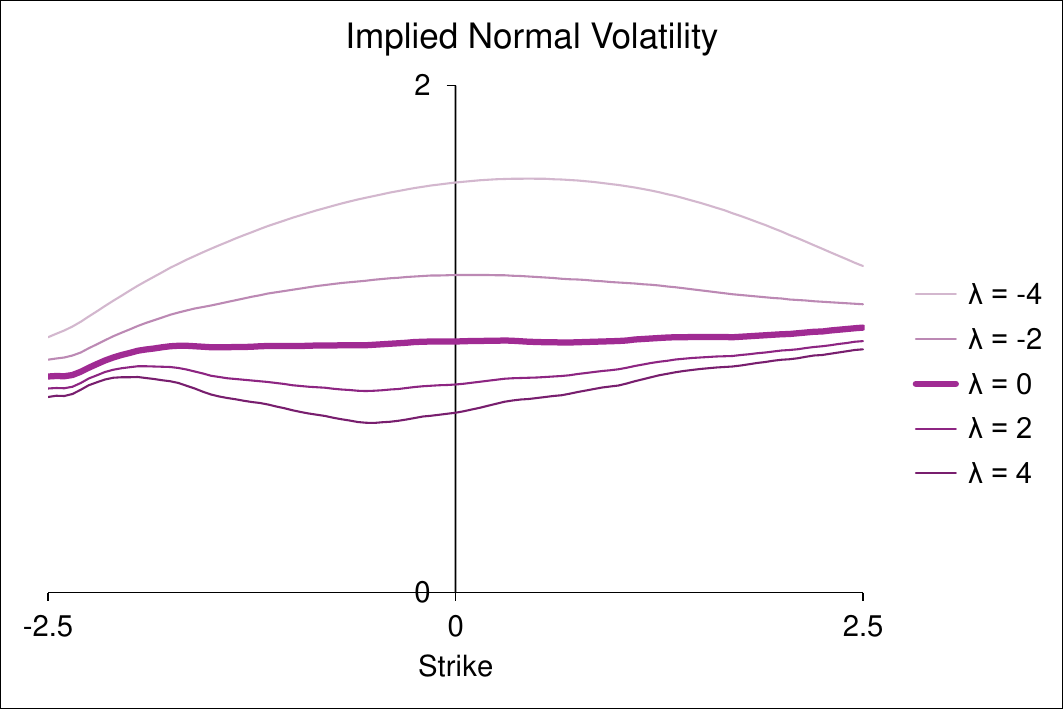}&
\includegraphics[width=0.33\textwidth-\tabcolsep]{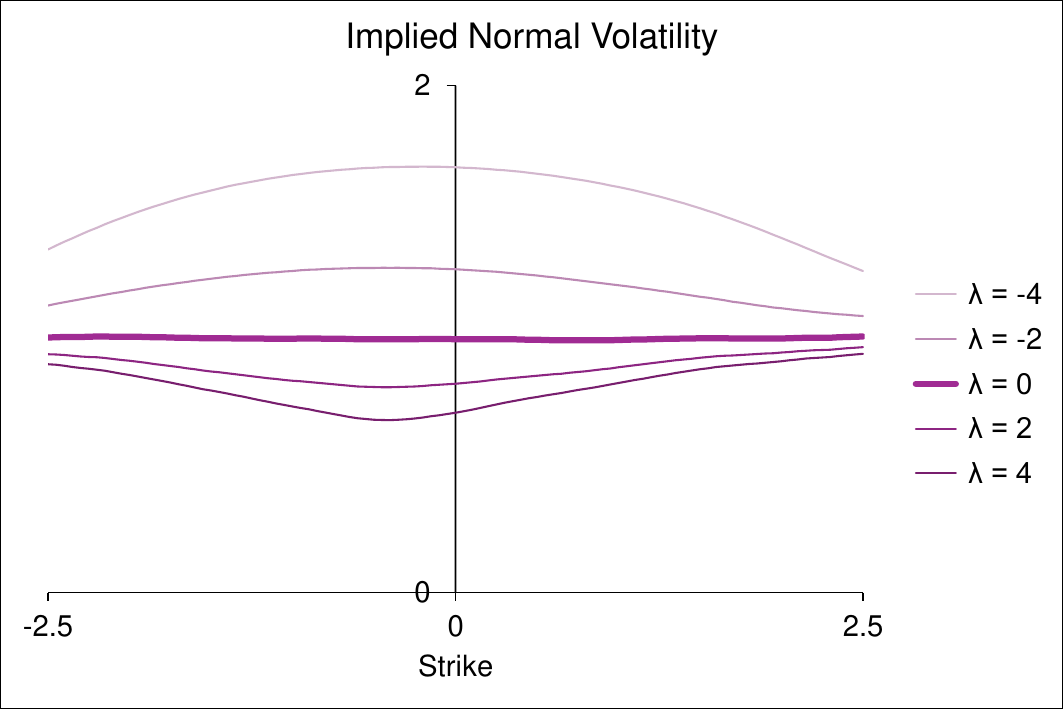}\\
Call Option ($n=\num{100}$)&Call Option ($n=\num{1000}$)&Call Option ($n=\num{10000}$)\\
\includegraphics[width=0.33\textwidth-\tabcolsep]{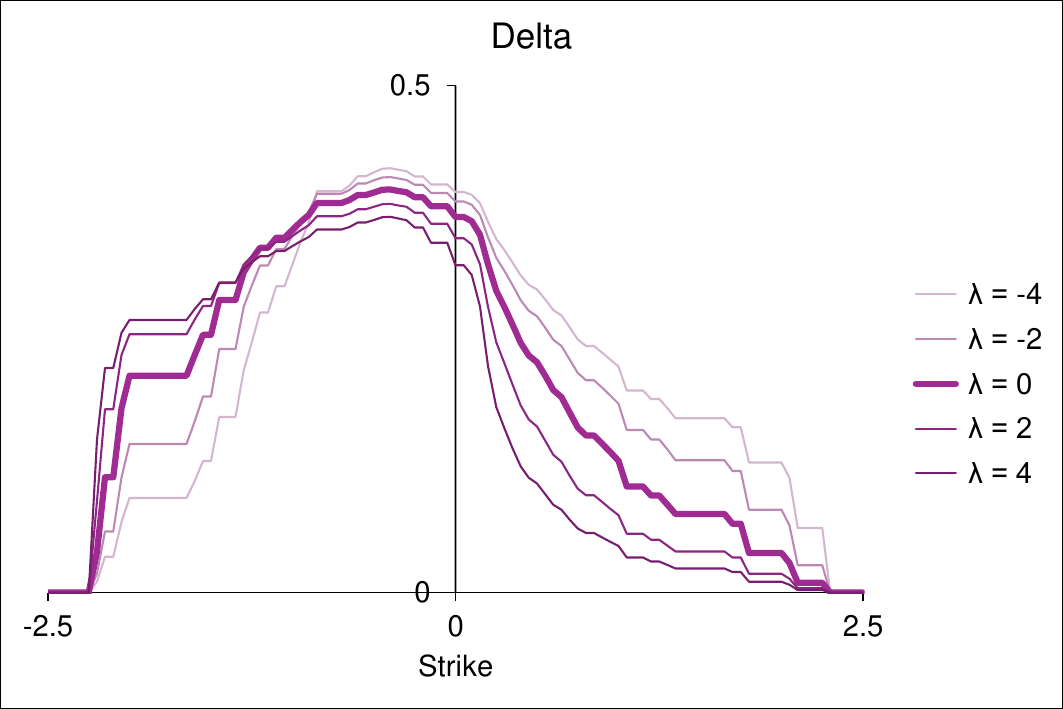}&
\includegraphics[width=0.33\textwidth-\tabcolsep]{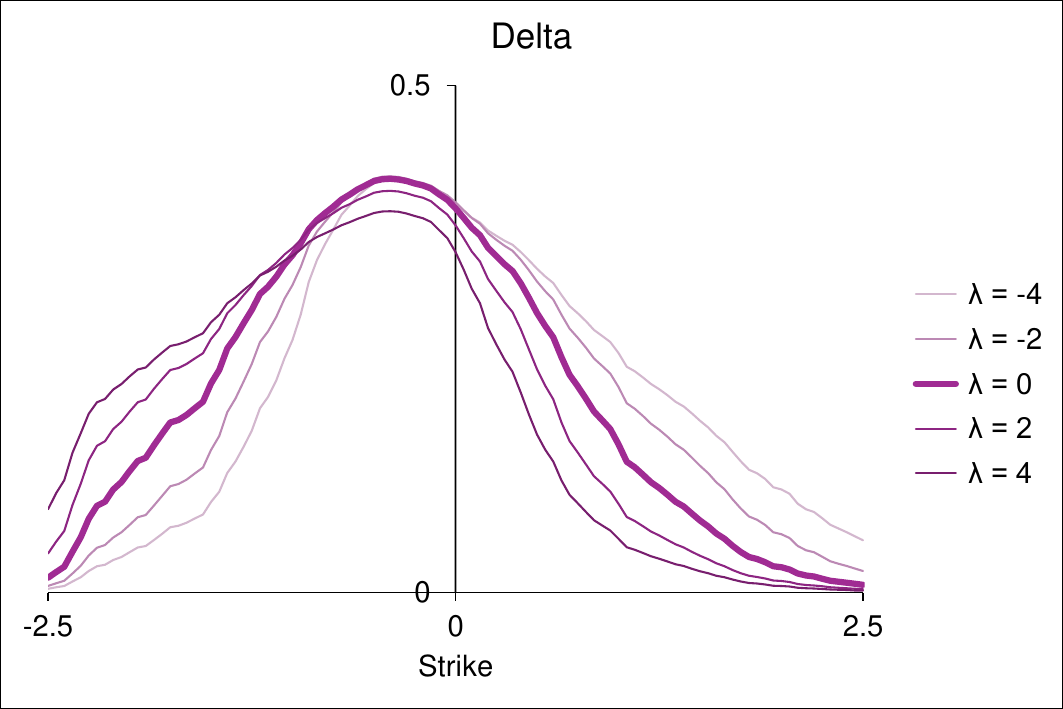}&
\includegraphics[width=0.33\textwidth-\tabcolsep]{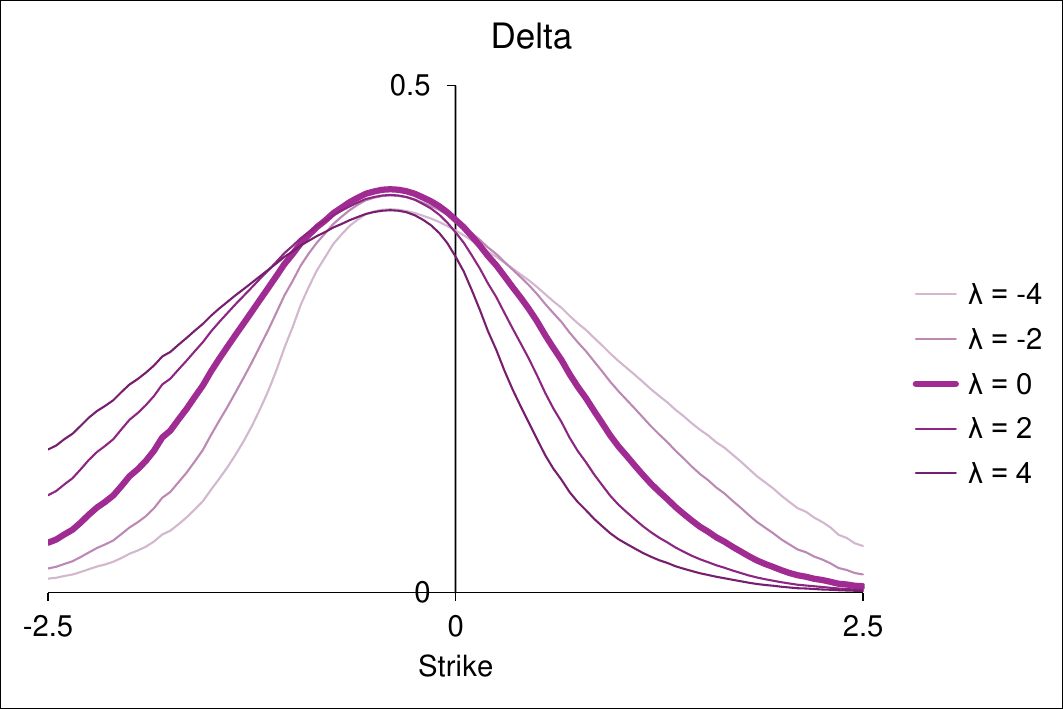}\\
\includegraphics[width=0.33\textwidth-\tabcolsep]{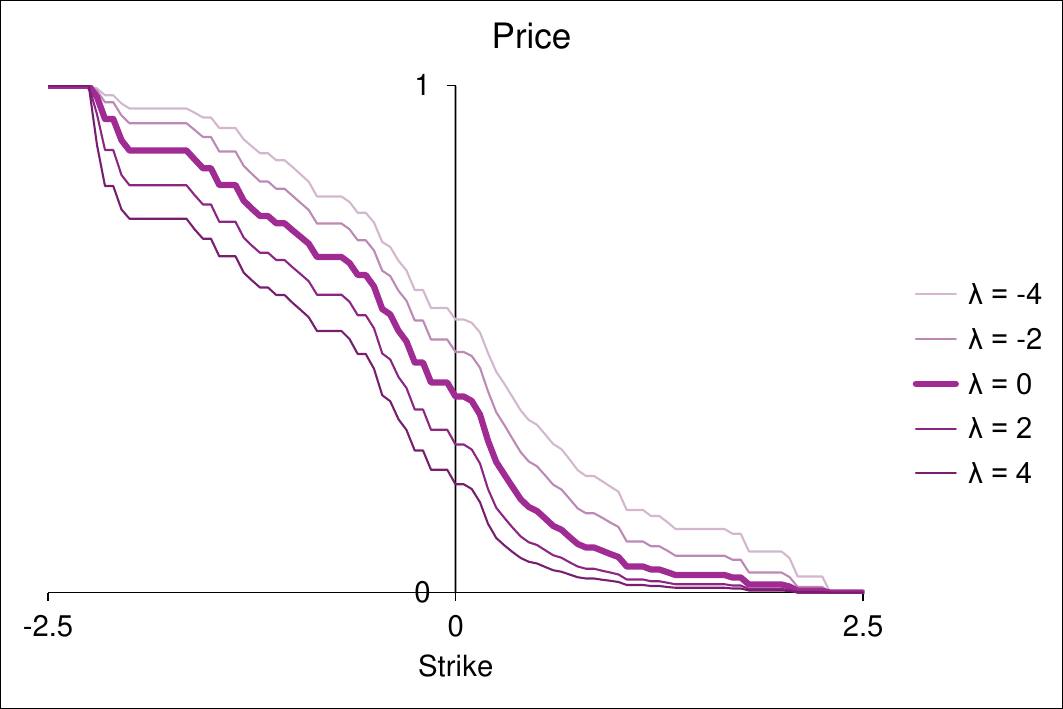}&
\includegraphics[width=0.33\textwidth-\tabcolsep]{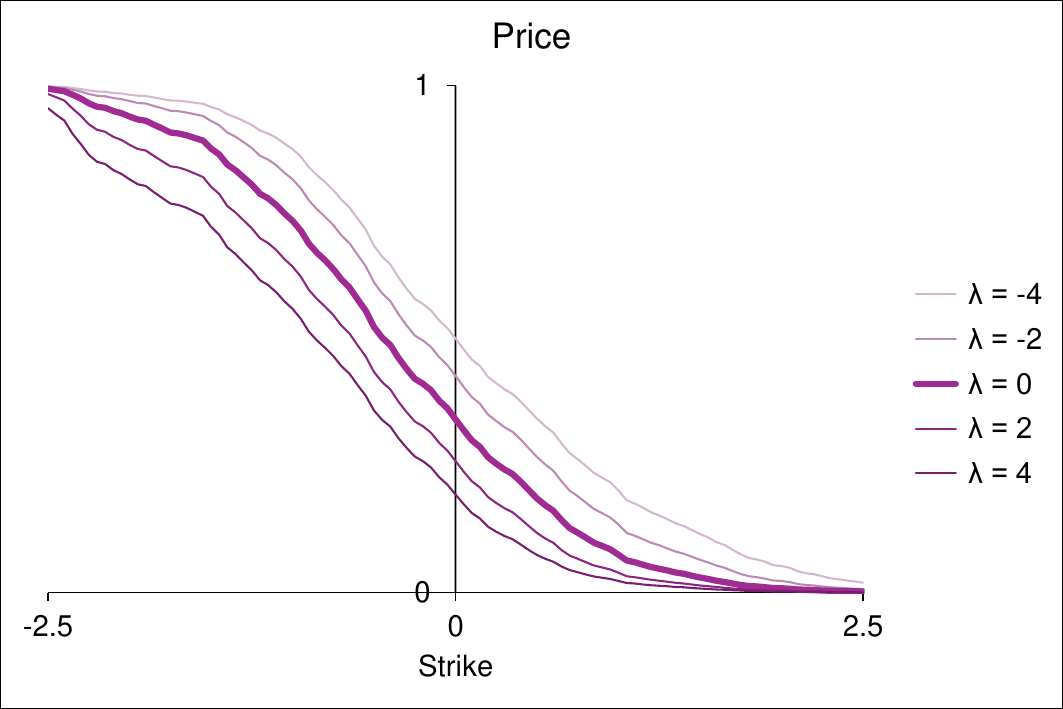}&
\includegraphics[width=0.33\textwidth-\tabcolsep]{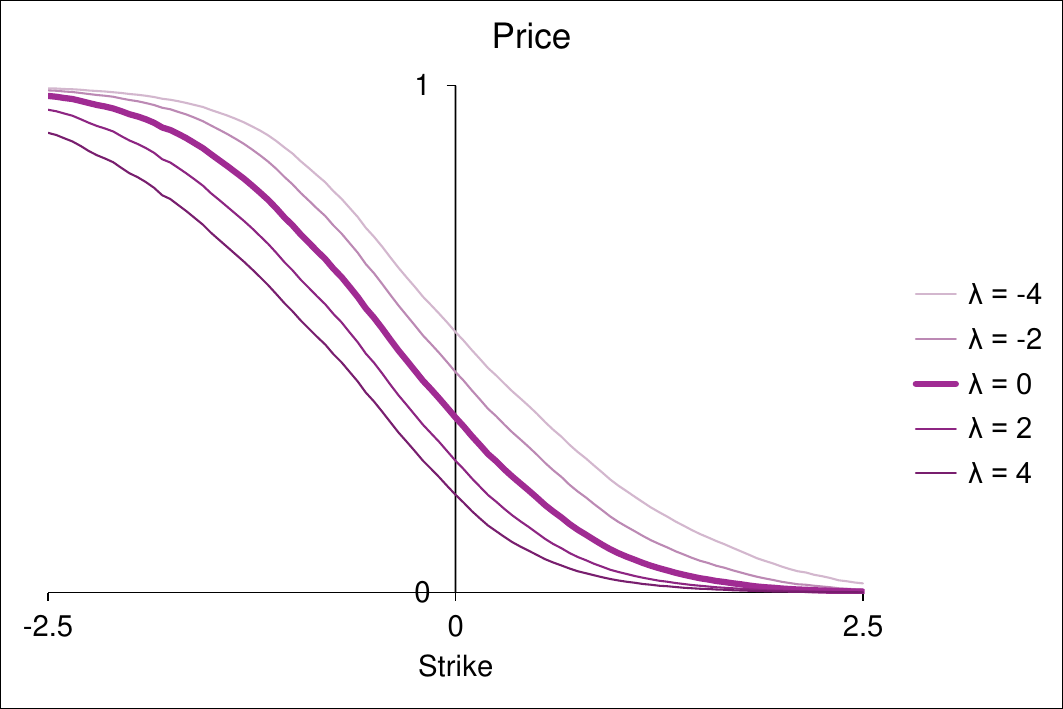}\\
Digital Option ($n=\num{100}$)&Digital Option ($n=\num{1000}$)&Digital Option ($n=\num{10000}$)
\end{tabular}
\caption{In these studies of discrete models on classical information, the final underlying price is uniformly distributed in the expectation measure on $n=\num{100}$, $\num{1000}$ and $\num{10000}$ samples from a standardised normal variable. The price measure is supported on the same samples with probabilities adjusted to match the target initial underlying price $q=-0.4$. Entropic risk optimisation identifies the mid and bid-offer for the price and the hedge ratio for a range of notionals $\lambda$.}
\label{fig:classicalM4}
\end{figure*}

\begin{figure*}[!pt]
\centering
\begin{tabular}{@{}C{0.33\textwidth-\tabcolsep}C{0.33\textwidth-\tabcolsep}C{0.33\textwidth-\tabcolsep}@{}}
\includegraphics[width=0.33\textwidth-\tabcolsep]{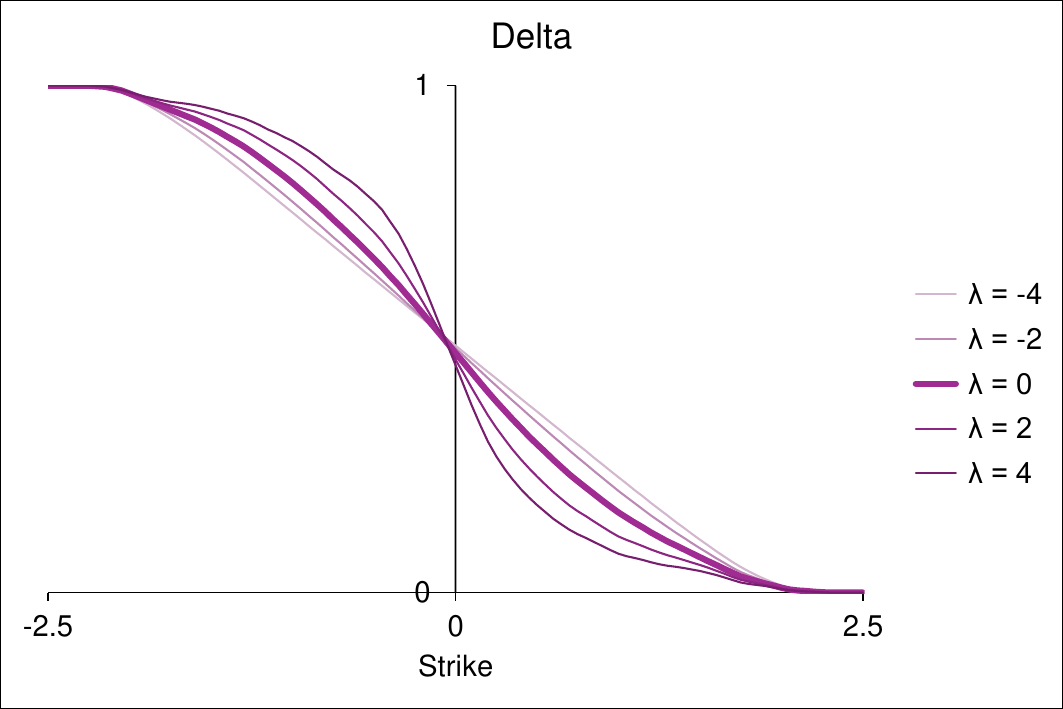}&
\includegraphics[width=0.33\textwidth-\tabcolsep]{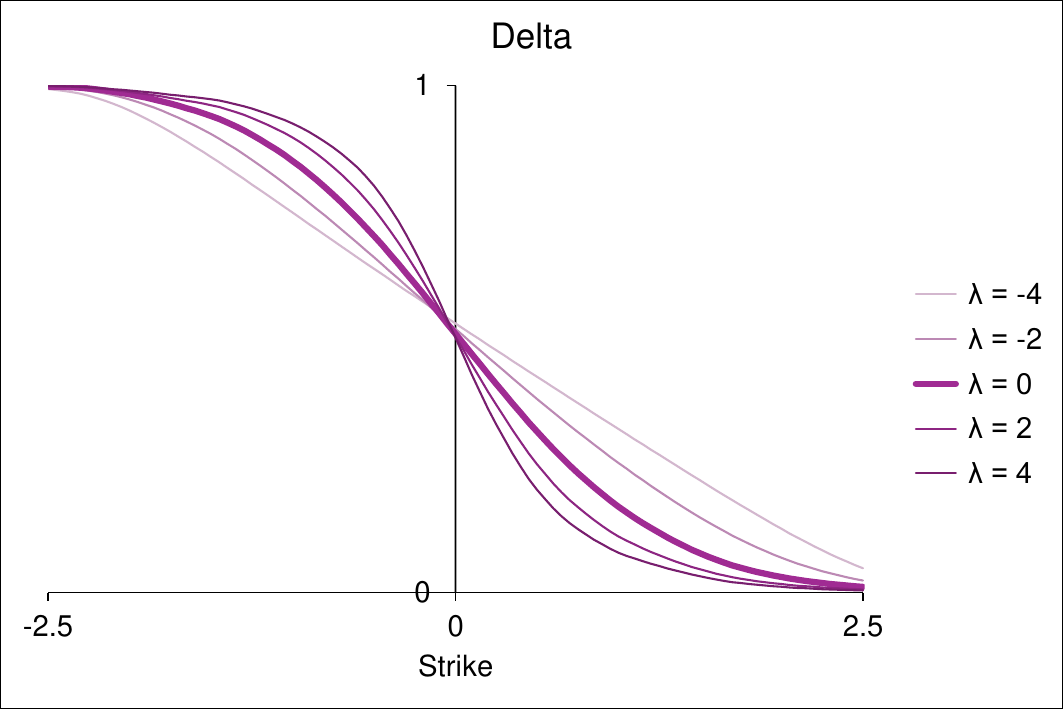}&
\includegraphics[width=0.33\textwidth-\tabcolsep]{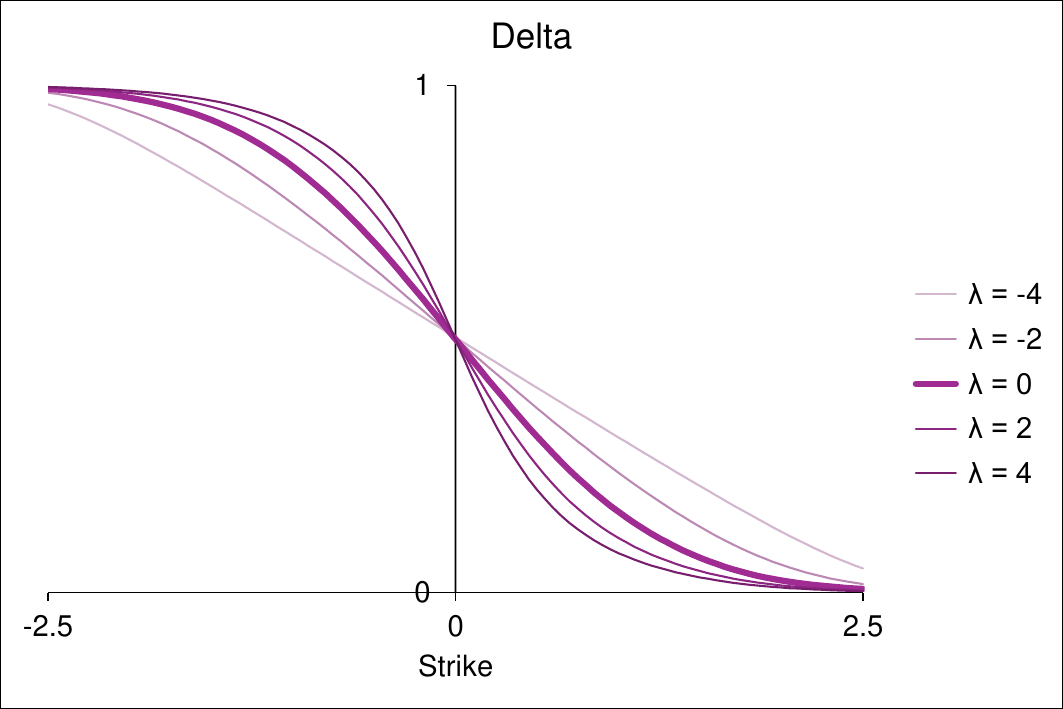}\\
\includegraphics[width=0.33\textwidth-\tabcolsep]{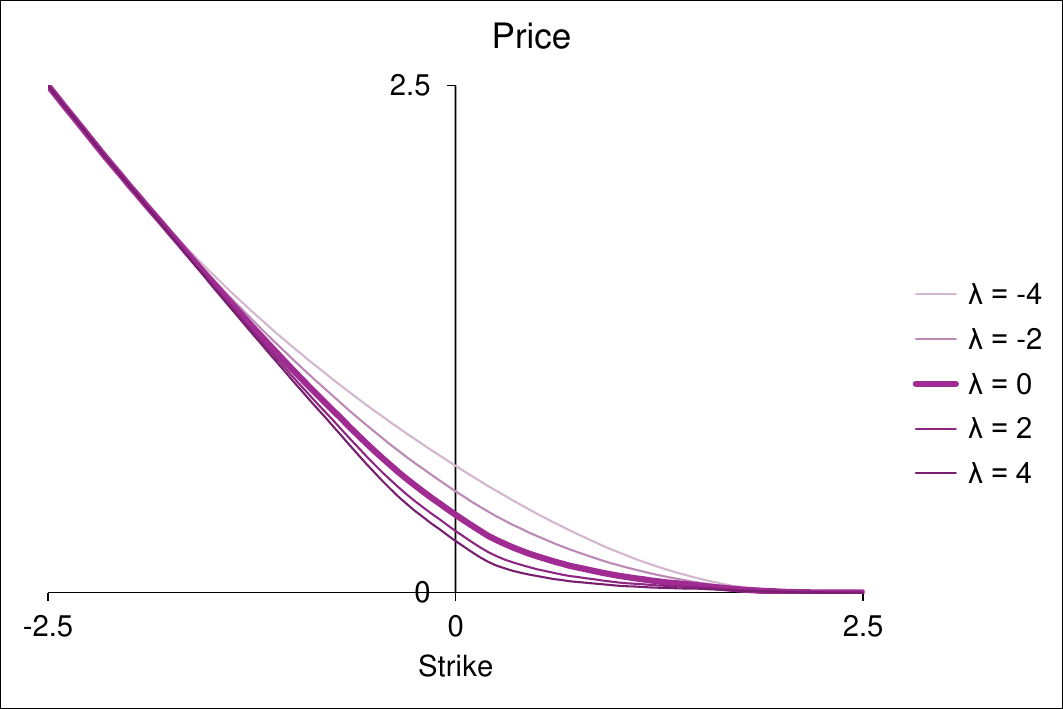}&
\includegraphics[width=0.33\textwidth-\tabcolsep]{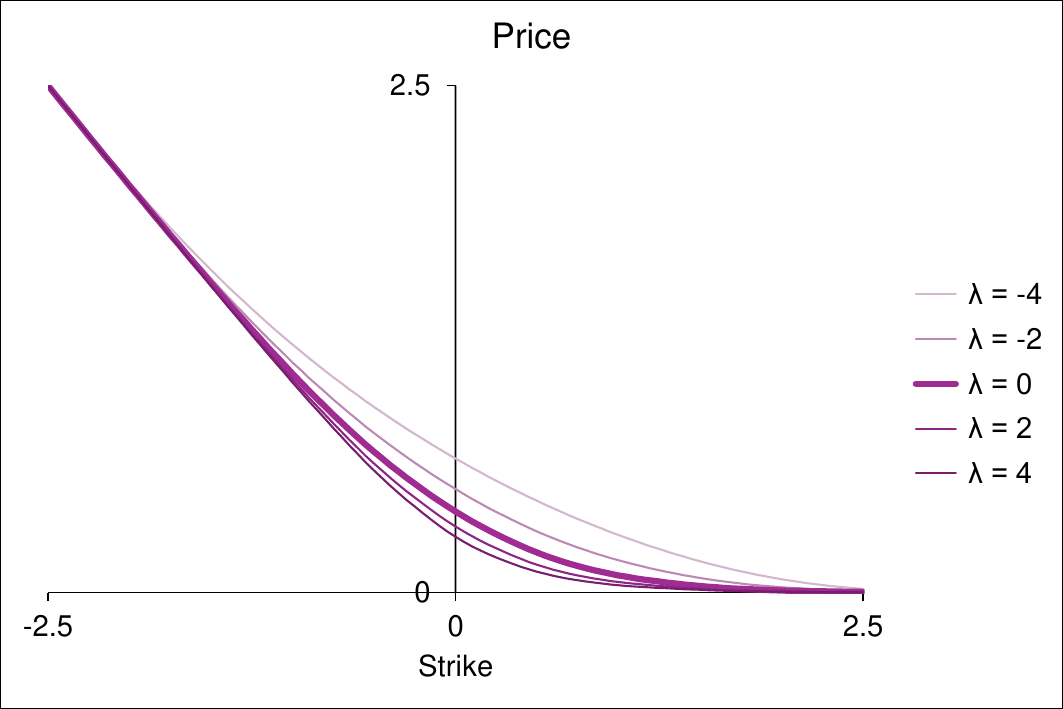}&
\includegraphics[width=0.33\textwidth-\tabcolsep]{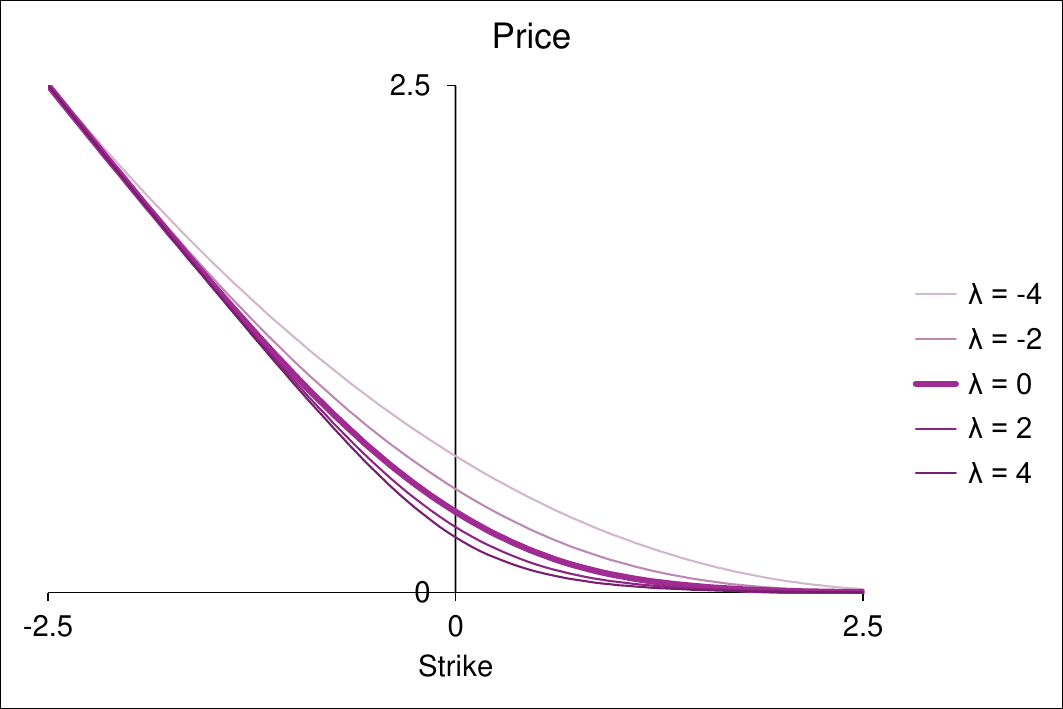}\\
\includegraphics[width=0.33\textwidth-\tabcolsep]{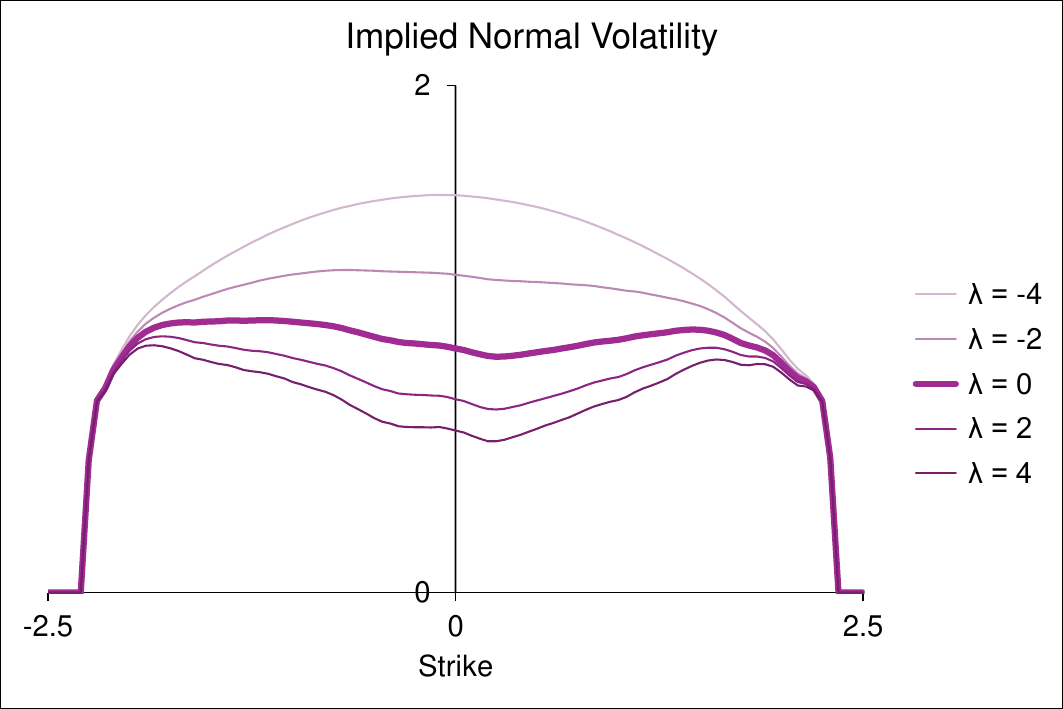}&
\includegraphics[width=0.33\textwidth-\tabcolsep]{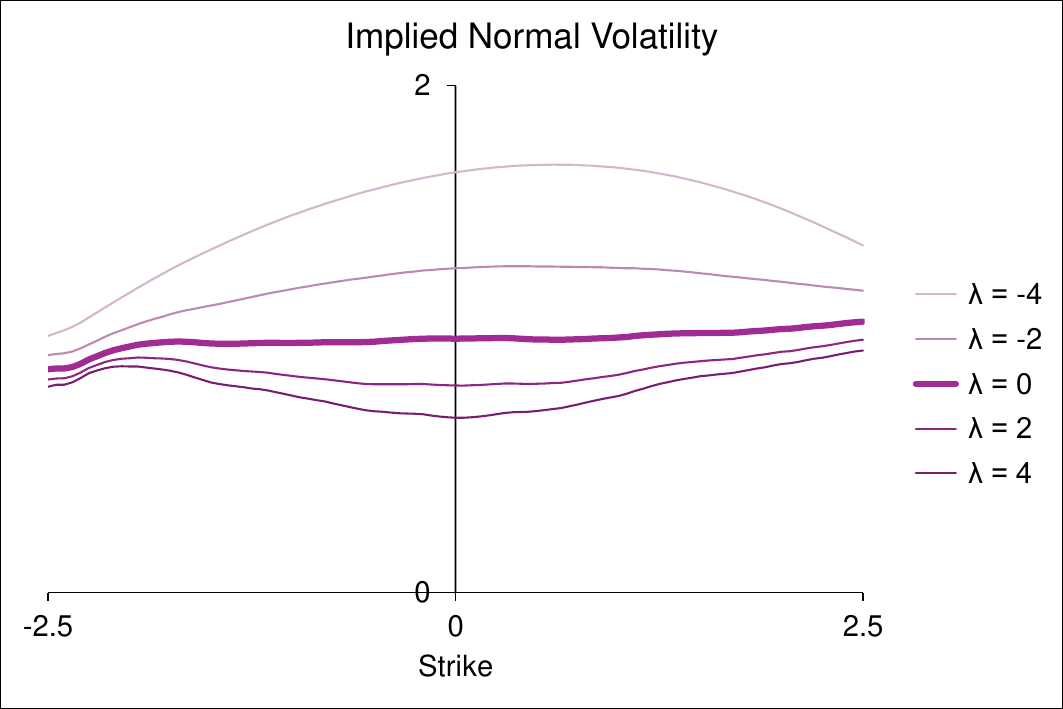}&
\includegraphics[width=0.33\textwidth-\tabcolsep]{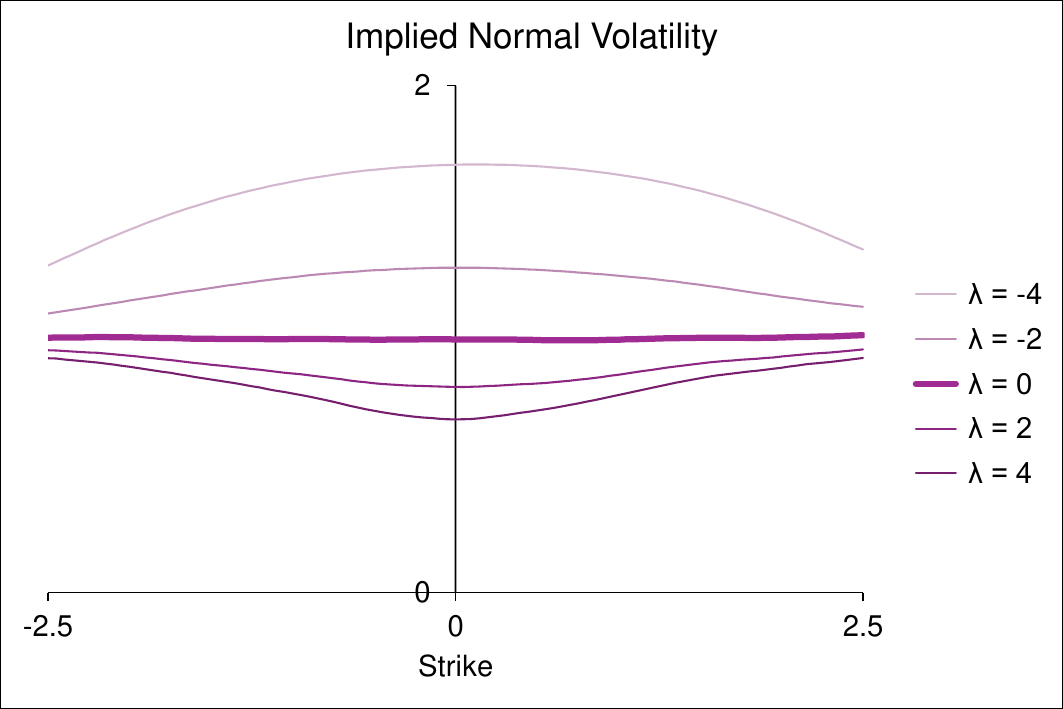}\\
Call Option ($n=\num{100}$)&Call Option ($n=\num{1000}$)&Call Option ($n=\num{10000}$)\\
\includegraphics[width=0.33\textwidth-\tabcolsep]{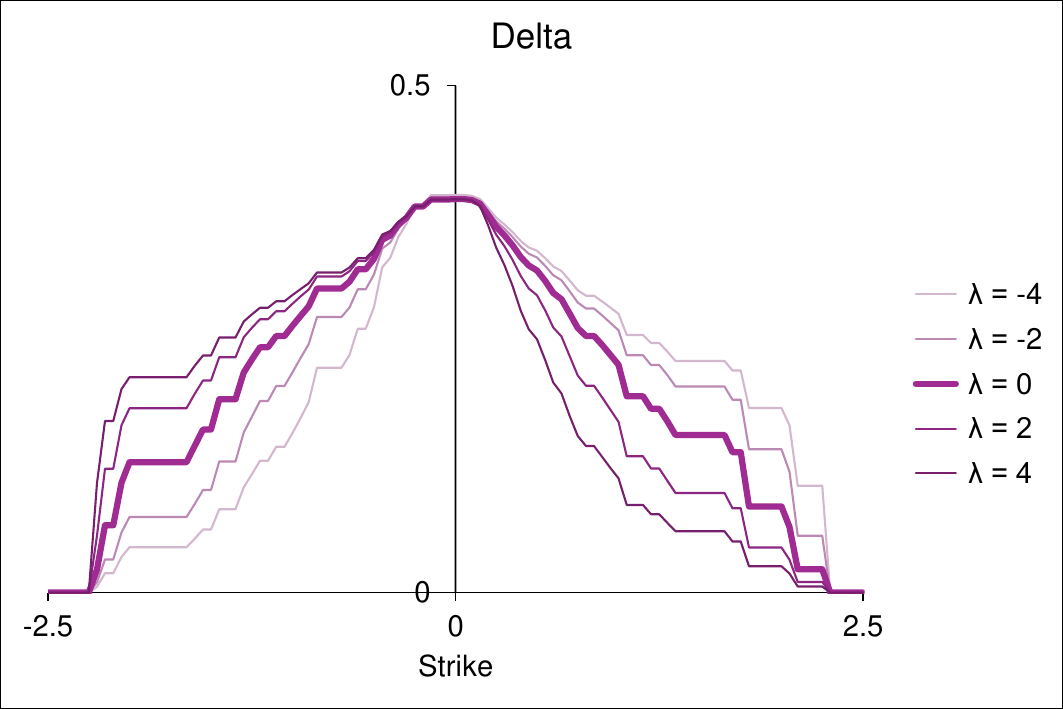}&
\includegraphics[width=0.33\textwidth-\tabcolsep]{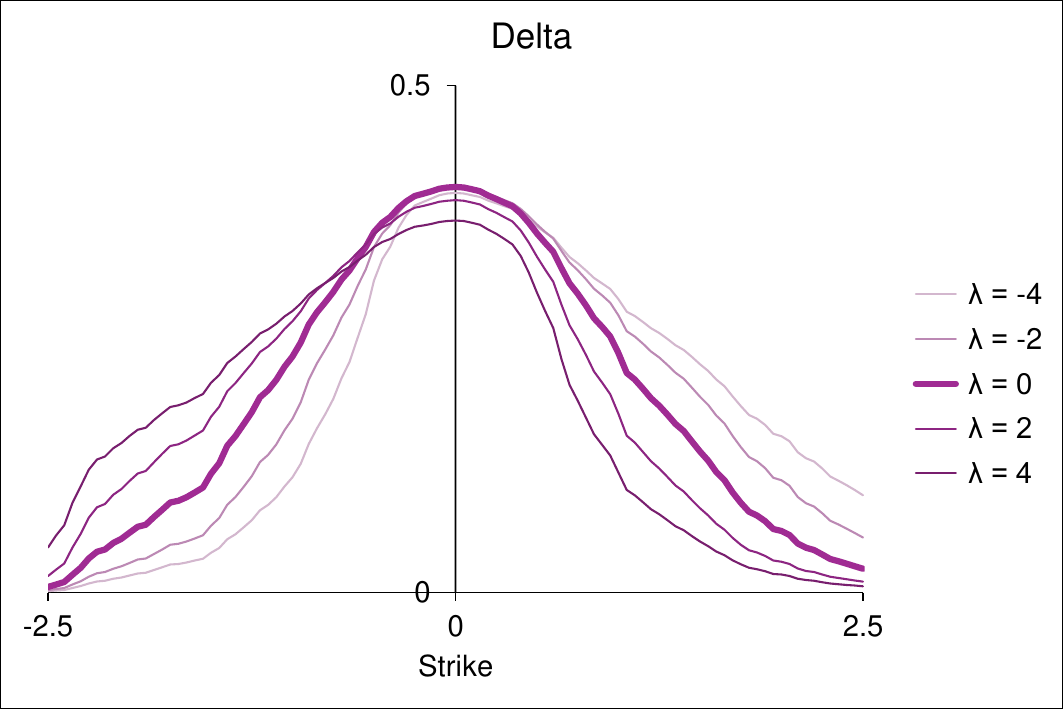}&
\includegraphics[width=0.33\textwidth-\tabcolsep]{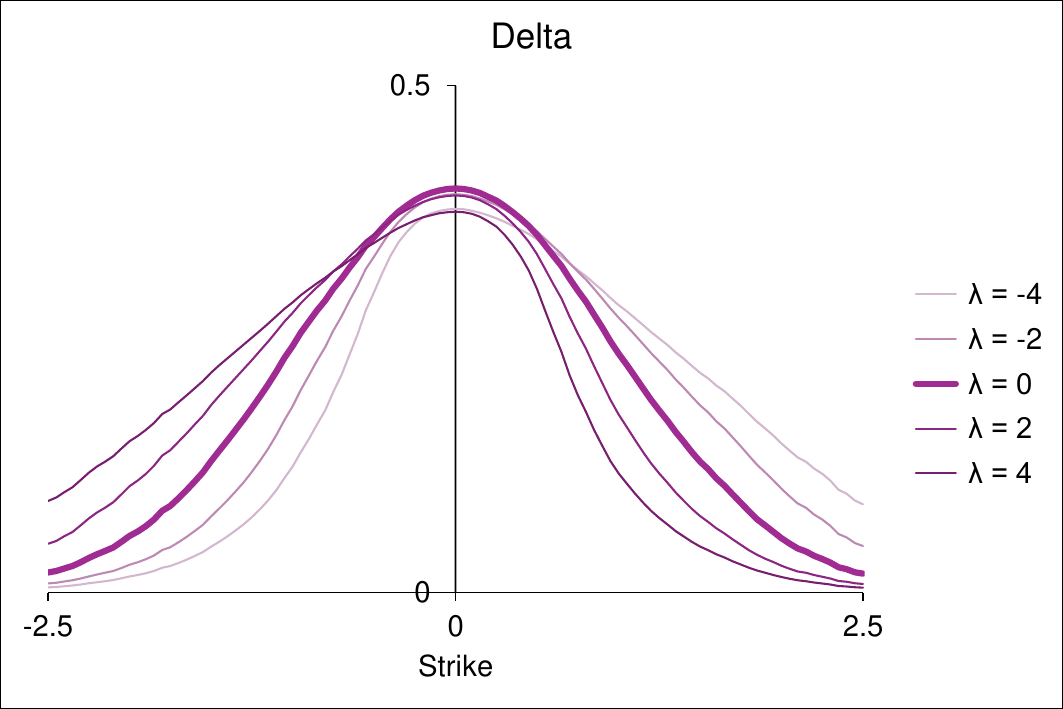}\\
\includegraphics[width=0.33\textwidth-\tabcolsep]{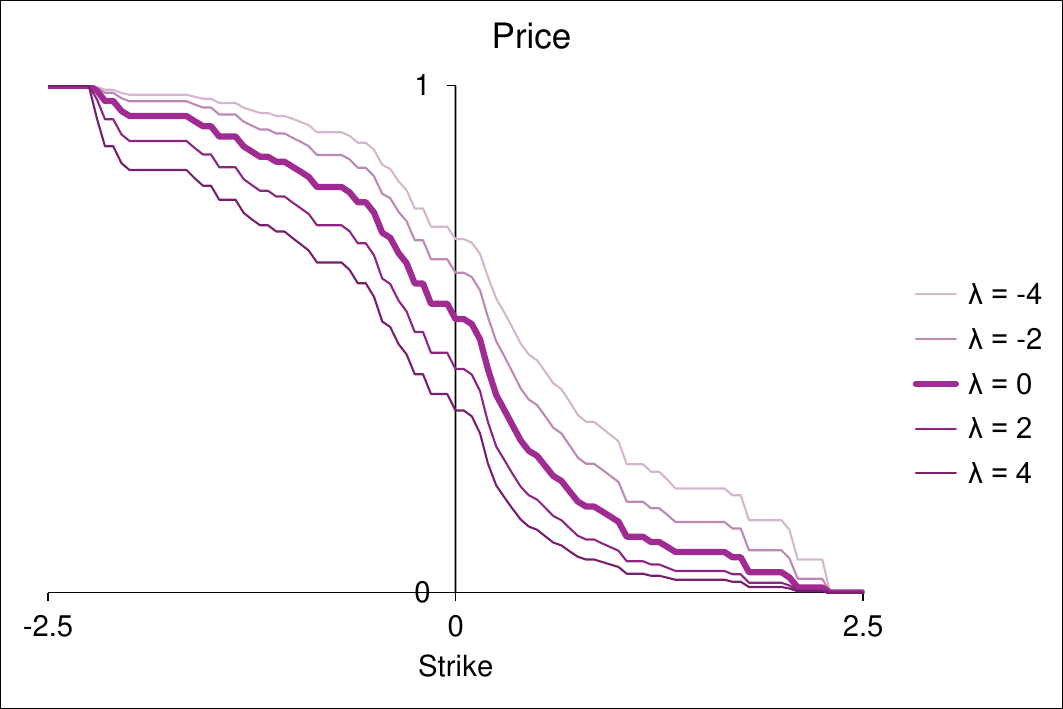}&
\includegraphics[width=0.33\textwidth-\tabcolsep]{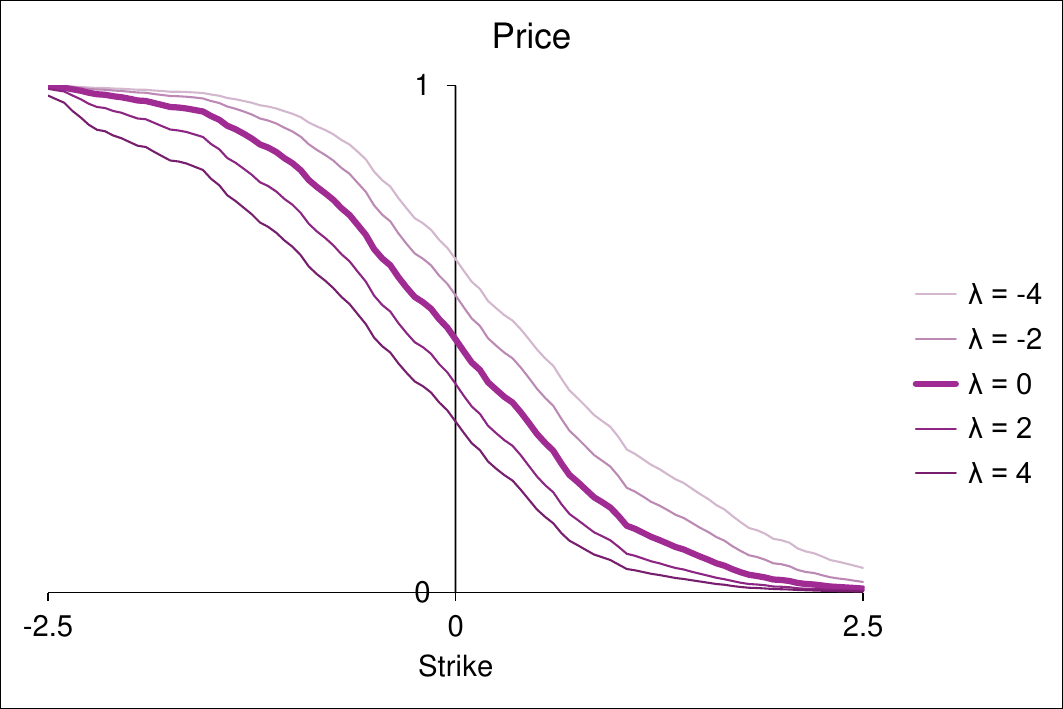}&
\includegraphics[width=0.33\textwidth-\tabcolsep]{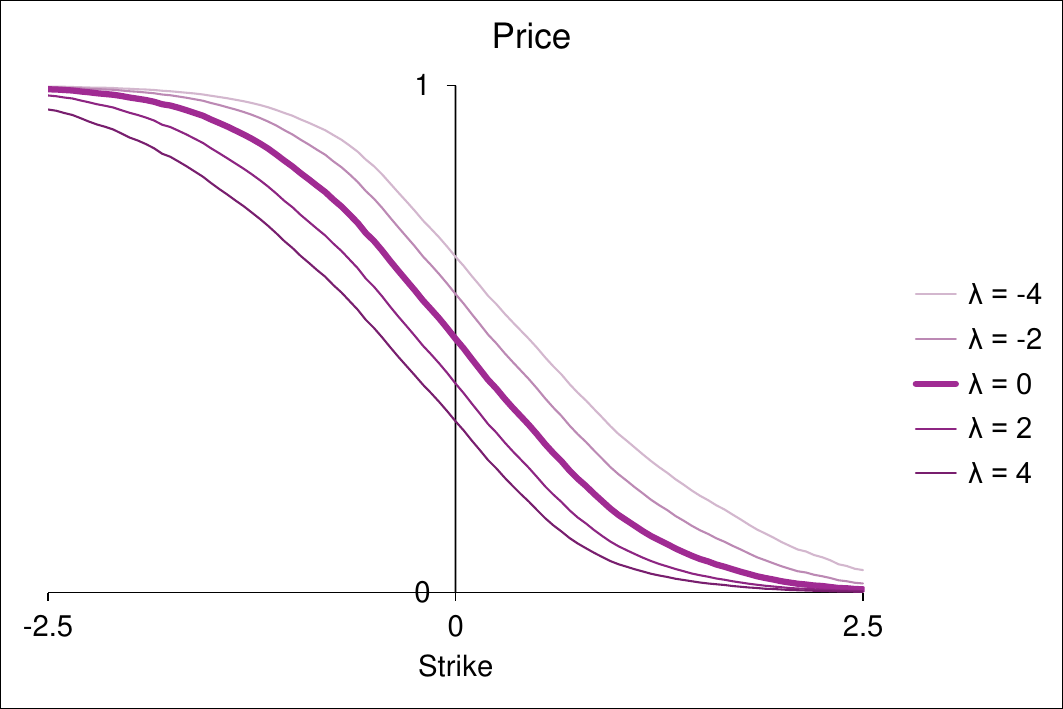}\\
Digital Option ($n=\num{100}$)&Digital Option ($n=\num{1000}$)&Digital Option ($n=\num{10000}$)
\end{tabular}
\caption{In these studies of discrete models on classical information, the final underlying price is uniformly distributed in the expectation measure on $n=\num{100}$, $\num{1000}$ and $\num{10000}$ samples from a standardised normal variable. The price measure is supported on the same samples with uniform probabilities that match the target initial underlying price $q=0$. Entropic risk optimisation identifies the mid and bid-offer for the price and the hedge ratio for a range of notionals $\lambda$.}
\label{fig:classical0}
\end{figure*}

\begin{figure*}[!pt]
\centering
\begin{tabular}{@{}C{0.33\textwidth-\tabcolsep}C{0.33\textwidth-\tabcolsep}C{0.33\textwidth-\tabcolsep}@{}}
\includegraphics[width=0.33\textwidth-\tabcolsep]{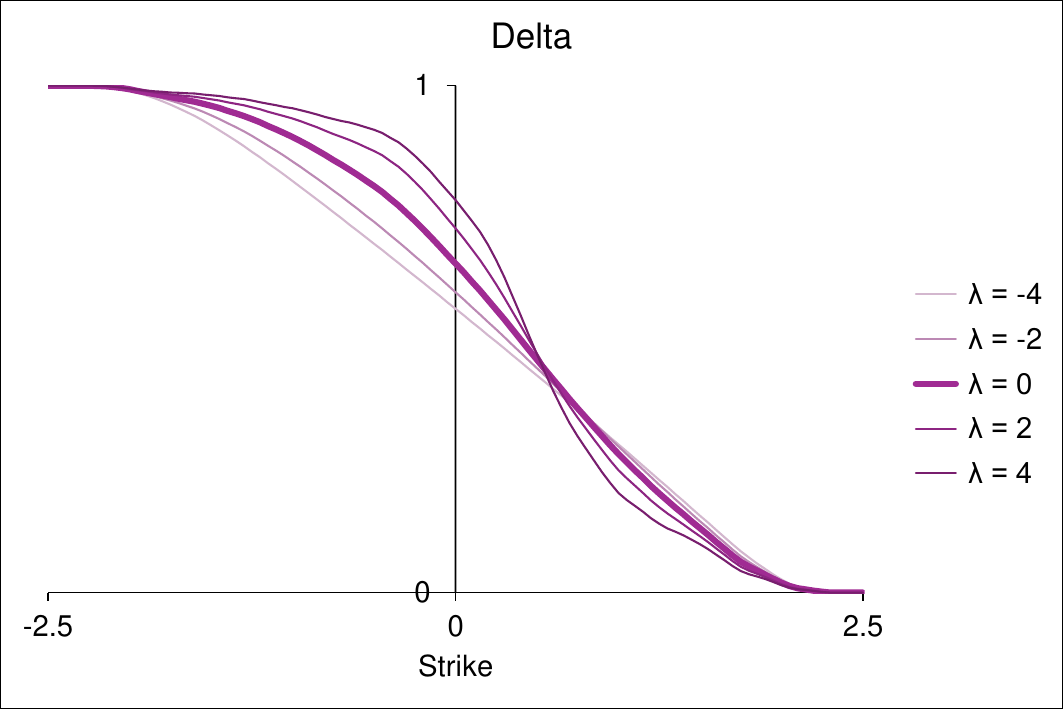}&
\includegraphics[width=0.33\textwidth-\tabcolsep]{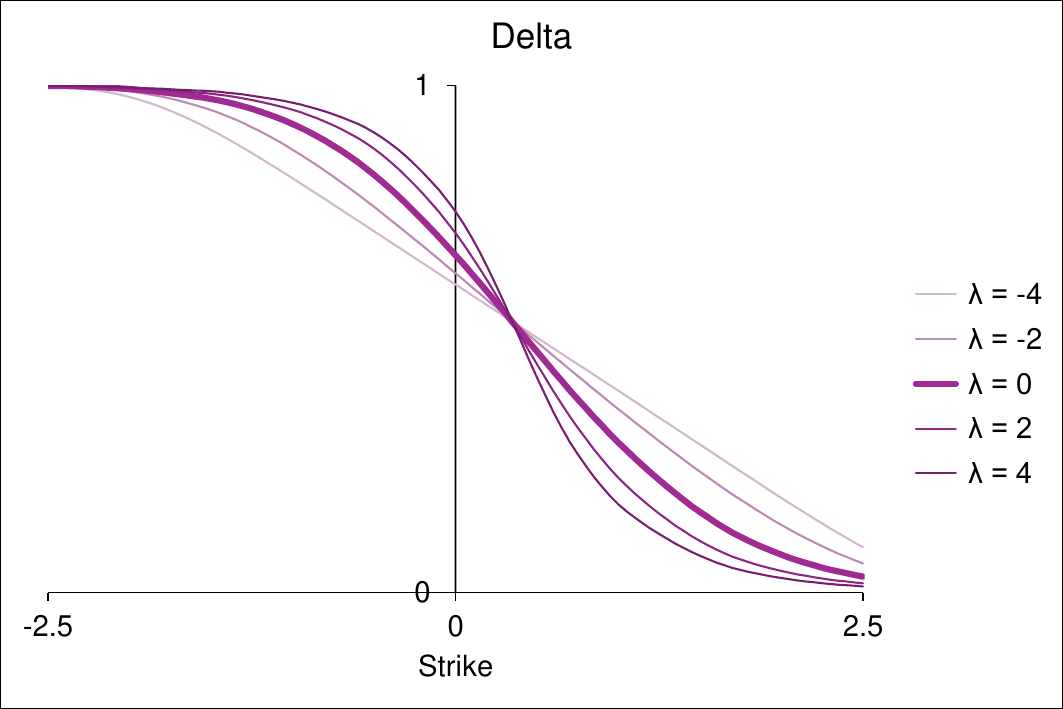}&
\includegraphics[width=0.33\textwidth-\tabcolsep]{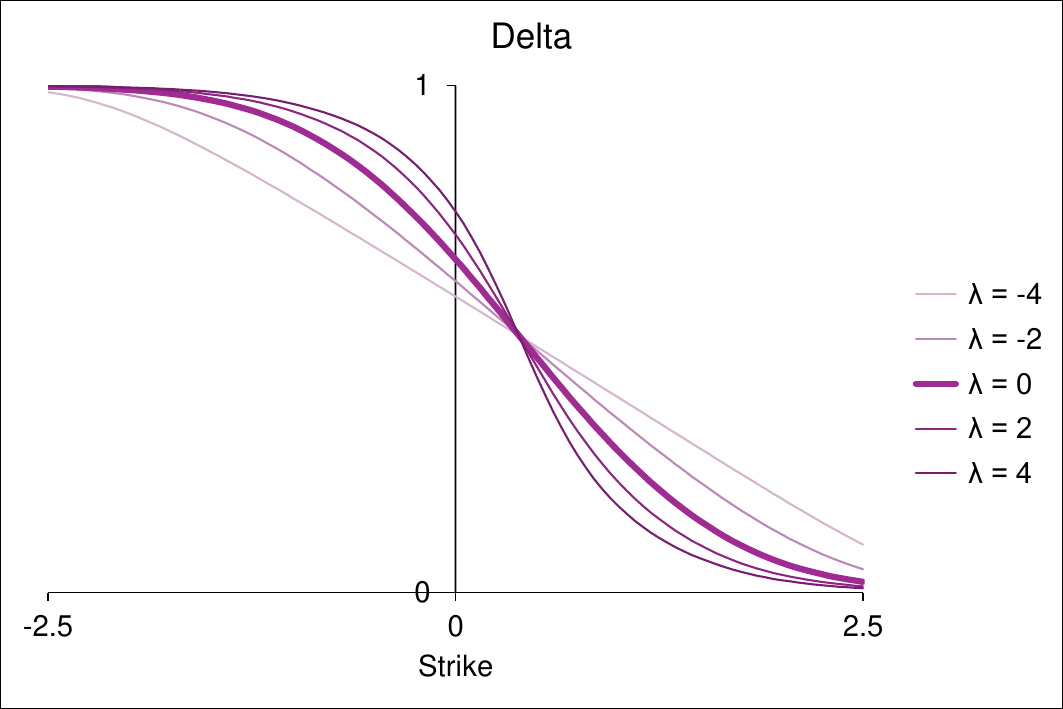}\\
\includegraphics[width=0.33\textwidth-\tabcolsep]{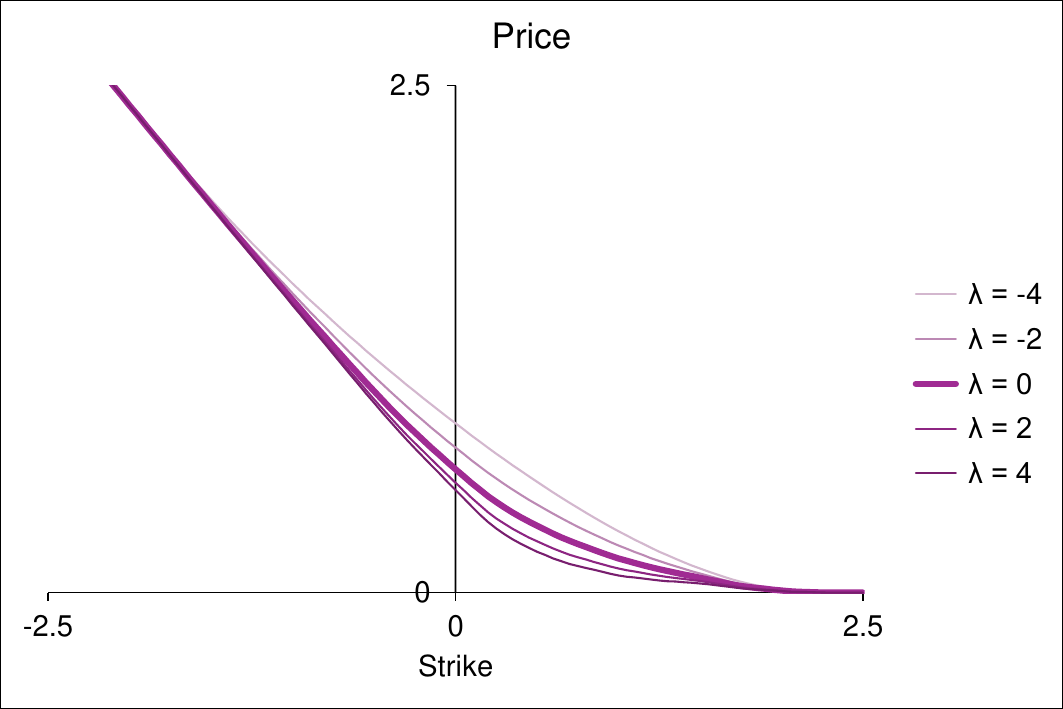}&
\includegraphics[width=0.33\textwidth-\tabcolsep]{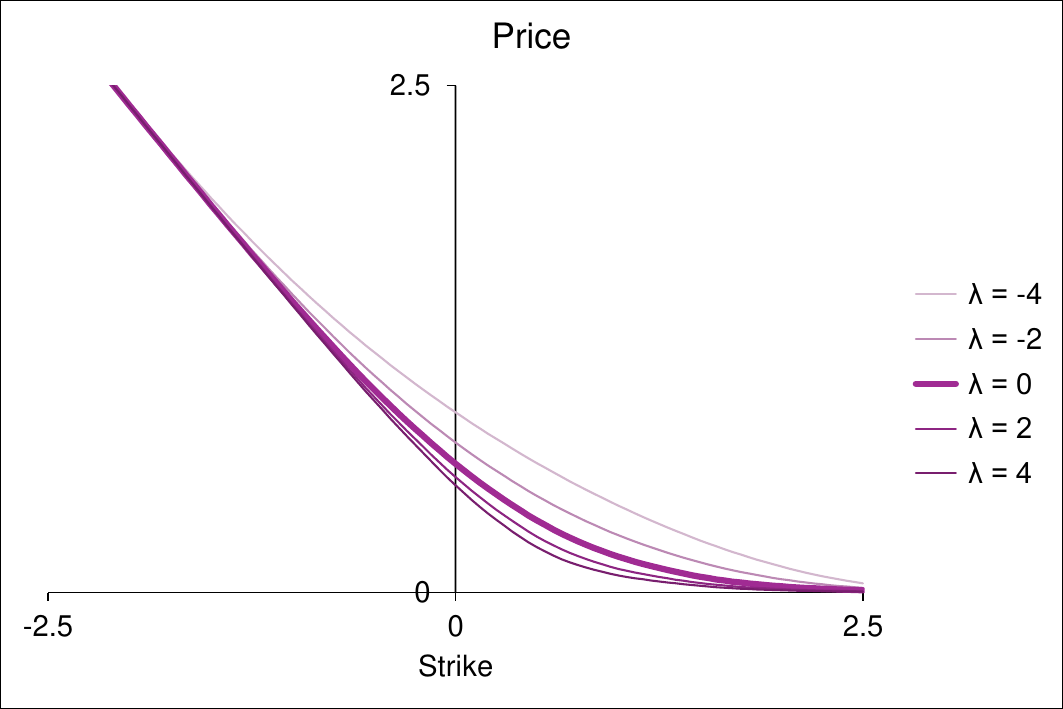}&
\includegraphics[width=0.33\textwidth-\tabcolsep]{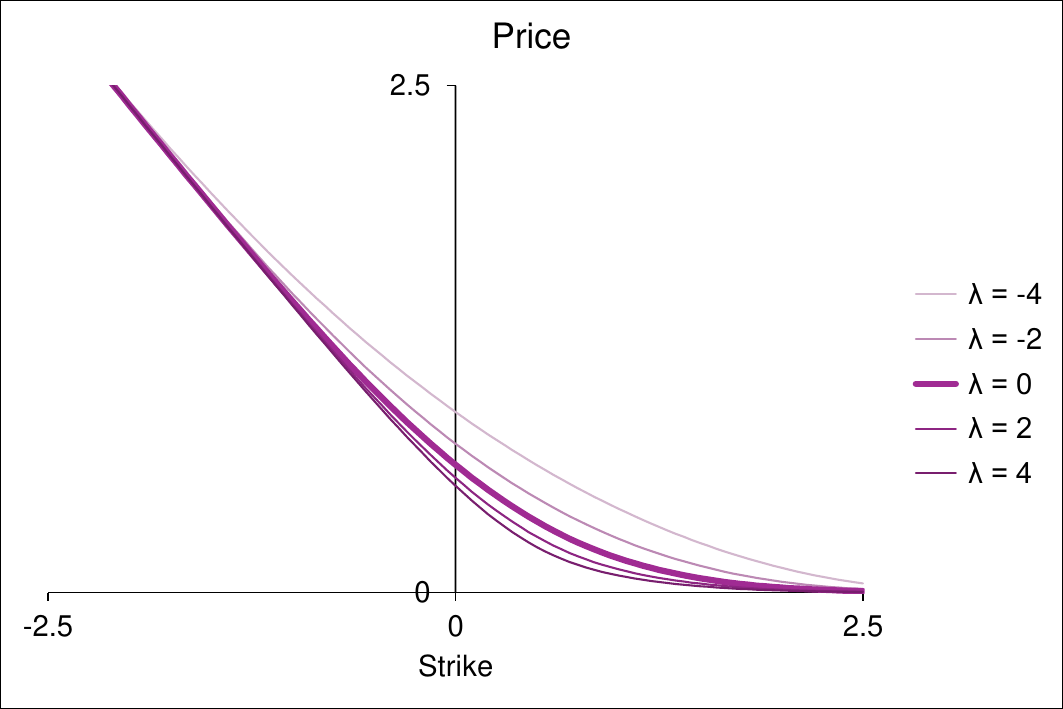}\\
\includegraphics[width=0.33\textwidth-\tabcolsep]{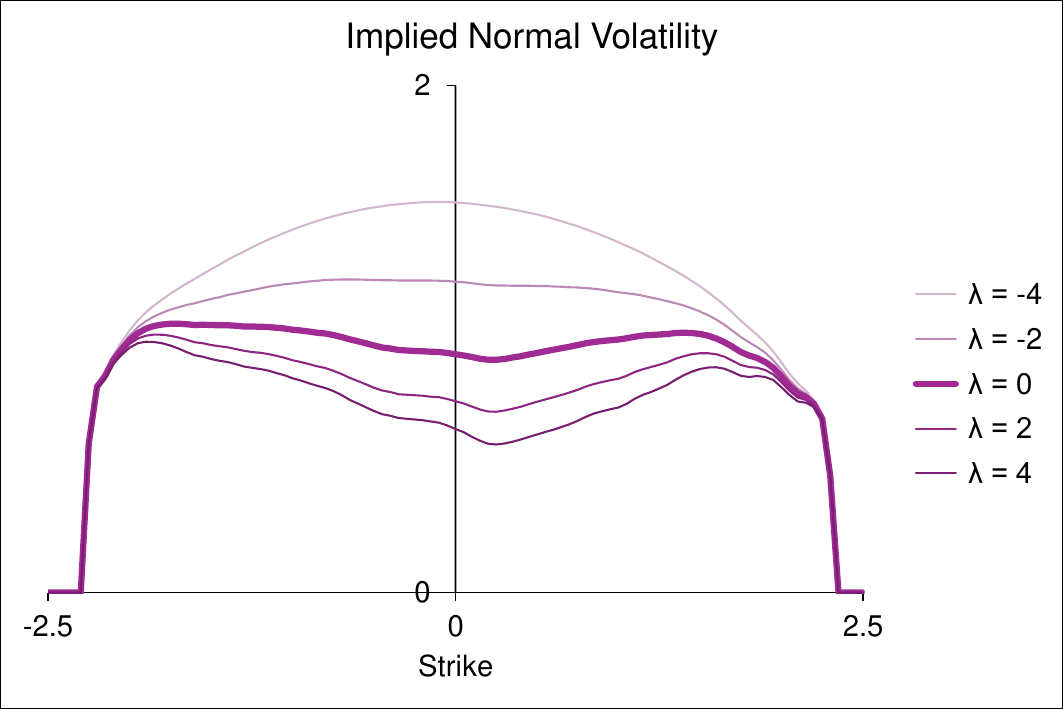}&
\includegraphics[width=0.33\textwidth-\tabcolsep]{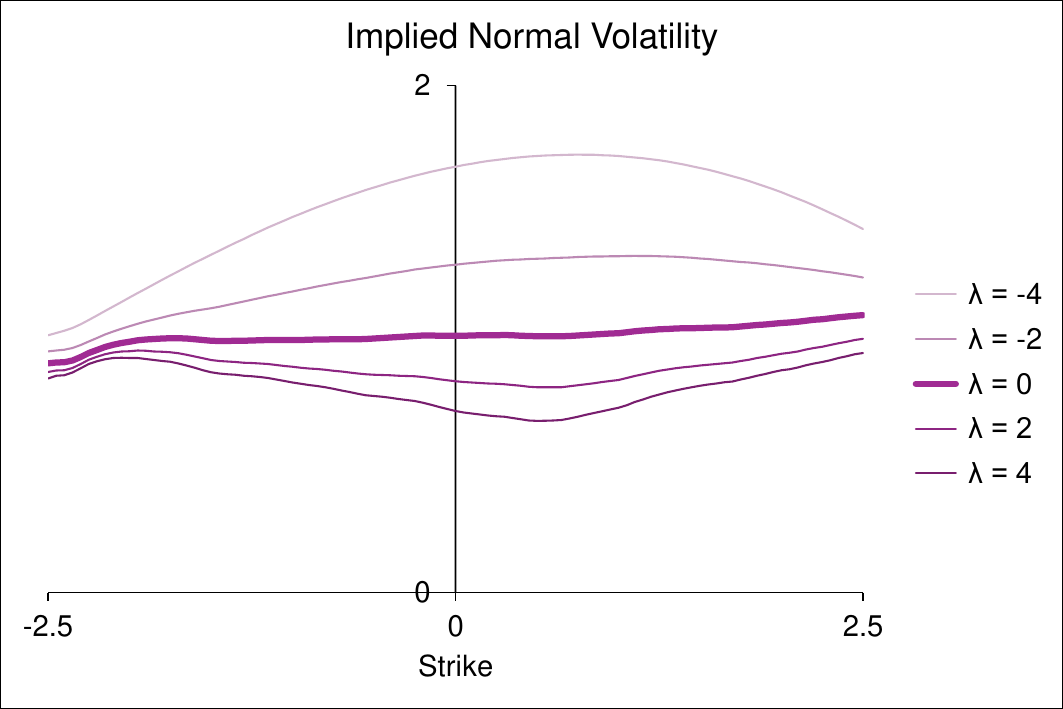}&
\includegraphics[width=0.33\textwidth-\tabcolsep]{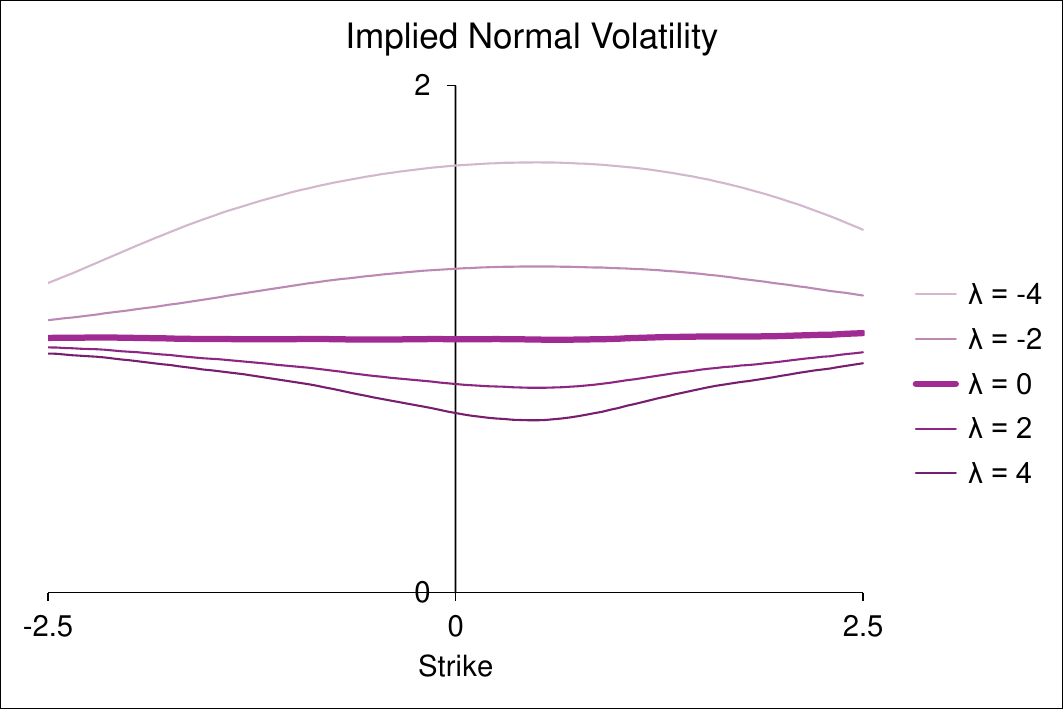}\\
Call Option ($n=\num{100}$)&Call Option ($n=\num{1000}$)&Call Option ($n=\num{10000}$)\\
\includegraphics[width=0.33\textwidth-\tabcolsep]{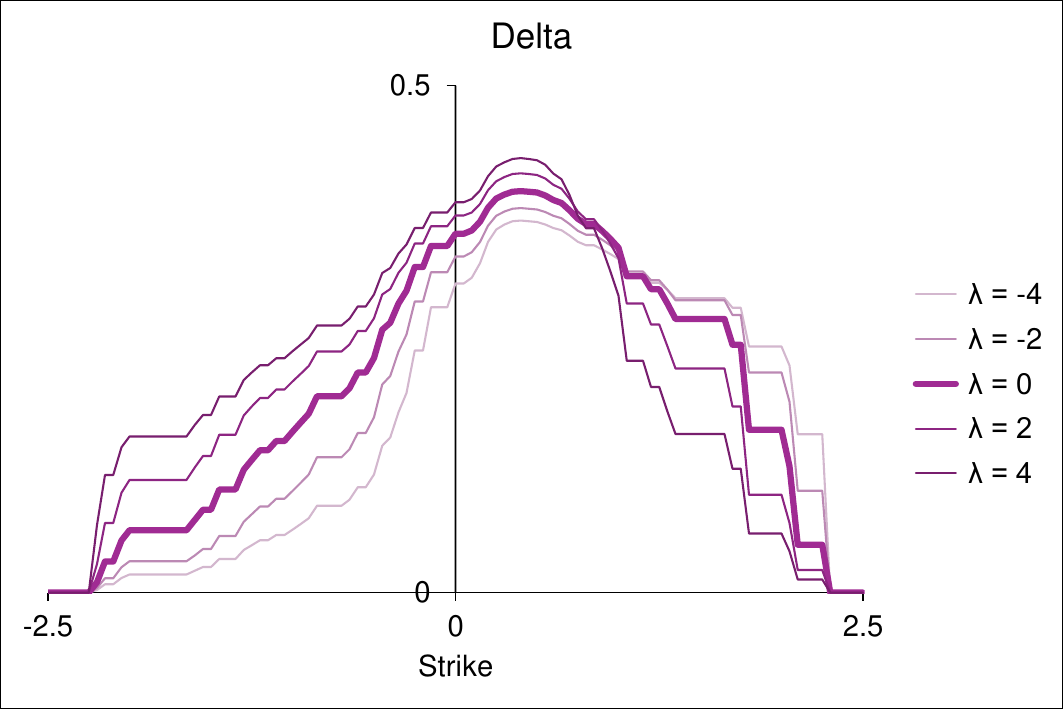}&
\includegraphics[width=0.33\textwidth-\tabcolsep]{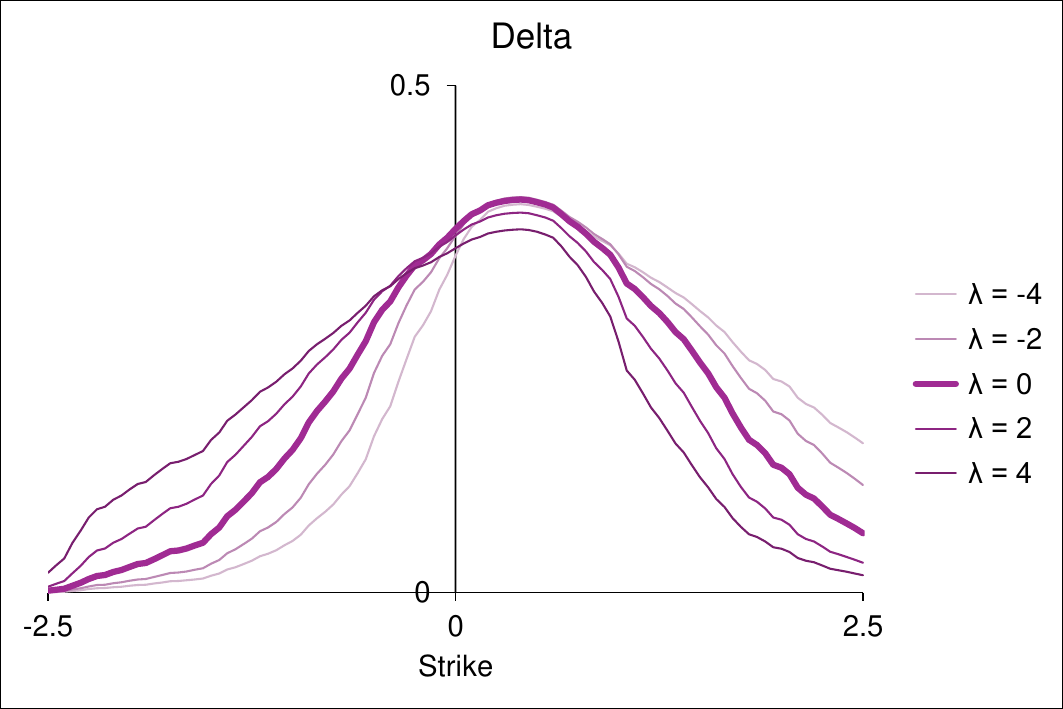}&
\includegraphics[width=0.33\textwidth-\tabcolsep]{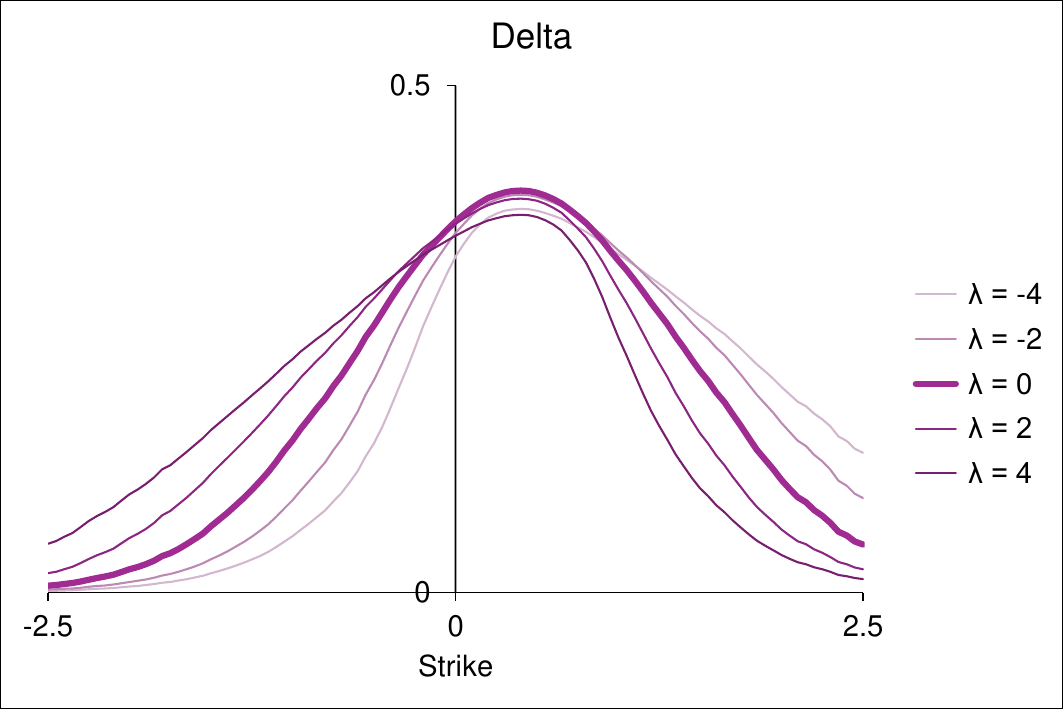}\\
\includegraphics[width=0.33\textwidth-\tabcolsep]{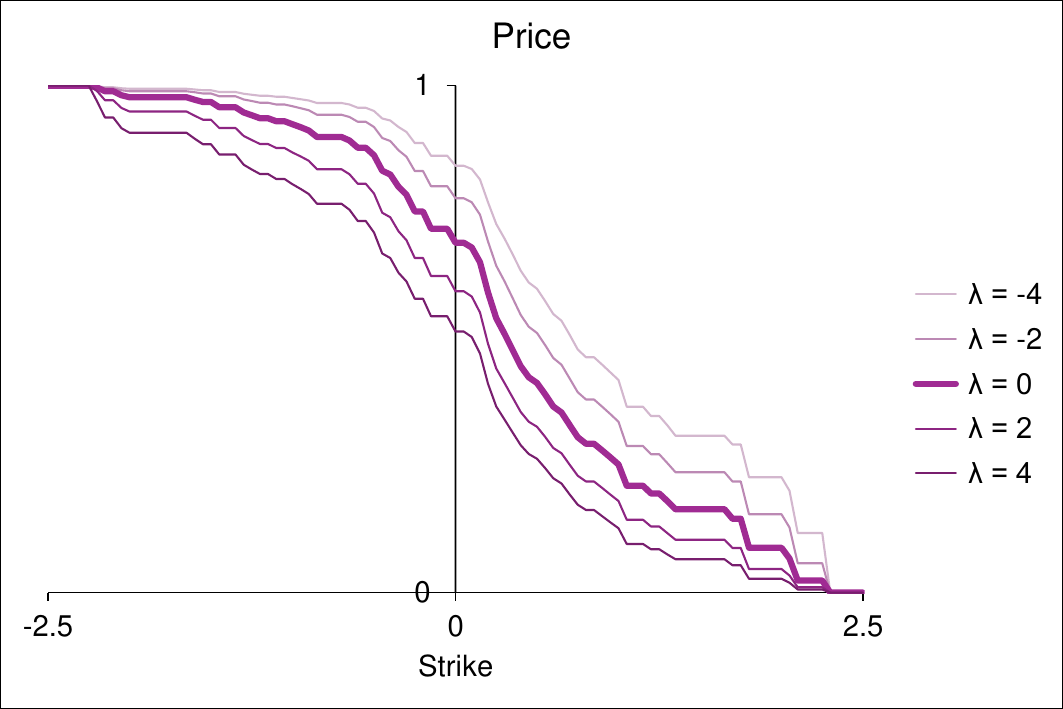}&
\includegraphics[width=0.33\textwidth-\tabcolsep]{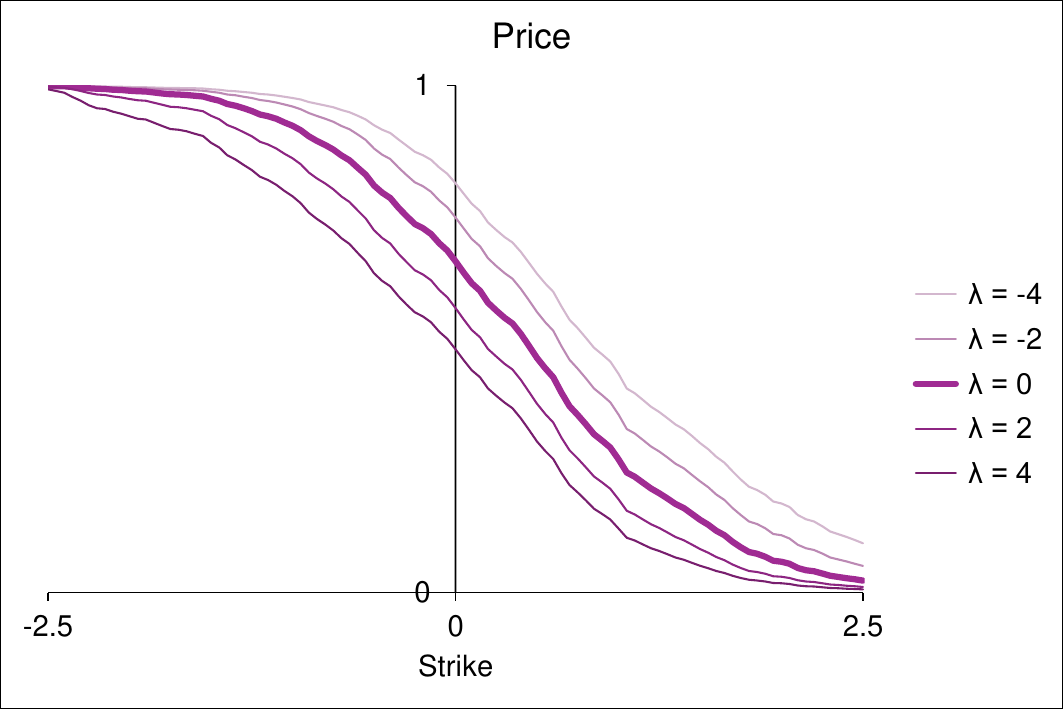}&
\includegraphics[width=0.33\textwidth-\tabcolsep]{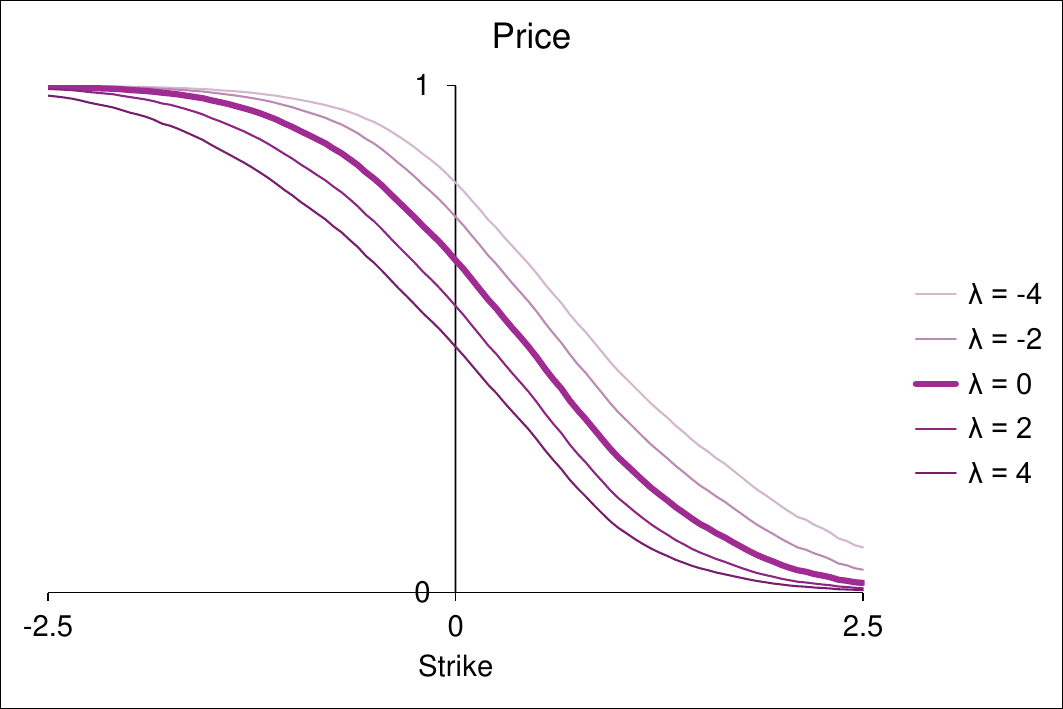}\\
Digital Option ($n=\num{100}$)&Digital Option ($n=\num{1000}$)&Digital Option ($n=\num{10000}$)
\end{tabular}
\caption{In these studies of discrete models on classical information, the final underlying price is uniformly distributed in the expectation measure on $n=\num{100}$, $\num{1000}$ and $\num{10000}$ samples from a standardised normal variable. The price measure is supported on the same samples with probabilities adjusted to match the target initial underlying price $q=0.4$. Entropic risk optimisation identifies the mid and bid-offer for the price and the hedge ratio for a range of notionals $\lambda$.}
\label{fig:classicalP4}
\end{figure*}

\begin{figure*}[!pt]
\centering
\begin{tabular}{@{}C{0.33\textwidth-\tabcolsep}C{0.33\textwidth-\tabcolsep}C{0.33\textwidth-\tabcolsep}@{}}
\includegraphics[width=0.33\textwidth-\tabcolsep]{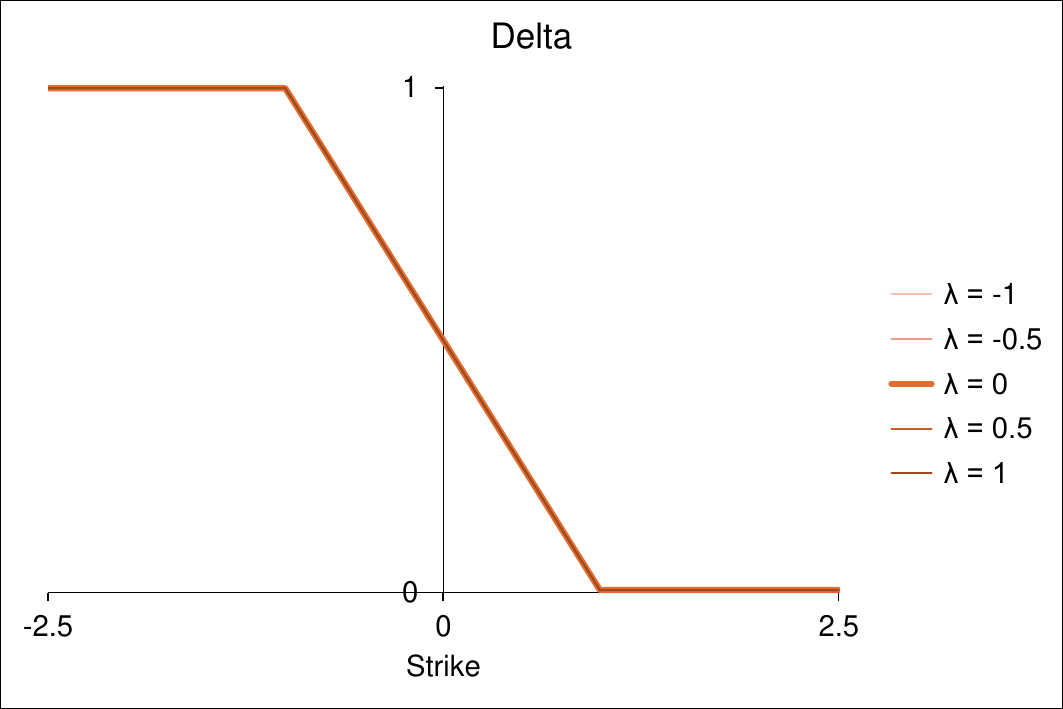}&
\includegraphics[width=0.33\textwidth-\tabcolsep]{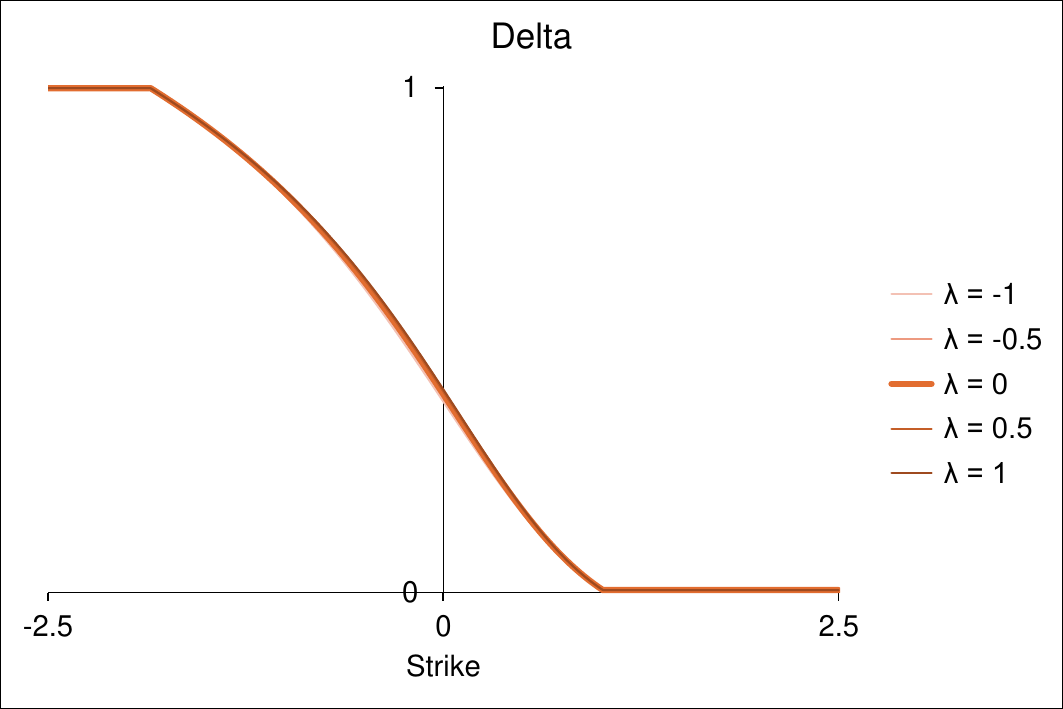}&
\includegraphics[width=0.33\textwidth-\tabcolsep]{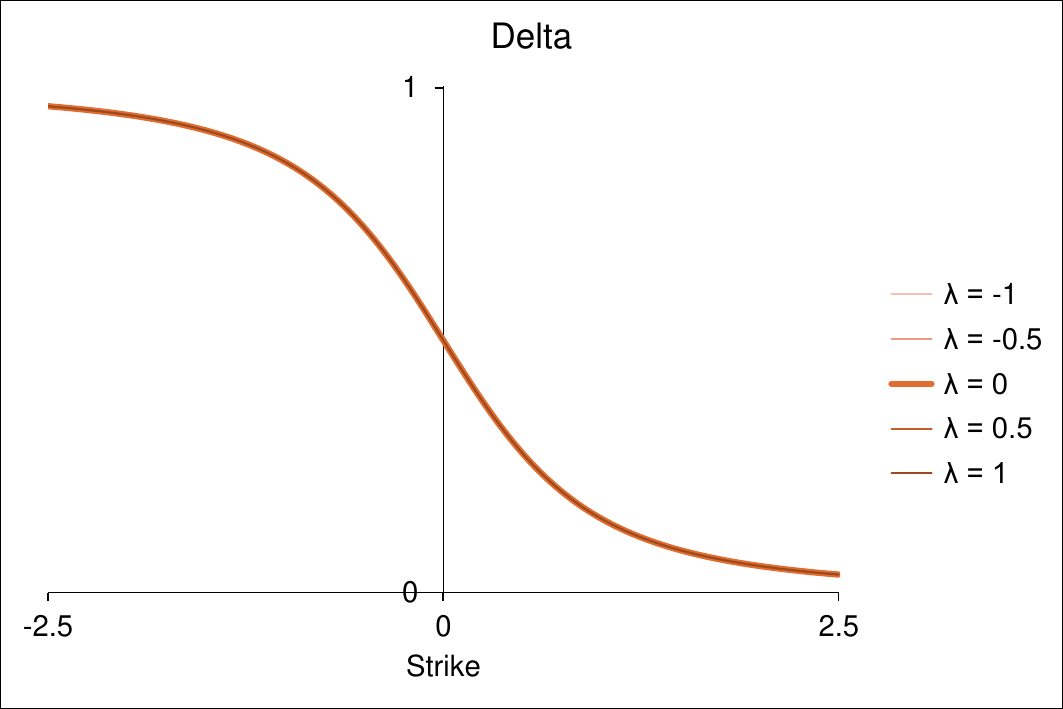}\\
\includegraphics[width=0.33\textwidth-\tabcolsep]{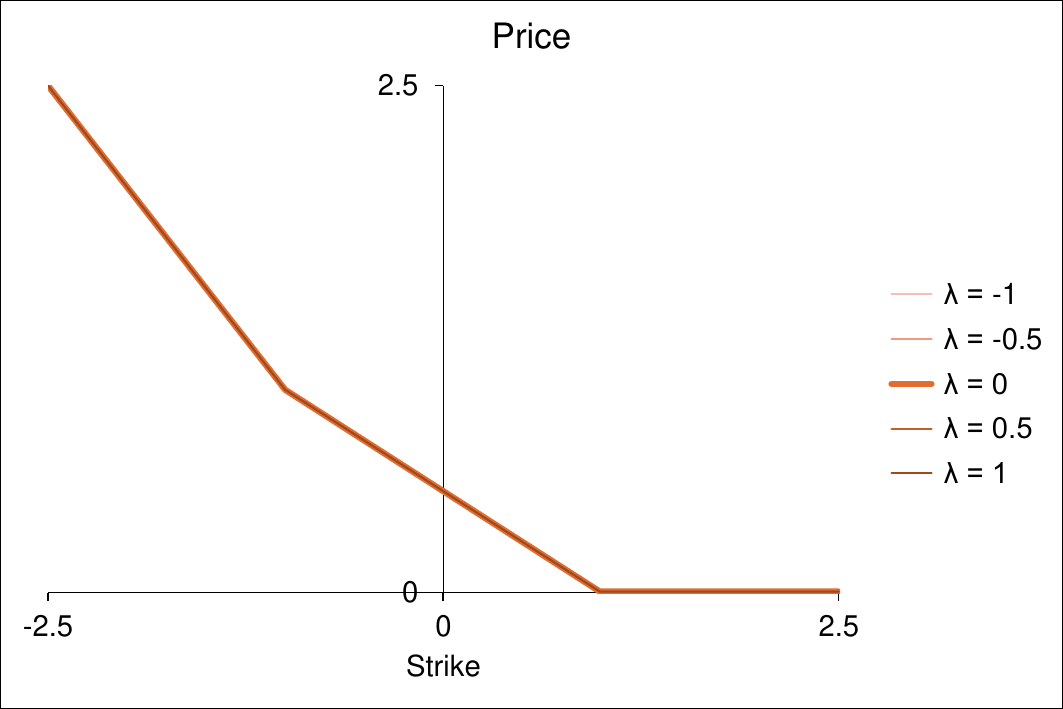}&
\includegraphics[width=0.33\textwidth-\tabcolsep]{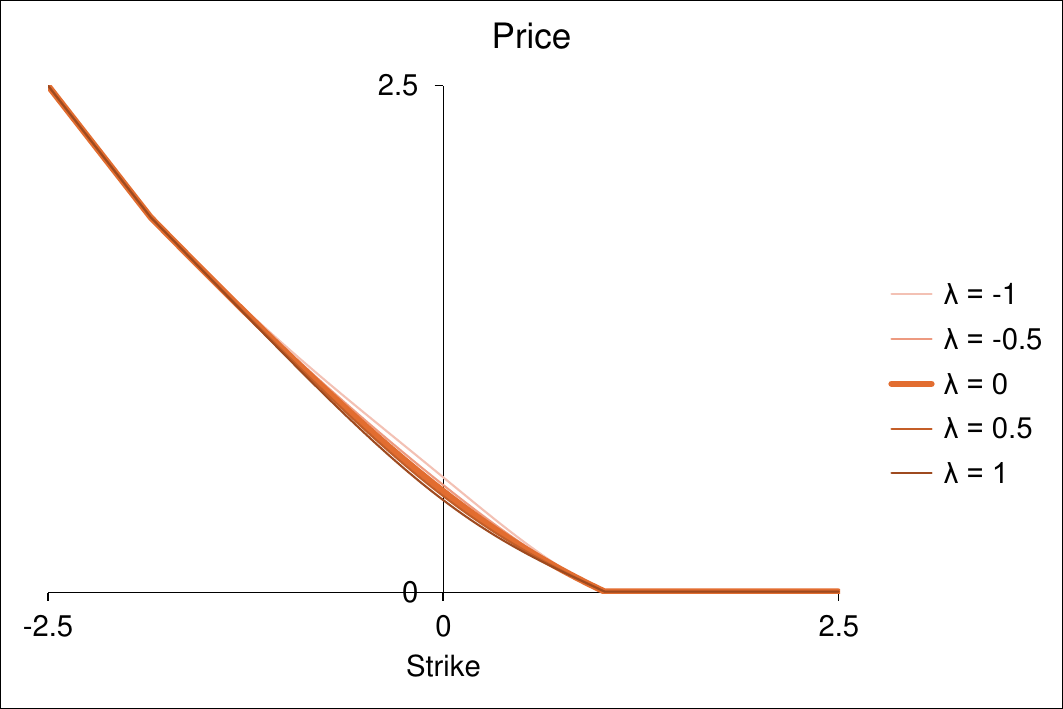}&
\includegraphics[width=0.33\textwidth-\tabcolsep]{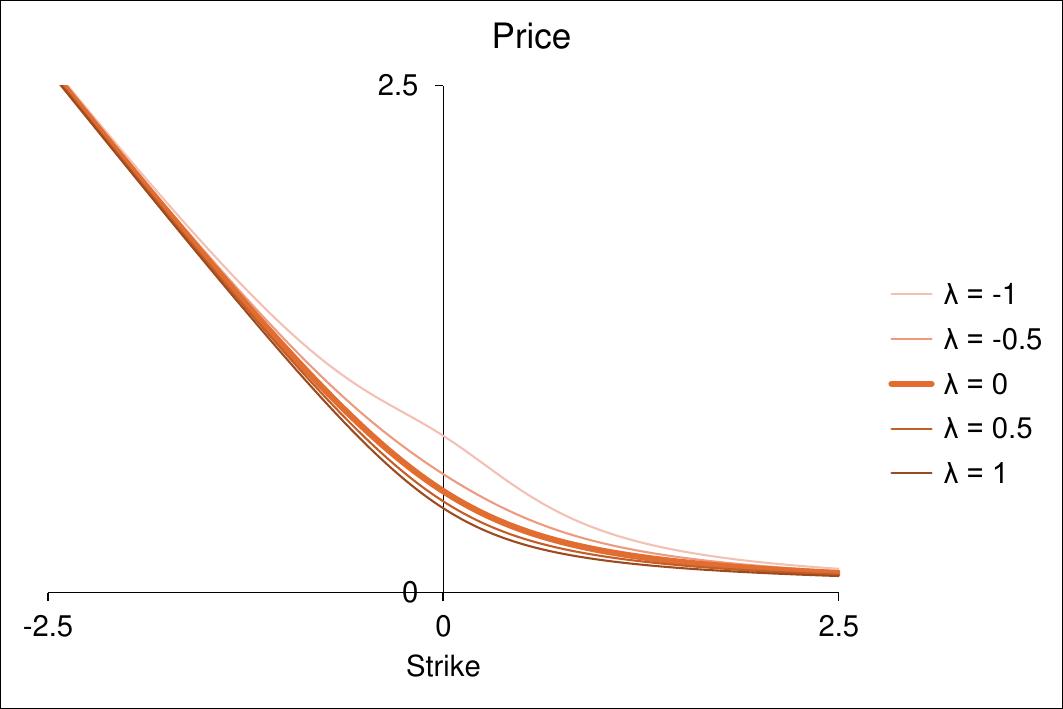}\\
\includegraphics[width=0.33\textwidth-\tabcolsep]{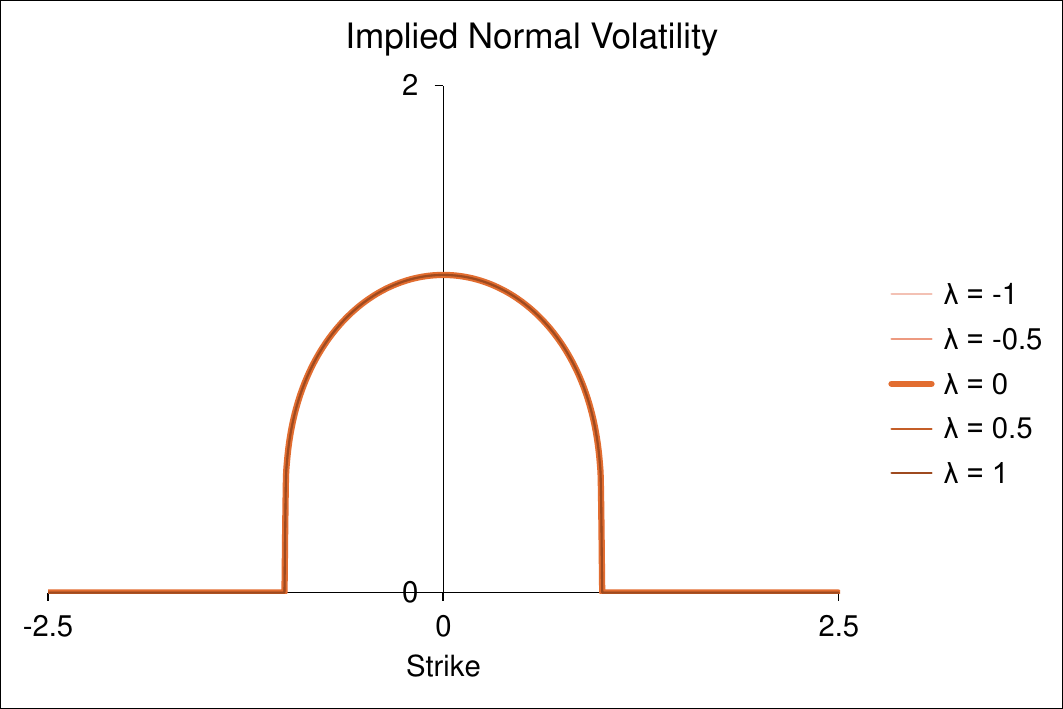}&
\includegraphics[width=0.33\textwidth-\tabcolsep]{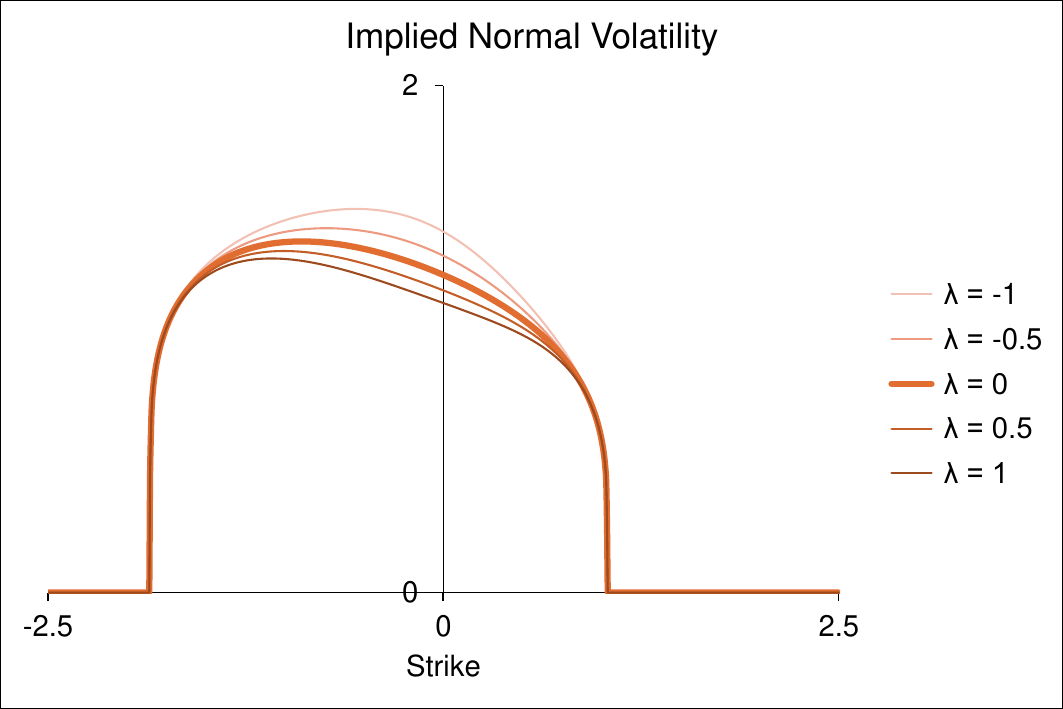}&
\includegraphics[width=0.33\textwidth-\tabcolsep]{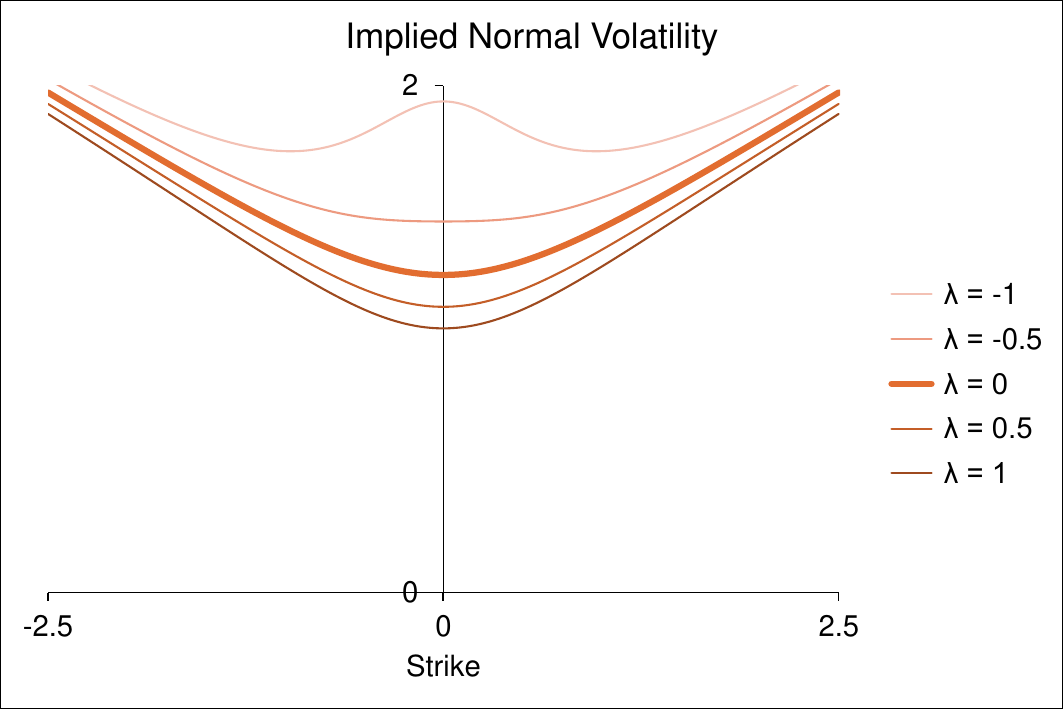}\\
Call Option ($\theta=0$)&Call Option ($\theta=\pi/5$)&Call Option ($\theta=\pi/4$)
\end{tabular}
\caption{In these studies of discrete models on quantum information, the final hedging price is uniformly distributed in the expectation measure on $n=2$ values $-1$ and $1$, and is decorrelated from the final funding price by applying the rotation $\theta=0$, $\pi/5$ and $\pi/4$ to its eigenbasis. Entropic risk optimisation identifies the mid and bid-offer for the price and the hedge ratio for a range of notionals $\lambda$.}
\label{fig:quantum}
\end{figure*}

Generalisations of this model are explored in the studies of \cref{fig:classicalprobs,fig:classicalM4,fig:classical0,fig:classicalP4,fig:quantum}. The model is developed for the final derivative price expressed as a function $P=f[Q-kB]$ of the spread portfolio $Q-kB$ at strike $k$, for the call option with $f[x]=x^+$ and the digital option with $f[x]=1_{x>0}$.

\subsection{Classical study}

The classical studies of \cref{fig:classicalprobs,fig:classicalM4,fig:classical0,fig:classicalP4} assume $b=1$ and $B=1$, and $Q$ is a diagonal matrix with eigenvalues $Q_i$ sampled from a standard normal variable, centralised to have zero sample mean. The matrix $P=f[Q-kB]$ is then diagonal with eigenvalues $P_i:=f[Q_i-k]$.

The calibration condition:
\begin{equation}
0=\sum_{i=1}^n(Q_i-q)\exp[-\phi Q_i]
\end{equation}
is solved for the unit optimal portfolio $\phi$, which generates the eigenstate probabilities:
\begin{equation}
\gamma_i:=\bfrac{\exp[-\phi Q_i]}{\sum\nolimits_{j=1}^n\exp[-\phi Q_j]}
\end{equation}
in the price measure. The hedge condition:
\begin{equation}
0=\sum_{i=1}^n\gamma_i(Q_i-q)\exp[-\alpha(P_i-\delta Q_i)]
\end{equation}
is then solved for the hedge portfolio $\delta$. Using these solutions, the price condition:
\begin{equation}
p=\delta q-\frac{1}{\alpha}\log\!\left[\sum_{i=1}^n\gamma_i\exp[-\alpha(P_i-\delta Q_i)]\right]
\end{equation}
generates the initial derivative price $p$.

Three cases are considered in the studies with $n=\num{100}$ (left column), $\num{1000}$ (middle column) and $\num{10000}$ (right column), for the call option $f[x]=x^+$ (top two rows) and the digital option $f[x]=1_{x>0}$ (bottom two rows). The derivative security cannot be exactly replicated by the funding and hedging securities, and entropic risk optimisation compensates for this model risk with a bid-offer around the mid price and a skew to the hedge ratio. In the study with $n=\num{100}$ samples, the support is contained within the interval (-2.5,2.5), and there can be no option value outside this range of strikes. Increasing the number of samples widens the support and brings the solution closer to the limit of the standard normal variable.

The support of the expectation measure thus plays a crucial role in determining hedge and price and also the extent to which unhedged risks impact the bid-offer. Entropic pricing does not allocate probability beyond the support of the expectation measure, and unrealistic assumptions on the support can mask model risks.

\subsection{Quantum study}

The quantum studies of \cref{fig:quantum} assume $b=1$ and $q=0$, and $B$ and $Q$ are matrices with $n=2$ eigenvalues:
\begin{align}
B&=1+\epsilon\begin{bmatrix}1&0\\0&-1\end{bmatrix} \\
Q&=\begin{bmatrix}\cos[\theta]&-\sin[\theta]\\\sin[\theta]&\cos[\theta]\end{bmatrix}\begin{bmatrix}1&0\\0&-1\end{bmatrix}\begin{bmatrix}\cos[\theta]&\sin[\theta]\\-\sin[\theta]&\cos[\theta]\end{bmatrix} \notag
\end{align}
The gap $\epsilon$ separates the eigenvalues of the funding security and the angle $\theta$ rotates the eigenbasis of the hedging security. The matrix $P=f[Q-kB]$ is obtained by first diagonalising the matrix $Q-kB$ and then applying the function $f$ to its eigenvalues:
\begin{align}
P&= \\
&\begin{bmatrix}\cos[\psi]&-\sin[\psi]\\\sin[\psi]&\cos[\psi]\end{bmatrix}\begin{bmatrix}P_+&0\\0&P_-\end{bmatrix}\begin{bmatrix}\cos[\psi]&\sin[\psi]\\-\sin[\psi]&\cos[\psi]\end{bmatrix} \notag
\end{align}
with:
\begin{align}
\psi&:=\frac{1}{2}\arctan\!\left[\frac{\sin[2\theta]}{\cos[2\theta]-\epsilon k}\right] \\
Q_\pm&:=\pm\frac{\sin[2\theta]}{\sin[2\psi]} \notag \\
P_\pm&:=f[Q_\pm-k] \notag
\end{align}
The non-linear relationship between a matrix and its eigenvalues is exploited to create a complex dependence of the option payoff on the strike. Fine tuning the rotation of the eigenbases can then be used to fit the implied volatility smile.

The calibration condition is trivially solved with $\phi=0$, and the hedge and price conditions become:
\begin{align}
\delta&=\hat{P}\cos[2(\psi-\theta)]-\epsilon p\cos[2\theta] \\
p&=\bar{P}-\frac{1}{\alpha}\log\cosh[\alpha(\hat{P}\sin[2(\psi-\theta)]+\epsilon p\sin[2\theta])] \notag
\end{align}
where $\bar{P}:=(P_++P_-)/2$ and $\hat{P}:=(P_+-P_-)/2$. The price condition is first solved for the initial derivative price $p$. Using this solution, the hedge condition generates the hedge portfolio $\delta$.

Three cases are considered in the studies with $\theta=0$ (left column), $\pi/5$ (middle column) and $\pi/4$ (right column), with $\epsilon=1-\cos[2\theta]$ in each case, for the call option $f[x]=x^+$. Contrasting with the classical studies where thousands of eigenstates are needed for acceptable convergence, the inclusion of non-commutativity between the final funding and hedging prices generates a credible model for the initial derivative price with as few as two eigenstates.

With two securities available for funding and hedging, the derivative security is exactly replicated in the classical binomial model with zero rotation. Heisenberg uncertainty breaks market completeness, and entropic risk optimisation creates a compensating bid-offer around the mid price in the quantum binomial model when the rotation is non-zero. At the maximum rotation $\theta=\pi/4$, the implied volatility smile has continuous support on the full line, a remarkable property for a model whose final underlying prices are binomial.

\section{Entropic XVA}

An unconstrained strategy that maximises the mean does not incentivise risk management and results in a portfolio with potentially disastrous downside risks. For this reason, a strategy that instead aims to maximise the risk-adjusted mean creates a more balanced portfolio that can be fine-tuned to the risk appetite of the investor.

It would appear from this statement that risk management is an investment choice, with the risk appetite a free parameter under the control of the investor. This is not the case. Without risk management there is a material risk that the investor will not meet their contractual commitments. At the very least, this potentially tarnishes the reputation of the investor as a reliable counterparty, with consequent impact on the prices they will be offered. In regulated markets, risk management is also a legal requirement with detailed and specific obligations for margin and capital.

\begin{figure*}[!pt]
\setlength{\abovecaptionskip}{5pt}
\setlength{\belowcaptionskip}{20pt}
\centering
\includegraphics[width=0.8\textwidth]{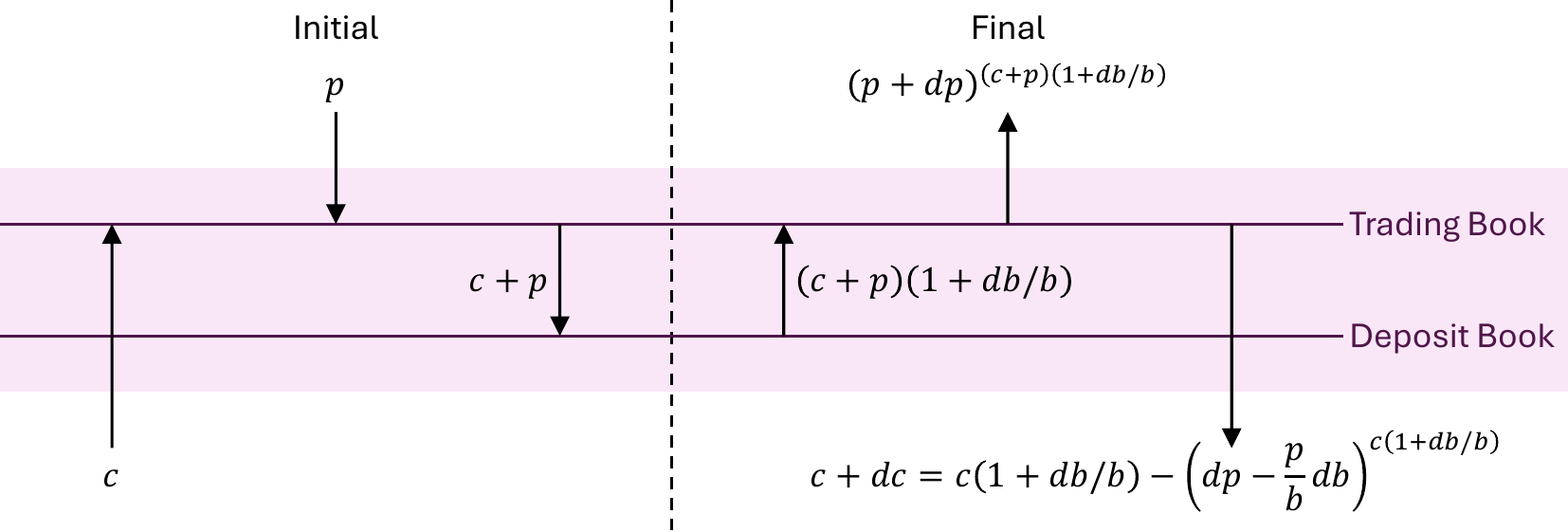}
\caption{In this example, the trader receives the initial derivative price $p$ and deploys additional capital $c$ to guarantee settlement, with margin $m:=c+p\ge0$. These proceeds are deposited in a margin account with unit price $b>0$. The final derivative price $p+dp$ is settled up to the value $(c+p)(1+db/b)$ of the margin account, with any remainder returned to the trader. Default occurs when the margin account is insufficient to meet the contractual obligation. In the expression for the final settlement, the cap is denoted by the superscript.}
\label{fig:cashflowsA}
\includegraphics[width=0.8\textwidth]{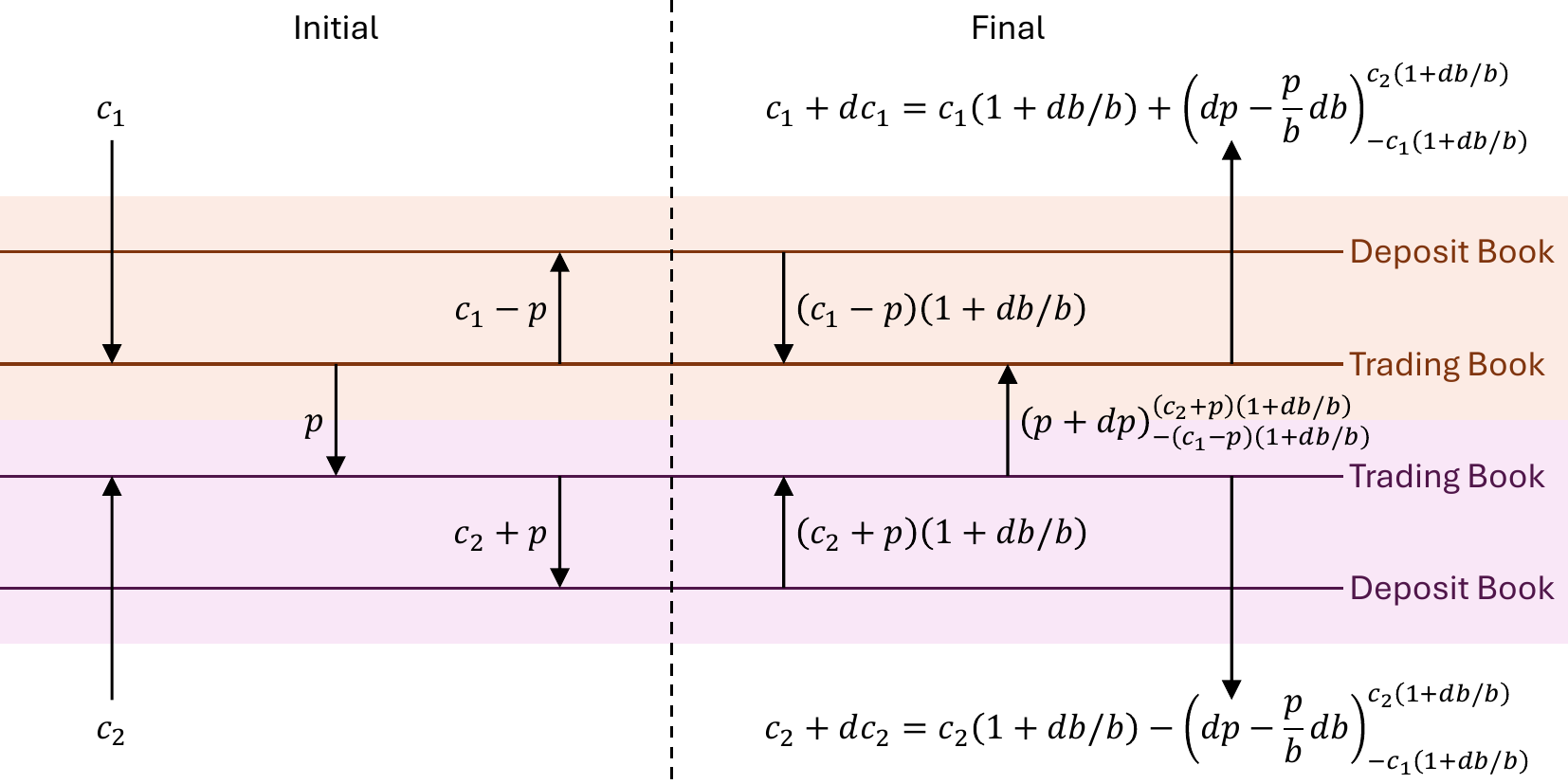}
\caption{Two counterparties deploy capital $c_1$ and $c_2$ respectively as their guarantee of derivative settlement, with margin $m_1:=c_1-p\ge0$ and $m_2:=c_2+p\ge0$, and they default on their contractual obligation if this is insufficient. Net of all considerations, the counterparties share the combined proceeds from the margin accounts, and the final return on capital is floored and capped by the respective default events. In the expressions for the final settlement and capital, the floor is denoted by the subscript and the cap is denoted by the superscript.}
\label{fig:cashflowsB}
\includegraphics[width=0.8\textwidth]{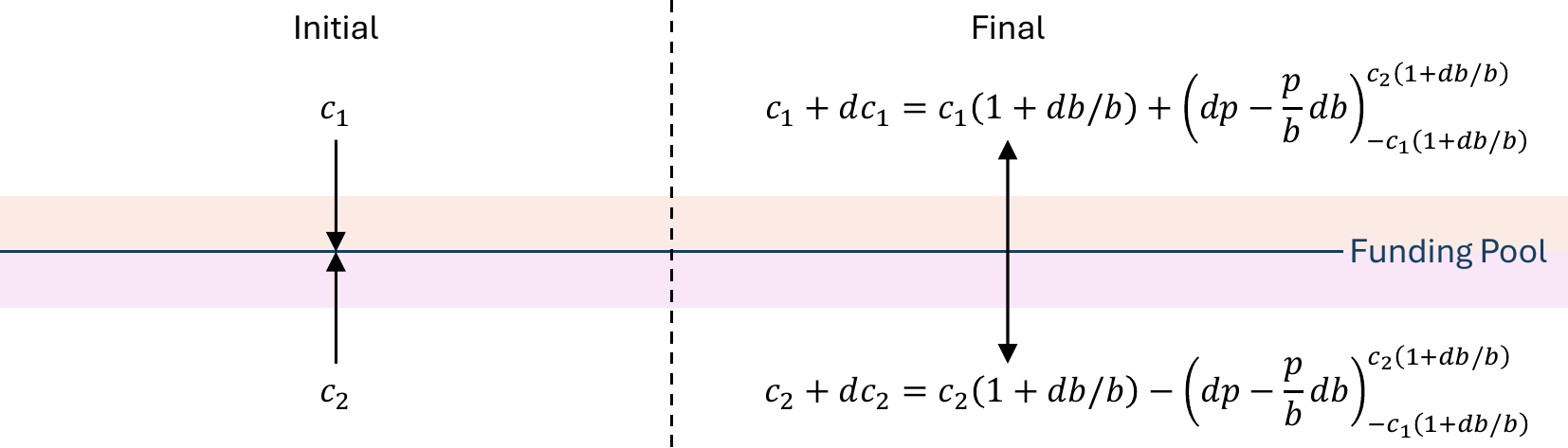}
\caption{Compressing internal settlements, the trade is simplified as a funding pool with contributions from both counterparties shared according to the contractual settlement of the derivative security adjusted for default. Using the entropic margin model, these settlements can be encoded in a smart contract for decentralised distribution.}
\label{fig:cashflowsC}
\end{figure*}

\subsection{Margined derivatives}

If the true costs of risk management are fully accounted in the return, reckless investments are naturally penalised by their associated margin costs. Returning to the bilateral model outlined in the introduction, suppose that two counterparties deploy capital $c_1$ and $c_2$ to a derivative trade with initial price $p$. Excess capital is deposited in the margin account with unit price $b>0$. The derivative is settled from the margin account with no additional funds injected, creating the risk that insufficient capital results in default by one of the counterparties.

As shown in \cref{fig:cashflowsA,fig:cashflowsB,fig:cashflowsC}, net of all considerations the capital returned from the derivative is:
\begin{align}
c_1+dc_1&=c_1(1+db/b)+\left(dp-\frac{p}{b}db\right)_{-c_1(1+db/b)}^{c_2(1+db/b)} \\
c_2+dc_2&=c_2(1+db/b)-\left(dp-\frac{p}{b}db\right)_{-c_1(1+db/b)}^{c_2(1+db/b)} \notag
\end{align}
where the subscript denotes the floor and the superscript denotes the cap. With all the internal settlements compressed, the structure reduces to a funding pool with initial capital $(c_1+c_2)$ deposited in the margin account and final capital $(c_1+c_2)(1+db/b)$ redistributed to the counterparties. Without the derivative settlement, the final capital is distributed according to the return on the initial contribution, with $c_1(1+db/b)$ returned to the first counterparty and $c_2(1+db/b)$ returned to the second counterparty. Introducing the derivative settlement, this distribution of capital is zero-sum adjusted by the funded derivative return up to the amount in the margin account. Default then occurs when the margin account has insufficient funds to meet the contractual settlement.

With the return on capital thus defined, entropic pricing can be applied. Each counterparty assesses the performance of the derivative security relative to their respective underlying markets with prices $q_1$ and $q_2$, using their expectation measures $\Exp_1$ and $\Exp_2$ and risk aversions $\alpha_1$ and $\alpha_2$. The hedge conditions:
\begin{align}
0&=\Exp_1^{\phi_1}[(dq_1-q_1r_1\,dt)\exp[-\alpha_1(dc_1-\delta_1\cdot dq_1)]] \\
0&=\Exp_2^{\phi_2}[(dq_2-q_2r_2\,dt)\exp[-\alpha_2(dc_2-\delta_2\cdot dq_2)]] \notag
\end{align}
are solved with the self-funding conditions $\delta_1\cdot q_1=c_1$ and $\delta_2\cdot q_2=c_2$ to generate the hedge portfolios $\delta_1$ and $\delta_2$, where $\phi_1$ and $\phi_2$ are the unit optimal portfolios in the underlying markets. The price conditions derive the expressions:
\begin{align}
0={}&\Exp_1^{\phi_1}[\alpha_1]\!\left[-\left(\delta_1\cdot dq_1-\frac{c_1}{b}db\right)\vphantom{\left(dp-\frac{p}{b}db\right)_{-c_1(1+db/b)}^{c_2(1+db/b)}}\right. \\
&\qquad\qquad+\left.\left(dp-\frac{p}{b}db\right)_{-c_1(1+db/b)}^{c_2(1+db/b)}\right] \notag \\
0={}&\Exp_2^{\phi_2}[\alpha_2]\!\left[-\left(\delta_2\cdot dq_2-\frac{c_2}{b}db\right)\vphantom{\left(dp-\frac{p}{b}db\right)_{-c_1(1+db/b)}^{c_2(1+db/b)}}\right. \notag \\
&\qquad\qquad-\left.\left(dp-\frac{p}{b}db\right)_{-c_1(1+db/b)}^{c_2(1+db/b)}\right] \notag
\end{align}
which are solved for the threshold prices $p_1$ and $p_2$ of the two counterparties. The trade is viable when the initial price $p$ is in the range $p_2\le p\le p_1$, and is therefore not viable when $p_1<p_2$. Price is thus constrained to the sector bounded by these viability thresholds, equilibrating to a price within the range.

\subsection{Entropic XVA}

\begin{figure*}[!pt]
\centering
\begin{tabular}{@{}C{0.33\textwidth-\tabcolsep}C{0.33\textwidth-\tabcolsep}C{0.33\textwidth-\tabcolsep}@{}}
\includegraphics[width=0.33\textwidth-\tabcolsep]{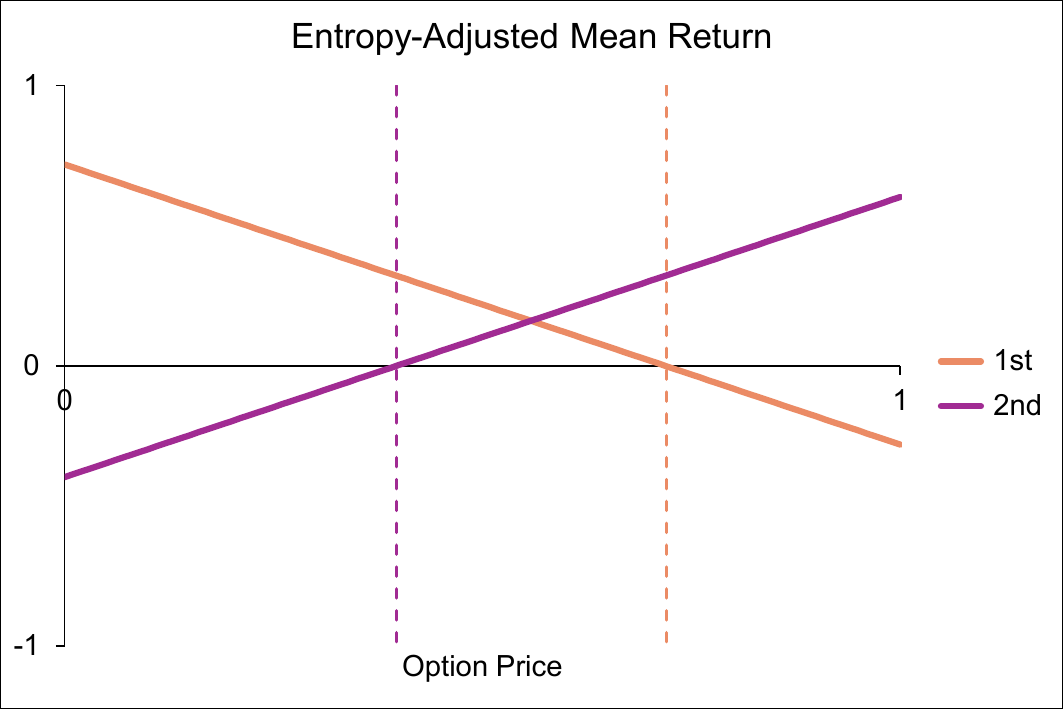}&
\includegraphics[width=0.33\textwidth-\tabcolsep]{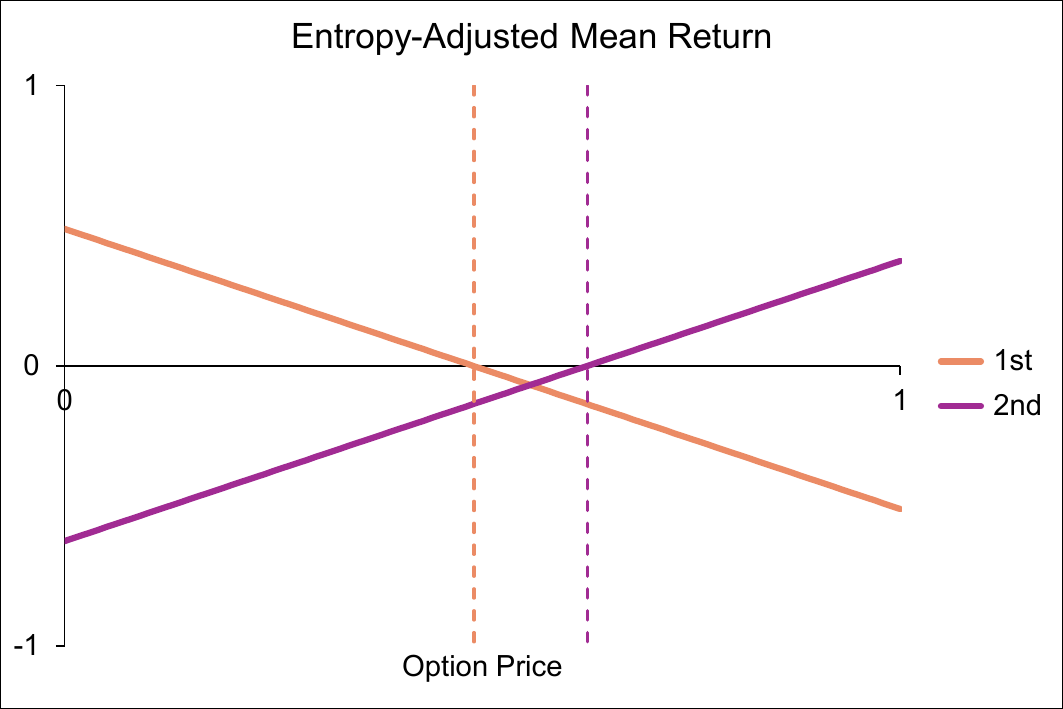}&
\includegraphics[width=0.33\textwidth-\tabcolsep]{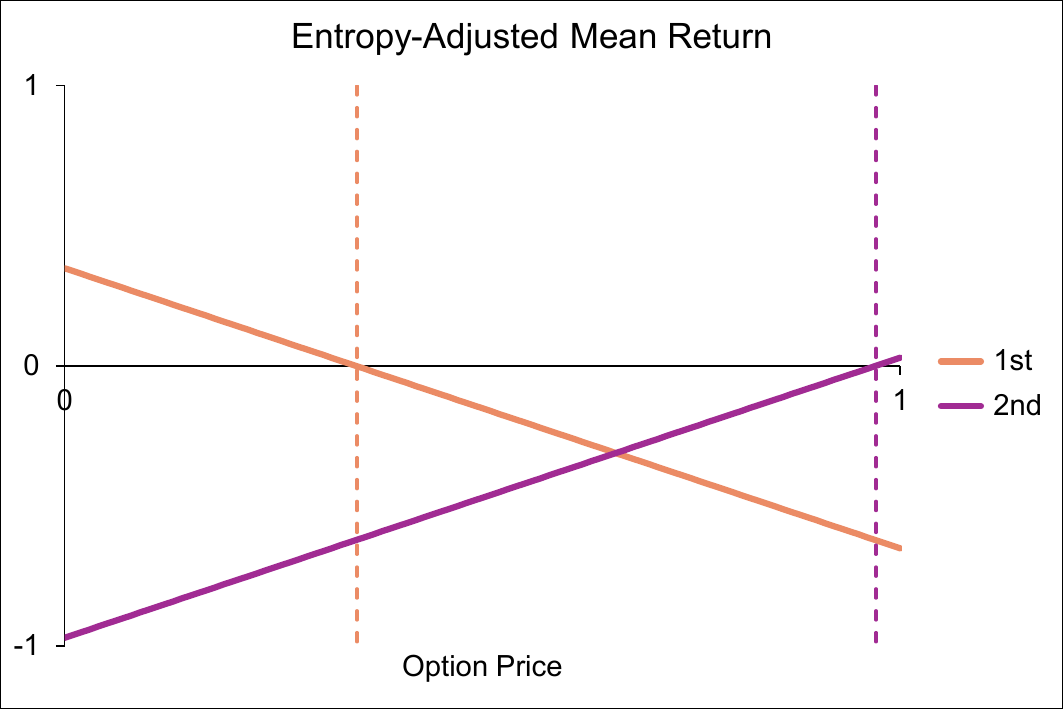}\\
$\lambda=0$ and $c=\infty$ & $\lambda=1$ and $c=\infty$ & $\lambda=2$ and $c=\infty$ \\
\includegraphics[width=0.33\textwidth-\tabcolsep]{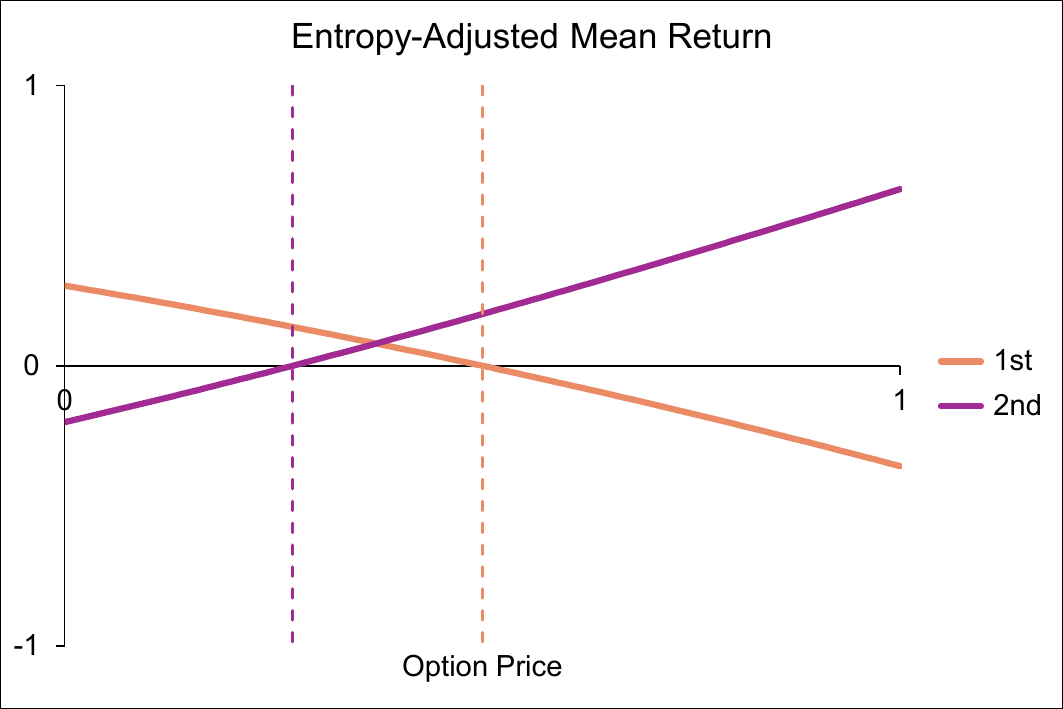}&
\includegraphics[width=0.33\textwidth-\tabcolsep]{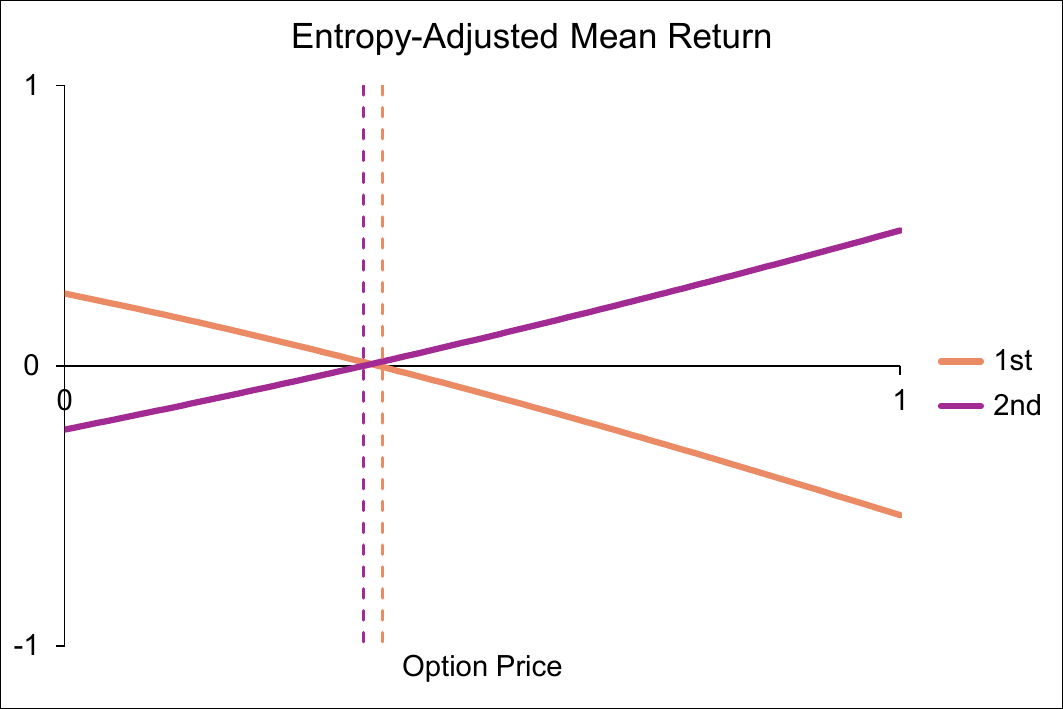}&
\includegraphics[width=0.33\textwidth-\tabcolsep]{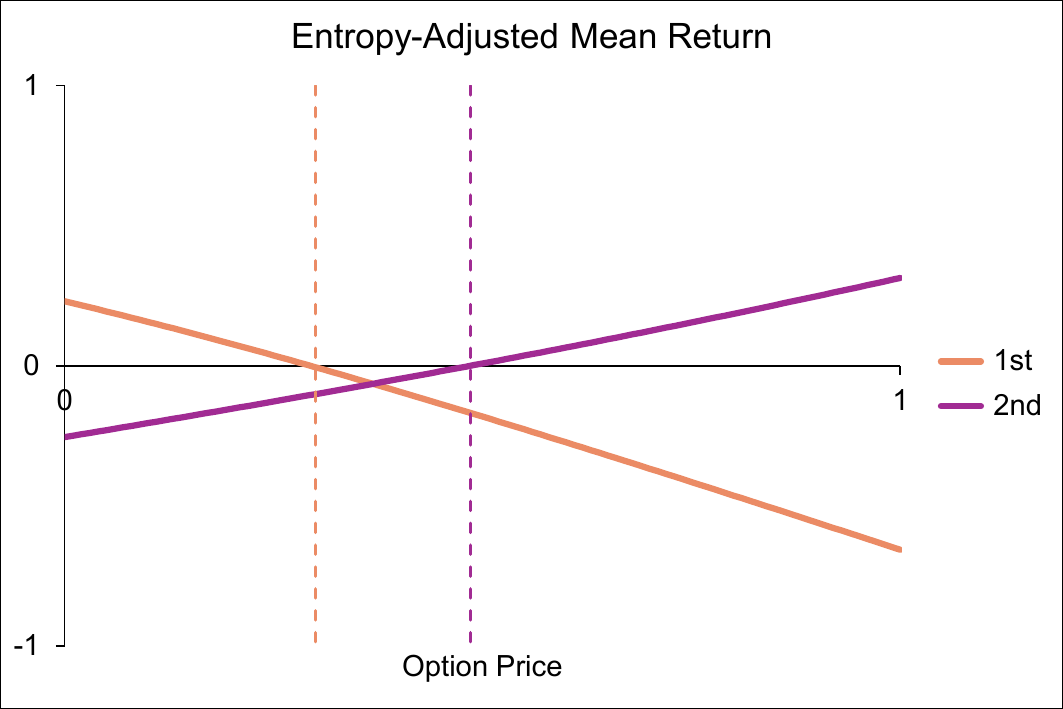}\\
$\lambda=0$ and $c=0.5$ & $\lambda=1$ and $c=0.5$ & $\lambda=2$ and $c=0.5$ \\
\end{tabular}
\caption{In this study, two counterparties have differing expectations for the future state. The first counterparty models the final underlying price as a uniform variable with mean $0.5$ and unit standard deviation, and the second counterparty models the final underlying price as a normal variable with zero mean and unit standard deviation. The trade is then a call option on the final underlying price with zero strike. In each graph, the entropy-adjusted mean return on capital, computed with risk aversion $\alpha=1$, is plotted for both counterparties as a function of the option price per unit notional. The first counterparty is fully capitalised against default, and each row includes cases $c=\infty,0.5$ for the capital of the second counterparty. Each column includes cases $\lambda=0,1,2$ for the notional, where the case $\lambda=0$ corresponds to mid-pricing.
\\\hspace*{\savedindent}The trade is viable when the entropy-adjusted mean return on capital is positive for both counterparties. The first counterparty is axed to buy the option when its price is low, and the second counterparty is axed to sell the option when its price is high. Model risk increases as the notional $\lambda$ is increased, which pushes the threshold prices in opposite directions and makes the trade less viable. Decreasing the capital $c$ increases the risk of default for the second counterparty. The CVA cost to the first counterparty and DVA benefit to the second counterparty then pushes both threshold prices lower.}
\label{fig:entropicxva}
\end{figure*}

Embedded in the expressions for the threshold prices are the contributions from XVA.
\begin{description}[leftmargin=0\parindent]
\item[FVA]Capital for the trade is raised by borrowing in the underlying market and deposited in the margin account. The mismatch between the return on the underlying portfolio and the return on the margin account generates the funding adjustment in the valuation.
\item[CVA]The cap on the potential profit from the funded derivative return generates the credit adjustment in the valuation as a cost to the counterparty.
\item[DVA]The floor on the potential loss from the funded derivative return generates the debit adjustment in the valuation as a benefit to the counterparty.
\end{description}
To understand the XVA contributions, consider the simplified model with predictable funding and no hedging:
\begin{align}
dq_1/q_1&=r_1\,dt \\
dq_2/q_2&=r_2\,dt \notag \\
db/b&=r\,dt \notag
\end{align}
where $r_1$ and $r_2$ are the funding rates of the funding securities and $r$ is the funding rate of the margin account. With only one security in each underlying market, the self-funding conditions derive $\delta_1=c_1/q_1$ and $\delta_2=c_2/q_2$ as well as $\phi_1=\phi_2=0$, and the threshold prices become:
\begin{align}
p_1={}&\frac{1}{1+r_1\,dt}(\Exp_1[\alpha_1][\hat{P}]-m_1(r_1-r)\,dt) \\
={}&\frac{1}{1+r_1\,dt}(\Exp_1[\hat{P}]-m_1(r_1-r)\,dt) \notag \\
&-\frac{\alpha_1}{2}\frac{1}{1+r_1\,dt}\Var_1[\hat{P}]+O[\alpha_1^2] \notag \\
p_2={}&\frac{1}{1+r_2\,dt}(\Exp_2[-\alpha_2][\hat{P}]+m_2(r_2-r)\,dt) \notag \\
={}&\frac{1}{1+r_2\,dt}(\Exp_2[\hat{P}]+m_2(r_2-r)\,dt) \notag \\
&+\frac{\alpha_2}{2}\frac{1}{1+r_2\,dt}\Var_2[\hat{P}]+O[\alpha_2^2] \notag
\end{align}
where $m_1:=c_1-p$ and $m_2:=c_2+p$ are the initial margins for the two counterparties. The final payoff is capped and floored by the default events:
\begin{equation}
\hat{P}:=P_{-m_1(1+r\,dt)}^{m_2(1+r\,dt)}
\end{equation}
which generates the credit and debit valuation adjustments. For each counterparty, the threshold price is the discounted expectation of the capped and floored payoff, adjusted by the funding valuation adjustment and including a reserve against the unhedged risk of the payoff.

\subsection{Entropic margin}

The combined proceeds from the margin accounts are shared between the counterparties, adjusting for the funded return on the derivative security, floored and capped so that the returned capital is positive.

The first counterparty defaults when the floor is breached and the second counterparty defaults when the cap is breached. Each counterparty assesses the threshold for default using their expectation measure, and default is avoided when the margin exceeds the maximum liability:
\begin{align}
m_1&\ge-\essinf\nolimits_1\!\left[\frac{P}{1+db/b}\right] \\
m_2&\ge\esssup\nolimits_2\!\left[\frac{P}{1+db/b}\right] \notag
\end{align}
for the contractual settlement $P$ of the derivative security. Setting these inequalities to equalities generates the minimum capital required to fully guarantee settlement, under the assumptions of the respective expectation measures.

More generally, the counterparties can exchange certainty of settlement for a reduction in capital. In the entropic margin model, this is parametrised by the default aversion $\beta_1$ and $\beta_2$ with:
\begin{align}
m_1&=\left(-\Exp_1[\beta_1]\!\left[\frac{P}{1+db/b}\right]\right)^{\!+} \\
m_2&=\left(\Exp_2[-\beta_2]\!\left[\frac{P}{1+db/b}\right]\right)^{\!+} \notag
\end{align}
where the case $\beta_1=\infty$ means the first counterparty provides full default protection and the case $\beta_2=\infty$ means the second counterparty provides full default protection. Entropy-adjusted mean is utilised in this parametrisation as it spans the range of possible values for the discounted derivative settlement.

Trading on margin is the great innovation of finance, but the protections offered to counterparties are often opaque and operationally cumbersome, and occasionally ineffective. With the derivative settlements net of funding and default compressed into a funding pool with guaranteed settlements, and an entropic margin model that codifies the implied level of default protection, entropic pricing can be automated in a smart contract that transparently declares the trade-off between return and default risk. By packaging funding with contractual settlements, this internalises the management of margin and capital and avoids their overheads and inefficiencies, supporting the decentralised distribution of derivative securities.

\section{Literature review}

Understanding the relationship between the economic and price measures has been the fundamental question of mathematical finance ever since Bachelier first applied stochastic calculus in his pioneering thesis \cite{Bachelier1900}. In this thesis, Bachelier develops a model for the logarithm of the security price as Brownian motion around its equilibrium, deriving expressions for options that would not look unfamiliar today. Methods of portfolio optimisation originated in the work of Markowitz \cite{Markowitz1952} and Sharpe \cite{Sharpe1964} using Gaussian statistics for the returns, generalised by later authors to allow more sophisticated distributional assumptions and measures of utility. These approaches are embedded in the economic measure, explaining the origin of price as the equilibrium of market activity uncovering the expectations of participants.

The impact of dynamic hedging on price was first recognised in the articles by Black and Scholes \cite{Black1973} and Merton \cite{Merton1973}. Adopting the framework devised by Bachelier, these authors observe that the option return is replicated exactly by a strategy that continuously offsets the delta of the option price to the underlying price. The theory matured with the work of Harrison and Pliska \cite{Harrison1981,Harrison1983}, providing a precise statement of the conditions for market completeness and the martingale property of price in a measure equivalent to the economic measure. The discipline has since expanded in numerous directions, with significant advances for term structure and default modelling and numerical methods for complex derivative structures.

Most approaches to derivative pricing begin with the assumption of continuous settlement, and stochastic calculus is an essential ingredient for these developments. The representation of the stochastic differential
equation used in this article follows the discoveries of L\'{e}vy \cite{Levy1934}, Khintchine \cite{Khintchine1937} and It\^{o} \cite{Ito1941}. The technical requirements of the L\'{e}vy-Khintchine representation can be found in these articles and other standard texts in probability theory. The change from economic measure to price measure implied by the entropic principle extends the result from Girsanov \cite{Girsanov1960} to include the scaling adjustment of the jump frequency in addition to the drift adjustment. The example model referenced here is the Heston model \cite{Heston1993}, a popular model in quantitative finance.

Closer consideration of market incompleteness led to the development of alternative pricing principles based on the dynamic optimisation of convex risk metrics, as presented in \cite{Xu2006,Kloppel2007,Ilhan2009}. With foundations that extend well beyond the domain of applicability of risk-neutral pricing, these principles enabled the development of data-driven methods including the pioneering work on deep hedging by Buehler and others \cite{Buehler2020,Buehler2021}. As a recent practical example, Oya \cite{Oya2024} applies these methods to explain the mark-to-market valuation adjustment for Bermudan swaptions. Applications of quantum information in option pricing are investigated further in \cite{McCloud2018}, and the technical requirements to extend this approach to continuous-time models are developed in \cite{McCloud2024}.

Entropy methods have found application across the range of mathematical finance, including portfolio optimisation \cite{Philippatos1972} and derivative pricing \cite{Buchen1996,Gulko1999} -- see also the review essay \cite{Zhou2013} which includes further references. The article \cite{Frittelli2000} by Frittelli proposes the minimal relative entropy measure as a solution to the problem of pricing in incomplete markets, and links this solution with the maximisation of expected exponential utility. The entropic solution is connected with risk optimisation in the introductory essay \cite{McCloud2020}, which this essay extends with the adoption of the entropic risk metric. As a measure of disorder, entropy performs a role similar to variance but is better suited to the real distributions of market returns. Proponents of the use of entropy justify the approach by appeal to the modelling of information flows in dynamical systems and the analogy with thermodynamics.

While market evidence for discounting basis existed earlier, the basis widening that occurred as a result of the Global Financial Crisis of 2007-2008 motivated research into funding and its impact on discounting. Work by Johannes and Sundaresan \cite{Johannes2007}, Fujii and Takahashi \cite{Fujii2011}, Piterbarg and Antonov \cite{Piterbarg2010,Piterbarg2012,Antonov2014}, Henrard \cite{Henrard2014}, McCloud \cite{McCloud2013a} and others established its theoretical justification and practical application.

It is remarkable that the principle of dynamic convex risk optimisation enables the consistent and simultaneous application of models across electronic trading algorithms, for pricing and hedging exotic derivatives, and to determine economic and regulatory capital. In this context, the entropic risk metric is promoted for its applicability within a diverse range of information models and its simple implementation as an extension of the risk-neutral framework, which supports integration into legacy platforms.

\end{document}